\documentclass[12pt]{article}

\usepackage{mkc}
\usepackage{mkc_refs}
\usepackage{stix2}
\usepackage{mathtools}

\allowdisplaybreaks

\raggedright
\title{Eliciting Informed Preferences\footnote{
		We thank Heetak Shah for excellent research assistance. For helpful comments and discussions, we are grateful to Mohammad Akbarpour, Ian Ball, Arjada Bardhi, Alex Bloedel, Lanier Benkard, Ben Brooks, Andrew Caplin, Juan Carrillo, Shuchi Chawla, Mark Dean, Eddie Dekel, Wioletta Dziuda, Jon Eguia, Drew Fudenberg, Samuel Goldberg, Duarte Gonçalves, Jason Hartline, Tetsuya Hoshino, En Hua Hu, Ravi Jagadeesan, Navin Kartik, Andreas Kleiner, Arvind Krishnamurthy, Rohit Lamba, Nicolas Lambert, Evi Micha, Axel Niemeyer, Michael Ostrovsky, Fernando Payr\'{o}, Antonio Penta, Jacobo Perego, Agathe Pernoud, Harry Pei, Phil Reny, Doron Ravid, Joseph Root, Marzena Rostek, Yuliy Sannikov, Grant Schoenebeck, Ilya Segal, Marciano Siniscalchi, Phillip Strack, Bruno Strulovici, Andrzej Skrzypacz, Takuo Sugaya, Jo\~{a}o Thereze, Quitz\'{e} Valenzuela-Stookey, Dong Wei, Weijie Zhong, and three anonymous referees at EC. We also thank audiences at Stanford, the Conference in Honor of Eddie Dekel, the ASU Annual Economic Theory Conference, the Cowles Conference on Economic Theory, DETC, MSU, NYU, USC, VSET, CETC, BRIC, D-TEA, EC, and the Stony Brook International Conference on Game Theory.
}}
\author{Modibo Camara\footnote{Graduate School of Business, Stanford University. Email: mcamara@stanford.edu.}\quad\quad\quad Nicole Immorlica\footnote{Yale University and Microsoft Research. Email: nicimm@gmail.com}\quad\quad\quad Brendan Lucier\footnote{Microsoft Research. Email: brlucier@microsoft.com}}

\begin{document}
	
	\begin{titlepage}
		\maketitle
		
		\begin{abstract}

            In many settings -- like market research and social choice -- people may be presented with unfamiliar options.
            Classical mechanisms may perform poorly because they fail to incentivize people to learn about these options, or worse, encourage counterproductive information acquisition.
            We formalize this problem in a model of robust mechanism design where agents find it costly to learn about their values for a product or policy.
            We identify sharp limits on the designer's ability to elicit, or learn about, these values.
            Where these limits do not bind, we propose two-stage mechanisms that are detail-free and robust: the second stage is a classical mechanism and the first stage asks participants to predict the results of the second stage.
            
			\bigbreak\noindent\textit{Keywords}: Preference elicitation; robust mechanism design; costly information.

		\end{abstract}
		
		\thispagestyle{empty}
		\newpage
		\setcounter{tocdepth}{2}
		\tableofcontents
		\addtocontents{toc}{\protect\thispagestyle{empty}}
		\thispagestyle{empty}
		
	\end{titlepage}
	
	\section{Introduction}

People often find it costly to process information about, for example, the products they consume and the policies that affect them.
As a result, their choices may not directly reveal the preferences they would have had if they had processed all the information available to them.
To what extent can we still learn about a population's preferences through careful incentive design?

Learning about a population's preferences (\emph{after} they have acquired information) is important for many applications.
For example, take market research. 
Suppose that OpenAI is about to launch a product called GPT-5.
In the present, consumers do not know how much value they would get from GPT-5.
They find it costly to learn about and test the product.
In the future, consumers are likely to be better informed.
Can we design an experiment that helps OpenAI forecast future demand?

To address this question, we develop a model of robust mechanism design with information acquisition.
The planner is powerful: he can commit to any mechanism, spend as much money as he likes, and recruit as many agents as he likes.
Agents are sampled from a large population, and do not necessarily know their values for an alternative.
The planner does not know the value distributions 
or how they are correlated across agents.
Agents can learn about their own (and potentially others') values by acquiring costly signals.
We impose little structure on the information acquisition technology, except that agents can learn their own values at finite cost.

We showcase the power of our model with two related applications. First, we study market research.
It can be difficult to elicit willingness to pay for products that consumers are not already familiar with (e.g., \cite{YTJ18}; \cite{CZ21}; \cite{MS23}), which makes it difficult to forecast demand for new products.
We propose a mechanism that presents agents with an incentivized survey and asks them, before completing the survey, to predict the survey results.
If the bonus for accurate predictions is large, this mechanism identifies the population's average value in favorable equilibria, and offers meaningful guarantees even in unfavorable equilibria (Theorem \ref{T2}).
But there are limits to what the planner can do; it is impossible to match these guarantees when estimating statistics like quantity demanded or revenue-maximizing price (Theorem \ref{T1}).

Second, we study social choice.
Empirical evidence suggests many voters are poorly informed (e.g., \cite{DK96}; \cite{AP24}) and would vote differently if they were better informed (e.g., \cite{Bartels96}; \cite{FM14}).
We propose a mechanism that provides better incentives for voters (e.g., in a committee or citizen's assembly) to become informed.
It selects nearly efficient outcomes in favorable equilibria, and provides meaningful guarantees even in unfavorable equilibria (Theorem \ref{T3}).
In a special case, we can motivate a simple mechanism: majority rule paired with subsidized election betting.


\paragraph{Model.}

A planner asks a random sample of $n$ agents to participate in a mechanism.
There is a product and each agent $i$ derives value $v_i\in[v_L,v_H]$ from that product.
Agents may be heterogeneous and their values may be correlated in complicated ways.
The planner does not know the distribution of values in the population.

Initially, agents do not know their own values.
Instead, each agent $i$ has some private information and may learn more -- both about her own value $v_i$ and perhaps about others' values $v_{-i}$ -- by acquiring costly signals.
We assume that each agent can learn her value $v_i$ at some finite cost.
Agents may also have access to additional, more complex signals that are unknown to the planner (as in e.g., \cite{Carroll19}).

The planner designs a mechanism that assigns an allocation $x_i$ and transfers $t_i$ to each agent.
He seeks mechanisms that perform well (according to objectives that we describe later), for large sample sizes $n\to\infty$, and for all distributions of values and information acquisition technologies that satisfy our assumptions.
We consider both best-case equilibria, as well as worst-case equilibria (as in e.g., \cite{Maskin99}).

The planner is empowered: he can recruit as many agents and spend as much money as he likes.
In particular, the planner can spend as much money per agent as he likes.
Although this is a strong assumption, it strengthens our negative result (Theorem \ref{T1}), and makes it possible to obtain clean positive results (Theorem \ref{T2} and \ref{T3}).

\paragraph{Market Research.}

As a first application, suppose the planner wants to forecast long-term demand for a new product.
For many products (like the hypothetical GPT-5), consumers may be initially unfamiliar with the product but then become familiar with it over time.
In that case, long-term demand is governed by the values $v_i$.

We find positive and negative results.
Our positive result concerns a planner who wants to estimate the value $v_i$ of an out-of-sample consumer $i$, subject to square loss.
The planner's estimate is \emph{ex-post optimal} if it matches his estimate in a first-best scenario, where the participants all knew their values $v_i$ and reported them truthfully.
His estimate is at least \emph{ex-ante optimal} if it is not worse than his estimate in a scenario where he knows the agents' common prior, but cannot run a mechanism.

Theorem \ref{T2} presents a sequence of mechanisms (indexed by the sample size $n$) that guarantees ex-post optimal estimates in favorable equilibria, and at least ex-ante optimal estimates in all equilibria.
By comparison, the Becker–DeGroot–Marschak (BDM) mechanism (a classical method for eliciting willingness to pay) can yield estimates that are not even ex-ante optimal in favorable equilibria.

These mechanisms are a relatively-straightforward modification of BDM, which we call \emph{BDM-with-betting}.
There are two stages.
In the second stage, agents report their willingness to pay in the BDM mechanism.
In the first stage, agents predict the average reported willingness to pay.
Accurate predictions are rewarded according to a proper scoring rule, which can be scaled up arbitrarily to strengthen incentives for information acquisition.
The planner's estimate is the average reported value.

The BDM-with-betting mechanism is naive in two senses.
First, unlike existing mechanisms with similar features (e.g., \cite{CM88}; \cite{MRZ05}), the planner does not use agent $i$'s prediction to infer her value $v_i$, or to incentivize truthful reporting.
Second, the mechanism incentivizes each agent $i$ to learn about \emph{others'} reported values, when ideally she would learn \emph{her own} value $v_i$.

Nonetheless, BDM-with-betting works because it incentivizes agents to acquire information up to the point where their reporting errors are uncorrelated.
If these errors -- the difference between an agent's value and her reported value -- are uncorrelated, the average reported value converges to the population's average value (i.e., the ex-post optimal estimate).
Essentially, betting restores the ``wisdom of the crowd'', which usually requires strong distributional assumptions.

Unfortunately, there are also limits to what the planner can accomplish.
For loss functions that are meaningfully different from square loss, Theorem \ref{T1} says that the planner cannot guarantee estimates that are at least ex-ante optimal in all equilibria, regardless of which mechanism he uses.
Together, Theorems \ref{T2}-\ref{T1} suggest that it may be easier to forecast the population's average value than other natural statistics, like the quantity demanded at a fixed price, or the revenue-maximizing price.

\paragraph{Social Choice.}

Next, suppose the planner chooses one of two options on behalf of a population.
He forms a committee (e.g., a citizen's assembly) of $n$ agents sampled from the population.
As before, he commits to a mechanism with transfers.

We are motivated by the problem of uninformed voting.
Existing electoral systems do not give voters much of a reason to become informed about the policies and candidates that appear on their ballot.
This is not just a theoretical concern.
Empirically, many voters are poorly informed (e.g., \cite{DK96}, \cite{AP24}), and making them better-informed could plausibly affect vote margins (e.g., \cite{LR97}, \cite{FM14}).

It turns out that the methods we developed to elicit preferences can be used to design social choice mechanisms that are robust to costly information processing.
Theorem \ref{T3} finds mechanisms that (i) choose the efficient outcome in favorable equilibria and, for any equilibrium, (ii) do no worse than simply choosing the ex-ante optimal alternative.
By comparison, standard voting procedures do not even guarantee (ii) in favorable equilibria (see e.g., \cite{NMS25}).

Under strong simplifying assumptions, we motivate the \emph{majority-rule-with-betting} mechanism, which does not involve vote buying.
There are two stages.
In the second stage, agents report their preferred alternative.
In the first stage, they predict the vote margin.
The planner chooses the alternative with majority support.\footnote{More precisely, the planner chooses this alternative with high probability. See Section \ref{S5-3} for details.}
Generous transfers are used to reward accurate predictions, but not to buy or sell votes.

There are two key reasons why this mechanism works.
The first, like before, is that betting restores the ``wisdom of the crowd''.
The second reason is that betting can promote truthful voting.
Recall that uninformed voters may not vote sincerely if they believe that how others vote may be informative about their own values (e.g., \cite{AB96}).
In our mechanism, an uninformed agent $i$ that believes her value $v_i$ is correlated with the vote margin is leaving money on the table.
She prefers to learn her value $v_i$ in order to better predict the vote margin.

\subsection{Related Literature}\label{S1-1}


\paragraph{Preference Elicitation.}

We contribute to research on preference elicitation (e.g., \cite{BDM64}) by relaxing the assumption that individuals know their own preferences.
We are motivated by prior work in marketing (e.g., \cite{YTJ18}; \cite{CZ21}). In particular, \textcite{MS23} finds that the BDM mechanism may not be effective when information is costly.
We propose a modification to BDM that may strengthen incentives for information acquisition.

Our results also complement prior work that studies the implications of limited or costly attention for revealed preferences (e.g., \cite{MNO12}; \cite{MM14}; \cite{CMMS20}) and demand estimation (e.g., \cite{BCMT21}).
We explore the extent to which costly attention is a fundamental barrier to learning about preferences.

\paragraph{Social Choice.}

The literature on voter information establishes several stylized facts that motivate our work.
First, many voters are poorly informed, and better-informed voters may not be representative of the electorate (e.g., \cite{PP87}; \cite{DK96}; \cite{AP24}).
Second, giving publicly-available information to people can affect how they vote (e.g., \cite{LR97}; \cite{FM14}).
Third, well-informed voters vote differently than poorly-informed ones, after controlling for observables (e.g., \cite{Bartels96}; \cite{Gilens01}).
Fourth, voters respond to incentives to acquire information (e.g., \cite{Larcinese09}; \cite{Shineman18}; \cite{PL08}).
Fifth, more informed electorates are better at removing low-quality officials (e.g., \cite{Pande11}).

We build on a line of work on combating uninformed voting (e.g., \cite{Persico04}; \cite{GY08}; \cite{GS09}; \cite{Cai09}; \cite{Tyson16}).
Relative to this work, we show that a more powerful planner (e.g., with access to transfers) can achieve nearly first-best outcomes with limited prior knowledge.

We also build on studies of voting behavior when information is costly (e.g., \cite{Gersbach95}; \cite{Martinelli06}; \cite{FS06}; \cite{Vaeth25}; \cite{HG25}) or imperfect (e.g., \cite{FP96}; \cite{KF07}; \cite{Bhattacharya13}; \cite{NMS25}).
\textcite{BMP10} provide empirical evidence that poorly-informed voters may abstain or appear to vote against their own interests.
In theory, the mechanisms we propose alleviate these issues. 

\paragraph{Peer Prediction.}

Peer prediction mechanisms elicit private information, typically by rewarding agents for making reports that are correlated with the reports of other agents (e.g., \cite{MRZ05}; \cite{DG13}; \cite{SAFP16}).

This literature is concerned with a problem that is harder than ours: how to design a transfer rule in order to elicit information.
In contrast, we design a transfer and allocation rule in order to elicit information.
This is a crucial difference, because it allows us to build on mechanisms like BDM that elicit willingness to pay by offering a product for sale.
If we restricted ourselves to mechanisms that only rely on transfers, we would run into existing impossibility results (e.g., \cite{ZC14}).\footnote{Similarly, the Bayesian truth serum (e.g., \cite{Prelec04}) tries to elicit information using only transfers. Although it incentivizes truthful reporting, it does not necessarily incentivize information acquisition.}

That raises the question: why can we not build on existing peer prediction prediction mechanisms to obtain our results?
First, we allow for rather general forms of information acquisition, which can undermine peer prediction mechanisms \parencite{GWL19}.
Second, many of our results refer to guarantees that hold in all equilibria.
To our knowledge, all peer prediction mechanisms have non-truthful equilibria, even after imposing natural refinements (e.g., \cite{DG13}).

Finally, two innovative papers rely on techniques used in peer prediction to design mechanisms that are robust to the possibility that agents do not know their own preferences (\cite{ST21}; \cite{Pakzad-Hurson22}).
Interestingly, these papers do not need to assume that agents can learn their payoff types in order to achieve efficiency.
We cannot use mechanisms from these papers to obtain our results -- primarily because they rely on particular information structures and secondarily because they feature multiple equilibria -- but very much share their goals.

\paragraph{Mechanism Design with Information Acquisition.}
Much of this literature focuses on different applications, like auctions (e.g., \cite{Persico00}; \cite{BSV09}), matching (e.g., \cite{ILLL20}), and settings with verifiable information (e.g., \cite{BDL24}).
We are especially motivated by two negative results.
First, \textcite{BV02} find that efficient mechanisms may not exist when information is costly and preferences are correlated.
Second, \textcite{PS24} find that planners with sufficiently-imprecise knowledge of preferences and the costs of information can only implement trivial social choice rules.

In particular, we draw from prior work that insists on robustness to unknown production or information acquisition technologies (e.g., \cite{Carroll15}; \cite{Carroll19}; \cite{DR24}).
We focus on a particular known technology: the revealing signal that allows an agent to learn her own value, but not necessarily anything else.

Finally, some prior work highlights how incentives to learn about others can undermine efficiency (e.g., \cite{GP24}).
In contrast, we obtain efficient social choice only after explicitly incentivizing agents to learn about each other.

\paragraph{Full Surplus Extraction.}

This literature on full surplus extraction asks when it is possible to incentivize agents to reveal their private types with no expected cost to the designer (e.g., \cite{CM88}).
Aside from the fact that our planner disregards costs, our work differs in three ways.
First, we do not assume that agents know their own payoff types.\footnote{\textcite{Bikhchandani10} considers full surplus extraction with information acquisition, but still assumes that agents know their own payoff types.}
Second, our mechanisms are robust (we do not need strong distributional assumptions) and detail-free (the planner does not need to know the distribution).\footnote{\textcite{FHHK21} allow for the possibility that the distribution is unknown, but still assume that there is a known finite set of possible distributions.}
Third, our mechanisms are simpler, in that we do not use agents' reported beliefs to make any inferences about their payoff types.

\paragraph{Proper Scoring Rules.}

This literature studies belief elicitation when the planner can make transfers that depend on the realized state.
In that case, there are transfers schemes that can incentivize agents to truthfully report their beliefs (e.g., \cite{McCarthy56}) and acquire information (e.g., \cite{LHSW22}).
We do not assume that the planner can make state-dependent transfers.
Nonetheless, proper scoring rules play an important role in our positive results; we use them to incentivize agents to acquire information that predicts summary statistics of other agents' reports.

We draw inspiration from \textcite{LPS08}.
Assuming that the planner can make transfers that depend on the state, they ask what kinds of statistics can be elicited by direct mechanisms.
We also find that certain statistics are easier to estimate than others, albeit in a different setting and allowing for indirect mechanisms.

	\section{Model}\label{S2}

A planner recruits $n$ agents, and can allocate a product to each agent.
Agent $i$'s allocation is $x_i\in\mathcal{X}=\{0,1\}$ and her value from the product is $v_i\in\mathcal{V}=[v_L,v_H]$.

As a running example, suppose the firm OpenAI develops a new product called GPT-5.
The outside option is an existing product called GPT-4, and each consumer $i$'s value $v_i$ captures how much she would benefit from upgrading GPT-4 to GPT-5.
Before GPT-5 is launched, the consumer does not know her value $v_i$, because she does not know enough about the new product's functionalities, use cases, flaws, etc.

	\subsection{Information Acquisition}\label{S2-1}

Agent $i$ has a \emph{type }$\theta_i\in\Theta$ that determines her value via the \emph{value function} $v:\Theta\to\mathcal{V}$.
She is initially uncertain about her type, but can learn by acquiring a \emph{signal }$s_i\in\mathcal{S}$.
Her cost of acquiring the signal given by the \emph{cost function}
$
c:\mathcal{S}\times\Theta\to\reals_+\cup\{\infty\}
$.

Types are drawn from a \emph{type distribution }$F_{\Theta\mid\omega}$ that depends on a hidden state $\omega\in\Omega$.
For example, the state might indicate how useful GPT-5 is for a given task.
The state is drawn from a \emph{state distribution} $F_{\Omega}$.
Let $F\in\Delta\left(\Omega\times\Theta\right)$ be the joint distribution of types and states.
Types $\theta_1,\ldots,\theta_n$ are i.i.d. conditional on the state $\omega$.
Essentially, the state acts as a correlating device.

We impose the following finiteness assumptions.

\begin{assume}\label{A7}\label{A2}
	The signal space $\mathcal{S}$ and the support of joint distribution $F$ are finite.
\end{assume}

A signal $s_i$ for agent $i$ maps the state $\omega$ and her type $\theta_i$ to a \emph{signal realization}.
It is convenient to let signal realizations be finite sequences of real numbers, so that
$s_i:\Omega\times\Theta\to\reals^*.$
One interpretation of the signal $s_i$ is as a deterministic function of random variables $(\omega,\theta_i)$.
Another interpretation is that the signal $s_i$ is itself a random variable, which maps the sample space $\Omega\times\Theta$ to the set $\reals^*$.

After observing 
signal realization $s_i(\omega,\theta_i)$, agent $i$ may learn something about the state $\omega$ and her type $\theta_i$.
For example, consider 
signal $s_i(\omega,\theta_i):=v(\theta_i)$.
After acquiring this signal, agent $i$ learns 
her value $v_i=v(\theta_i)$.
If types are correlated with the state 
then she may also update her beliefs about the state 
and the other agents' types.

In the running example, a consumer could learn about GPT-5 by using it.
She might think of different tasks that matter to her (e.g., writing, coding, translation, etc.) and evaluate how well it performs in each case.
Each task corresponds to a different signal, and GPT-5's performance at the task corresponds to a signal realization.
Learning about GPT-5's performance could cause her to learn directly about her own value $v_i$, and indirectly about other consumers' values $v_{-i}$.

This model of information acquisition may appear restrictive.
For example, each agent $i$ acquires exactly one signal, and initially knows nothing about her type $\theta_i$.
However, we can accommodate richer forms of information acquisition as follows.
\begin{enumerate}
	\item \emph{Acquiring multiple signals.}
	For any two signals $s_i,s'_i\in\mathcal{S}$, we can enrich the signal space by adding a combined signal, i.e.,
	\begin{equation}\label{E4}
		s''_i(\omega,\theta_i)=\left(s_i(\omega,\theta_i),s'_i(\omega,\theta_i)\right)
	\end{equation}
	
	\item \emph{Dynamic information acquisition.}
	Suppose that agent $i$ wants to acquire signal $s'_i$ if and only if the realization of signal $s_i$ is in a set $U\subseteq s_i(\Omega,\Theta)$.
	We can enrich the signal space by adding a dynamic signal, i.e.,
	\begin{equation}\label{E7}
		s''_i(\omega,\theta_i)=\begin{cases}
			\left(s_i(\omega,\theta_i),s'_i(\omega,\theta_i)\right)&s_i(\omega,\theta_i)\in U\\
			s_i(\omega,\theta_i)&\text{otherwise}
		\end{cases}
	\end{equation}
	Similarly, we can create dynamic signals that combine more than two signals.\footnote{When adding dynamic signals, we need to be careful to avoid violating Assumption \ref{A2}. This is not a problem when the state space $\Omega$ and type space $\Theta$ are finite. In that case, the number of subsets $U\subseteq s_i(\Omega,\Theta)$ will be finite. Therefore, the number of dynamic signals that combine two signals -- which correspond to different combinations of $(s_i,s'_i,U)$ -- will also be finite. We can extend this same argument when combining three or more signals. Even if the state and type spaces are infinite, we can use the finite support condition in Assumption \ref{A7} to restrict attention to finite subsets.}
	\item \emph{Initial information.}
	To represent agents' initial information, define a signal $s_i$ with cost $c(s_i,\theta_i)=0$ for all types $\theta_i$.
	Expand the signal space $\mathcal{S}$ by allowing all dynamic signals that combine $s_i$ with another signal $s'_i\in\mathcal{S}$. 
\end{enumerate}

Finally, the \emph{model instance} $I=\left(\Omega,\mathcal{V},\Theta,\mathcal{S},F,v,c\right)\in\mathcal{I}$ collects all model primitives.
We assume that the agents know the instance $I$, but the planner does not.

	\subsection{Assumptions}\label{S2-4}

We maintain two key assumptions.
First, we assume that acquiring one signal does not prevent an agent from acquiring another one.
For example, evaluating GPT-5 on a coding task need not prevent a consumer from also evaluating GPT-5 on a translation task.
This may be violated if the consumer only has time to prepare one task.

\begin{assume}\label{A5}
	The set of signals $\mathcal{S}$ consists of base signals and combined signals.
	For every set of base signals $\left\{s^1_i,\ldots,s^\ell_i\right\}\subseteq\mathcal{S}$, there is a combined signal $s_i\in\mathcal{S}$ where
	\begin{equation}\label{E2}
		\forall(\omega,\theta_i)\in\Omega\times\Theta:\quad s_i\left(\omega,\theta_i\right)=\left(s^1_i(\omega,\theta_i),\ldots,s^\ell_i(\omega,\theta_i)\right)
	\end{equation}
	If these base signals have a finite cost, then the combined signal also has a finite cost, i.e.,
	\begin{equation}\label{E3}
		\forall \theta_i\in\Theta:\quad
		\sum_{j=1}^{\ell}c\left(s^j_i,\theta_i\right)<\infty\implies c\left(s_i,\theta_i\right)<\infty
	\end{equation}
\end{assume}

Second, we assume that agent $i$ can learn her own value $v_i$ at a finite cost.
Although this is a strong assumption, it is much weaker than the standard assumption that agents know their own payoff types.
It can also be weakened somewhat.\footnote{What we need is for each agent to be the ``authority'' on her own value. That is, there should \emph{not }be information that is relevant to agent $i$'s value $v_i$, not available to agent $i$, and yet available to agent $j\neq i$.}

\begin{defn}\label{D7}
	A signal $s_i$ is \emph{revealing} if agent $i$ learns her value $v_i$ after acquiring it, i.e.,
	\begin{equation}\label{E5}
		\var{v_i\mid s_i(\omega,\theta_i)}=0
	\end{equation}
\end{defn}

A revealing signal for agent $i$ does not necessarily reveal everything.
For example, it does not necessarily reveal the state $\omega$, or the values $v_{-i}$ of other agents.
For example, a consumer that dedicates an entire week to testing GPT-5 may get a pretty good sense of her value from that product, even if she does not necessarily learn how much value other consumers will derive from that product.

\begin{assume}\label{A6}
	The set of signals $\mathcal{S}$ includes a revealing signal $s_i$ that has finite cost.
	That is, for all types $\theta_i$, $c\left(s_i,\theta_i\right)<\infty$.
\end{assume}

That is, the planner knows that the agent can learn her value $v_i$ by acquiring the revealing signal $s_i$.
Of course, he might be concerned that the agent ends up choosing another signal $s'_i\in\mathcal{S}$, where she does not learn her value $v_i$.
Here, we follow prior work that assumes there are known actions available to agents but seeks robustness to other unknown actions (e.g., \cite{Carroll15}; \cite{Carroll19}).

	\subsection{Mechanism Design}\label{S2-2}

The planner commits to a mechanism that determines allocations $x_i$, transfers $t_i$, and an estimate $e$.
Although the estimate $e$ is not payoff-relevant for the agents, it is convenient to include it as part of the mechanism.

\begin{defn}\label{D2}
	Fix a set of message profiles $M=\left(M_i\right)_{i=1}^n$ and a set of possible estimates $\mathcal{E}$.
	A \emph{mechanism }$(\textbf{x},\textbf{t},\textbf{e})$ consists of allocation rules $\textbf{x}_i:M\to\Delta\left(\mathcal{X}\right)$, transfer rules $\textbf{t}_i:M\to\Delta\left(\reals\right)$, and an estimator $\textbf{e}:M\to\Delta(\mathcal{E})$.
\end{defn}

After the mechanism is announced, the game proceeds in three steps.
First, each agent $i$ chooses a signal $s_i$ to acquire.
Second, agent $i$ observes realization $s_i(\omega,\theta_i)$.
Third, agent $i$ sends a message $m_i$.
Let $m=(m_1,\ldots,m_n)$ be the \emph{message profile}.

\begin{defn}
	Agent $i$'s \emph{strategy} $(\textbf{s}_i,\textbf{m}_i)$ pairs a \emph{signal rule} $\textbf{s}_i\in\Delta\left(\mathcal{S}\right)$ with a \emph{message rule} $\textbf{m}_i:\mathcal{S}\times\reals^*\to\Delta\left(M_i\right)$ that maps a signal and realization  to a message distribution.

\end{defn}


Agent $i$'s utility depends on her values $v_i$, her allocation $x_i$, the transfers $t_i$ she receives, and her cost $c(s_i,\theta_i)$ of information acquisition.
That is,
\[
u_i\left(v_i,x_i,t_i,s_i,\theta_i\right)=v_ix_i+t_i-c\left(s_i,\theta_i\right)
\]
In turn, agent $i$'s expected utility from a strategy profile $(\textbf{s},\textbf{m})$ is
\[
U_i\left(\textbf{s},\textbf{m},\textbf{x},\textbf{t}\right)
	=\ex{
		u_i(v_i,x_i(m),t_i(m),s_i,\theta_i)
}
\]
\[\text{where}\:
s_i\sim\textbf{s}_i,
\: m_i\sim\textbf{m}_i(s_i,s_i(\omega,\theta_i)),
\: x_i\sim\textbf{x}_i(m),
\: t_i\sim\textbf{t}_i(m)
\]
The expectation is taken with respect to the joint distribution $F$, any randomness in the strategy profile $(\textbf{s},\textbf{m})$, and any randomness in the mechanism $(\textbf{x},\textbf{t})$.

\begin{defn}
	A strategy profile $(\textbf{s},\textbf{m})$ is a \emph{Bayes-Nash equilibrium} of mechanism $(\textbf{x},\textbf{t})$ if every agent $i$ prefers her strategy $(\textbf{s}_i,\textbf{m}_i)$ to every alternative strategy $\left(\textbf{s}'_i,\textbf{m}'_i\right)$.
	That is,
	\[
	U_i\left(\textbf{s},\textbf{m},\textbf{x},\textbf{t}\right)\geq U_i\left(\left(\textbf{s}'_i,\textbf{s}_{-i}\right),\left(\textbf{m}'_i,\textbf{m}_{-i}\right),\textbf{x},\textbf{t}\right)
	\]
\end{defn}

	\subsection{Planner's Problem}\label{S2-3}

The planner wants to minimize the expected loss of his estimate $e$.
The planner's \emph{loss function} $L:\mathcal{E}\times\Delta(\mathcal{V})\to\reals$ compares his estimate $e$ to the distribution of the values $v_i$.

The values $v_i$ are drawn i.i.d. from the conditional marginal distribution $F_{\mathcal{V}\mid\omega}$, which we obtain by pushing forward the type distribution $F_{\Theta\mid\omega}$ through the value function $v$.
We refer to $F_{\mathcal{V}\mid\omega}$ as the \emph{realized demand curve}.
The realized demand curve depends on the realized state $\omega$, which captures any relevant demand shifters.

The planner's expected loss from mechanism $(\textbf{x},\textbf{t},\textbf{e})$ in equilibrium $(\textbf{s},\textbf{m})$ is
\[
\ex{L\left(e,F_{\mathcal{V}\mid\omega}\right)}\quad\text{where}\:
s_i\sim\textbf{s}_i,
\: m_i\sim\textbf{m}_i(s_i,s_i(\omega,\theta_i)),\:
e\sim\textbf{e}(m)
\]

Our definition of loss functions is slightly more general than usual.
It captures a number of problems that require the planner to learn about the distribution of values in a population (Section \ref{S4}).
We can also adapt it to study social choice (Section \ref{S5}).

We compare the planner's expected loss to two performance benchmarks.
The \emph{ex-post benchmark} asks what the planner's expected loss would be if he knew both the instance $I$ and the state $\omega$.
This is the first-best scenario; the planner cannot do better.\footnote{To be clear, the ``ex post'' stage occurs after the realization of the state $\omega$. If the planner knows the instance $I$ and the realized state $\omega$, then he knows the realized demand curve $F_{\mathcal{V}\mid\omega}$.}

\begin{defn}
	The \emph{ex-post benchmark} $\overline{B}:\mathcal{I}\to\reals$ is
	\[
	\overline{B}(I)=\ex[I]{\min_{e\in\mathcal{E}}\ex[I]{L\left(e,F_{\mathcal{V}\mid\omega}\right)\mid\omega}}
	\]
\end{defn}

The \emph{ex-ante benchmark} asks what the planner's expected loss would be if he knew the instance $I$, but nothing else.
This is the best that the planner could hope to achieve in a scenario where agents do not acquire any signals.
We treat the ex-ante benchmark as a low bar that any compelling mechanism should be able to clear.

\begin{defn}
	The \emph{ex-ante benchmark} $\underline{B}:\mathcal{I}\to\reals$ is
	\[
	\underline{B}(I)=\min_{e\in\mathcal{E}}\ex[I]{L\left(e,F_{\mathcal{V}\mid\omega}\right)}
	\]
\end{defn}

A sequence of mechanisms $(\textbf{x}^n,\textbf{t}^n,\textbf{e}^n)$, indexed by the sample size $n$, \emph{guarantees }a benchmark if the planner's expected loss is less than the benchmark in the limit.

\begin{defn}\label{D15}
	A sequence of mechanisms $(\textbf{x}^n,\textbf{t}^n,\textbf{e}^n)$ \emph{guarantees }benchmark $B:\mathcal{I}\to\reals$ in \emph{favorable equilibria} if, for all $I\in\mathcal{I}$, there is a sequence of equilibria $(\textbf{s}^n,\textbf{m}^n)$ such that
	\begin{equation}\label{E70}
	\limsup_{n\to\infty}\ex[I]{L\left(e,F_{\mathcal{V}\mid\omega}\right)}\leq B(I)
	\:\text{where}\:
s_i\sim\textbf{s}_i,
\: m_i\sim\textbf{m}_i(s_i,s_i(\omega,\theta_i)),
\:e\sim\textbf{e}(m)
	\end{equation}
	It guarantees benchmark $B$ in \emph{all equilibria} if inequality \eqref{E70} holds for all instances $I\in\mathcal{I}$ and all sequences of equilibria $(\textbf{s}^n,\textbf{m}^n)$.
\end{defn}

This guarantee is demanding in the sense that it must hold over all instances $I$ that satisfy our assumptions.
In particular, we insist on robustness to a wide range of joint distributions $F$ and information acquisition technologies $\left(\mathcal{S},c\right)$.
Our main results hold even if we restrict attention to a much smaller set of instances.
Nonetheless, we insist on robustness to a large set of instances because that is what distinguishes the mechanisms that we propose.
It is often easy to find mechanisms that work well for a given instance $I$, but they tend to be fragile and hard to generalize.

Finally, we emphasize two key limitations of this model.
First, we do not impose a budget constraint; we also do not impose upper bounds on transfers to individual agents.
Second, we assume that the planner can recruit an arbitrarily-large sample of agents, and ignore finite-sample performance.
These modeling choices strengthen our negative results, and allow us to obtain clean positive results.	
However, taking budget and sampling constraints seriously is a natural direction for future work.

	\section{Market Research}\label{S4}

Consider a planner who wants to learn about long-term demand for a new product.
For many products (e.g., experience goods, new technologies, subscriptions), consumers that are unfamiliar with the product before launch may become familiar with it over time.
In that case, long-term demand is governed by the values $v_i$.

More concretely, take our running example.
Before product launch, consumers may not know their values $v_i$ for GPT-5.
However, they are likely to learn more in the years after the product launch.
Subscribers gain first-hand experience with GPT-5, and will not resubscribe if their value $v_i$ is too low.
Even non-subscribers are likely to become better-informed, as early adopters and critics share their experiences.

We ask whether the planner can conduct an experiment---before the product launch---that provides information about long-term demand.%
\footnote{This information is useful because firms must make a number of strategic decisions before product launch (related to e.g., marketing and inventory management) that can have long-term implications. Similarly, a fund deciding whether to invest in a firm cares about long-term demand for its products.}
To begin, consider a planner seeking a point estimate $e\in\mathcal{V}$ of the value $v_i$ of a consumer $i\notin\{1,\ldots,n\}$ that did not participate in the experiment.
He wants to minimize \emph{square loss}, i.e.,
\begin{equation*}\label{E71}
	L\left(e,F_{\mathcal{V}\mid\omega}\right)=\ex[F_{\mathcal{V}\mid\omega}]{(e-v_i)^2}
\end{equation*}

\begin{theorem}\label{T2}
	For square loss, there exists a sequence of mechanisms that:
	\begin{enumerate}
		\item Guarantees the ex-ante benchmark in all equilibria.
		\item Guarantees the ex-post benchmark in favorable equilibria.
	\end{enumerate}
\end{theorem}

The mechanism in Theorem \ref{T2} identifies the population's average value $\ex[F_{\mathcal{V}\mid\omega}]{v_i}$ in at least one equilibrium, and that statistic minimizes expected square loss.
In some instances $I$, it identifies this statistic in every equilibrium (see Remark \ref{R1}).
In others, there are ``bad equilibria'' where the mechanism fails to identify this statistic.
Indeed, no mechanism can rule out bad equilibria in all instances.

Theorem \ref{T2} also ensures that this mechanism is at least somewhat robust to the possibility that agents coordinate on unfavorable equilibria.
In every equilibrium, the planner's estimate will be at least as good as the best-possible estimate in a scenario where agents did not acquire any signals.
This may seem like a low bar, but we will see in Section \ref{S3-2} that classical mechanisms fail to meet this standard.

What if the planner is interested in loss functions other than square loss, or seeks estimates other than point estimates?
Unfortunately, this is more difficult.
It is not possible to extend Theorem \ref{T2} to loss functions that are \emph{not square-like}.

\begin{defn}\label{D8}
	A loss function $L$ is \emph{not square-like} if there exists a pair of constants $a,b\in\mathcal{V}$ where, for distributions $G=\textsc{Uniform}\{a,b\}$ and $G'=\textsc{Dirac}(\nicefrac{a+b}{2})$,
	\[
	\arg\min_{e\in\mathcal{E}}\ex{L\left(e,G\right)}\cap\arg\min_{e\in\mathcal{E}}\ex{L\left(e,G'\right)}=\emptyset
	\]
\end{defn}

Intuitively, a loss function is not square-like if the statistic that minimizes square loss is sensitive to the dispersion of the distribution $F_{\mathcal{V}\mid\omega}$.
Here are two examples.
\begin{enumerate}
	\item \emph{Quantity demanded.} Given a fixed price $p\in\reals_+$, the planner wants to estimate the probability that a consumer's value exceeds $p$, where
	\[
	L\left(e,F_{\mathcal{V}\mid\omega}\right)=\ex[F_{\mathcal{V}\mid\omega}]{(e-\textbf{1}(v_i\geq p))^2}
	\]
	\item \emph{Revenue-maximizing price.} The planner interprets estimate $e\in\reals_+$ as the price. He wants to estimate the price $e$ that maximizes revenue, where
	\[
	L\left(e,F_{\mathcal{V}\mid\omega}\right)=-\ex[F_{\mathcal{V}\mid\omega}]{e\cdot\textbf{1}(v_i\geq e)}
	\]
\end{enumerate}

\begin{theorem}\label{T1}
	Suppose the loss function $L$ is not square-like.
	Then there does not exist a sequence of mechanisms that guarantees the ex-ante benchmark in all equilibria.
\end{theorem}

Essentially, for many loss functions, any given mechanism either fails to minimize expected loss, or runs into severe equilibrium multiplicity (see Section \ref{S3-1}).
Together, Theorems \ref{T2}-\ref{T1} suggest that estimating statistics like quantity demanded and revenue-maximizing price is harder than estimating average willingness to pay.

The rest of this section provides intuition.
First, we describe the problems that existing methods for willingness-to-pay elicitation run into.
Next, we propose a new mechanism.
Finally, we provide intuition for Theorems \ref{T2}-\ref{T1} through examples.

	\subsection{BDM Mechanism}\label{S3-2}

We begin by analyzing the \emph{Becker–DeGroot–Marschak (BDM) mechanism}, which is widely-used in behavioral experiments to elicit willingness to pay.

\begin{defn}
	The \emph{BDM mechanism} $\left(\textbf{x},\textbf{t},\textbf{e}\right)$ has each agent $i$ report her value $\hat{v}_i\in\mathcal{V}$.
	She receives the product if her reported value exceeds a randomly-drawn price $p$, i.e.,
	\[
	\textbf{x}_i(\hat{v})=\textbf{\emph{1}}(\hat{v}_i\geq p)\quad\text{where}\quad p\sim\textsc{Uniform}[v_L,v_H]
	\]
	She is charged the price $p$ if she receives the product, i.e.,
	$
	\textbf{t}_i(\hat{v})=-p\cdot\textbf{x}_i(\hat{v})
	$.
	Finally, the estimator is the average reported value, i.e., $\textbf{e}(\hat{v})=\nicefrac{1}{n}\sum_{i=1}^n\hat{v}_i$.
\end{defn}

The BDM mechanism might not meet the ex-post benchmark if 
signals are 
too costly.
Less obviously, BDM can also fail to meet the ex-ante benchmark.

\begin{prop}\label{P1}
	For square loss, the BDM mechanism does not guarantee the ex-ante benchmark, even in favorable equilibria.
\end{prop}
\begin{proof}
	We prove Proposition \ref{P1} using an example.
	There are binary states, types, and values, where
	\[
	\omega=\textsc{Bernoulli}(0.5)\quad\text{and}\quad \theta_i=v_i\sim\textsc{Bernoulli}(0.5)\:\text{ i.i.d.}
	\]
	The estimate $e=0.5$ minimizes expected square loss, regardless of the state $\omega$.
	To meet the ex-ante benchmark, the planner must estimate $\textbf{e}(\hat{v})=0.5$.
	
	The signal space consists of a revealing signal $s_i^R$ and a free signal $s_i^F$.
	(Similar arguments apply if we include combined and dynamic signals.)
	The revealing signal $s^R_i$ has realization $v_i$ and costs $\bar{c}>1$.
	The free signal $s^F_i$ costs nothing and has 
    realization
	$s^F_i(\omega,\theta_i)=\textbf{1}\left(v_i\geq\omega\right).$
	That is, the free signal is uninformative if the state is $\omega=0$ and reveals $v_i$ if the state is $\omega=1$.
    Recall agents do not initially know the state.
	
	The agents always acquire the free signal.
	However, they do not acquire the revealing signal, since the cost $\bar{c}$ exceeds their maximum values $v_i$.
	After observing the signal realization, agent $i$'s reported value is either
	\[
	\hat{v}_i=\ex{v_i\mid s^F_i(\omega,\theta_i)=0}=0\quad\text{or}\quad\hat{v}_i=\ex{v_i\mid s^F_i(\omega,\theta_i)=1}=\frac{2}{3}
	\]
	If the state is $\omega=0$, then all of the agents report $\hat{v}_i=2/3$.
	If the state is $\omega=1$, then roughly half of the agents report $\hat{v}_i=0$ and the other half report $\hat{v}_i=2/3$.
	
	Observe that the planner's estimate is suboptimal.
	The optimal estimate is $e=0.5$ regardless of the state $\omega$, whereas the actual estimate $\textbf{e}(\hat{v})$ is roughly $2/3$ when $\omega=0$ and $1/3$ when $\omega=1$.
	Therefore, the planner fails to meet the ex-ante benchmark.
\end{proof}

	\subsection{BDM-with-Betting Mechanism}\label{S4-1}

Although the BDM mechanism can perform quite poorly when information is costly, we propose a modification that obtains much better guarantees.

The \emph{BDM-with-betting mechanism} presents agents with an incentivized survey and asks them, before completing the survey, to predict the survey results.
The BDM mechanism is the incentivized survey, and the survey results are the average reported values.
Of course, it is not essential that we build on the BDM mechanism; we could also use alternatives like multiple price lists or randomized posted prices.

BDM-with-betting relies on \emph{proper scoring rules}.
These are rules that incentivize an agent $i$ to report her beliefs about some random variable.\footnote{Proper scoring rules tend to give higher scores to beliefs that assign higher probability to the observed value.
	There are many known proper scoring rules and they are easy to construct (e.g., \cite{McCarthy56}).
	We mostly rely on the \emph{continuous ranked probability score} and the \emph{quadratic scoring rule}.}
Here, the random variable is the average reported value, among agents other than $i$, i.e.,
$\tilde{v}_i=\tfrac{1}{n-1}\sum_{j\neq i}\hat{v}_j.$
We denote agent $i$'s beliefs over the average reported value $\tilde{v}_i$ by $b_i\in\mathcal{B}=\Delta\left(\reals\right)$.

\begin{defn}\label{D11}
	A \emph{scoring rule} $\textnormal{SR}:\mathcal{B}\times\reals\to\reals$ for agent $i$ maps her reported belief $\hat{b}_i\in\mathcal{B}$ and the average reported value $\tilde{v}_i\in\reals$ to a numerical score.
	It is \emph{proper} if she maximizes the expected score by reporting her beliefs truthfully, i.e.,
	\[
	\forall b_i\in\mathcal{B},\quad b_i\in\arg\max_{\hat{b}_i}\ex[b]{\textnormal{SR}\left(\hat{b}_i,\tilde{v}_i\right)}
	\]
\end{defn}

BDM-with-betting involves two stages.
The second stage is the BDM mechanism.
The first stage asks agents to predict the average reported value in the second stage.
Agents are paid more if their predictions turn out to be more accurate. 

\begin{defn}\label{D12}
	The \emph{BDM-with-betting mechanism} $\left(\textbf{x},\textbf{t},\textbf{e}\right)$ is parameterized by a proper scoring rule $\textnormal{SR}$ and scaling parameter $\lambda_n$.
	Each agent $i$ sends a message
	$m_i=\left(\hat{v}_i,\hat{b}_i\right)\in\reals\times\mathcal{B}$
	that consists of a reported value $\hat{v}_i$ and a reported belief $\hat{b}_i$.
	She receives the product if her reported value exceeds the random price $p$, i.e.,
	\[
	\textbf{x}_i\left(\hat{v},\hat{b}\right)=\textbf{\emph{1}}(\hat{v}_i\geq p)\quad\text{where}\quad p\sim\textsc{Uniform}[v_L,v_H]
	\]
	She is charged if she receives the product and earns a bonus from the scoring rule, i.e.,
	\[
	\textbf{t}_i\left(\hat{v},\hat{b}\right)=\lambda_n\cdot\textnormal{SR}\left(\hat{b}_i,\tilde{v}_i\right)-p\cdot\textbf{\emph{1}}(\hat{v}_i\geq p)
	\]
	Finally, the estimator is the average reported value, i.e., $\textbf{e}(\hat{v},\hat{b})=\nicefrac{1}{n}\sum_{i=1}^n\hat{v}_i$.
\end{defn}

To prove Theorem \ref{T2}, we rely on a sequence of BDM-with-betting mechanisms.
Specifically, we require a sequence of scaling parameters $\lambda_n\to\infty$ that grow with the sample size $n$.
This means that transfers become arbitrarily large.

\begin{remark}
    If the planner knew the instance $I$, then it would not be necessary to let transfers grow arbitrarily large.
    Instead, the planner could fix the scaling parameter $\lambda_n=\lambda(I)$ to some level $\lambda(I)$ that is appropriate for that instance $I$.
    The reason why we need transfers to grow large when the planner does not know the instance $I$ is that, for any constant $\alpha>0$, there exists an instance $I$ such that $\lambda(I)>\alpha$.

    It is natural to ask what additional assumptions on the instances $I$, or relaxations of our benchmarks, are needed in order to obtain positive results with bounded transfers.
    We leave this question to future work.
\end{remark}

BDM-with-betting shares some features with full surplus extraction mechanisms (e.g., \cite{CM88}), peer prediction mechanisms (e.g., \cite{MRZ05}), and Bayesian truth serums (e.g., \cite{Prelec04}).
Namely, all of these mechanisms ask agents to explicitly or implicitly bet on each others' reports, albeit in models that are qualitatively different from ours (see the related literature in Section \ref{S1-1}).

In a sense, BDM-with-betting is much more naive than any of these mechanisms.
It does not try to infer anything about values $v_i$ or the state $\omega$ from reported beliefs $\hat{b}$.
The incentives to truthfully report $\hat{v}_i=\ex{v_i\mid s_i(\omega,\theta_i)}$ are the same as in BDM.
The estimator is the same as in BDM.
The novelty of BDM-with-betting is the use of off-the-shelf scoring rules to incentivize agents to acquire signals that help them predict others' reported values.
But even that is naive, given that what the planner really wants is for agents to learn about \emph{their own} values.

Despite its naivety, we find that BDM-with-betting can incentivize agents to acquire precisely the kind of information needed to identify the population's average value (and not much else).
Next, we explain why.

	\subsection{Intuition for Theorem \ref{T2}}\label{S4-2}

We use a simple example to provide intuition for why BDM-with-betting guarantees the ex-post benchmark in favorable equilibria (i.e., the first part of Theorem \ref{T2}).
Consider an instance with binary states, types, and values, where
\[
\omega=\textsc{Bernoulli}(0.5),\quad \theta_i=v_i=\beta\omega+\epsilon_i,\quad\epsilon_i\sim\textsc{Uniform}\{-1,1\}\:\text{ i.i.d.}
\]
There are two signals: a revealing signal that costs $\bar{c}\gg 1$ and reveals the value $v_i$; and an uninformative signal that costs nothing.
There are two cases to consider.

\paragraph{Uncorrelated Values.}

Let values $v_i$ be uncorrelated across agents ($\beta=0$).
The planner's estimate when the agents acquire information is
$\tfrac{1}{n}\sum_{i=1}^nv_i=\tfrac{1}{n}\sum_{i=1}^n\epsilon_i\to_p0$,
His estimate when the agents do not acquire information is
$\tfrac{1}{n}\sum_{i=1}^n\ex{v_i}=0.$
Either way, the planner's estimate converges to the population's average value.

More generally, when values are uncorrelated, the planner can take advantage of the \emph{wisdom of the crowd}.
It is not necessary that agents learn their values perfectly.
What is important is that the errors $\hat{v}_i-v_i$ they make are uncorrelated across agents $i$.

\paragraph{Correlated Values.}

Let values $v_i$ be correlated across agents, by setting $\beta=1$.
Let the parameter $\lambda$ be large enough relative to the cost $\bar{c}$ of information.
Then agents acquire the revealing signal if it helps them predict the average reported value.

Suppose that agent $i$ expects all other agents $j\neq i$ to acquire their revealing signals and learn their values $v_j$.
Then she expects the average reported value $\tilde{v}_i$ to be
\[
\tilde{v}_i=\frac{1}{n-1}\sum_{j\neq i}v_j\to_p\omega
\]
Ideally, agent $i$ would learn directly about the state $\omega$, but that is not feasible in this example.
Instead, she learns about her value $v_i$, which is correlated with $\omega$.

This represents one of two equilibria.
In the informed equilibrium, all agents acquire the revealing signal and report $\hat{v}_i=v_i$.
The planner's estimate is optimal, i.e.,
\[
\textbf{e}(\hat{v},\hat{b})=\frac{1}{n}\sum_{i=1}^n\hat{v}_i=\frac{1}{n}\sum_{i=1}^nv_i=\omega
\]
In the uninformed equilibrium, all agents acquire the uninformative signal and report their expected values $\hat{v}_i=\ex{\omega}$.
Now, the planner's estimate is not optimal, i.e.,
\[
\textbf{e}(\hat{v},\hat{b})=\frac{1}{n}\sum_{i=1}^n\hat{v}_i=\frac{1}{n}\sum_{i=1}^n\ex{v_i}=\ex{\omega}
\]
There is no wisdom of the crowd in the uninformed equilibrium, since errors $\hat{v}_i-v_i$ are correlated across agents $i$.
Of course, even in the uninformed equilibrium, the planner's estimate at least meets the ex-ante benchmark.

\begin{remark}\label{R1}
	This equilibrium multiplicity does not necessarily arise if we provide agents with initial information about their values $v_i$.
	Suppose the uninformative signal $s_i$ is slightly informative, where $s_i(\omega,\theta_i)=\textbf{1}(v_i\geq 0)$, and allow agents to acquire dynamic signals.
	Now, agents always know whether they like or dislike the product, and can incur cost $\bar{c}$ to learn how intensely they like or dislike it.
	
	In this modified example, only the informed equilibrium survives.
	To see this, suppose that no agent acquires the revealing signal.
	If the state is $\omega=1$, then all agents $i$ report $\hat{v}_i=2/3$.
	If the state is $\omega=0$, then roughly half of the agents report $\hat{v}_i=2/3$ and the others report $\hat{v}_i=-2$.
	It follows that the average reported value is correlated with the state $\omega$, which in turn is correlated with agents' values $v_i$.
	These correlations make it profitable for agents to acquire the revealing signal.
\end{remark}

\paragraph{Generalizing the Example.}

While this example is useful for intuition, it may be misleadingly simple.
There is not always a clean separation between ``uncorrelated'' and ``correlated'' values.
Whether values are correlated depends on what information agents begin with and the properties of the signals they acquire.
This is especially true for non-binary signal spaces, where agents may deviate from the revealing signal in complicated ways.
For example, they may prefer to learn directly about the state $\omega$, or only acquire partial information about their values $v_i$.

Indeed, this example obscures some counter-intuitive phenomena that can arise in other cases.
For example, it is possible for values to be uncorrelated if agents are well-informed, but correlated if agents are not.
Pure-strategy equilibria may not exist, since agents want to acquire information if and only if others do not.
This can arise even in simple instances, like the one we used to prove Proposition \ref{P1} in Section \ref{S3-2}.
Rather than attempting to construct mixed-strategy equilibria in strategy spaces of arbitrary size/dimension, our non-constructive proof circumvents these issues.

In general, we find that BDM-with-betting sustains a sequence of equilibria where a wisdom-of-the-crowd effect holds in the limit.
That is, betting incentivizes agents to acquire information until the point where their reporting errors $\hat{v}_i-v_i$ are nearly-uncorrelated.
We outline the proof in Appendix \ref{Ap2-T2}.

	\subsection{Intuition for Theorem \ref{T1}}\label{S3-1}

The BDM-with-betting mechanism is intended to elicit the population's average value, which minimizes expected square loss.
Next, we provide intuition for why it is harder to elicit other statistics, like quantity demanded or revenue-maximizing price.

Consider an instance with a single state and binary values/types, where
$\theta_i=v_i\sim\textsc{Uniform}\{-\alpha,\alpha\}\:\text{ i.i.d.}$
Since there is only one state, the ex-post and ex-ante benchmarks coincide.
Let $e^*(\alpha)$ be the estimate that minimizes expected loss.
For loss functions that are not square-like, $e^*(\alpha)$ is not necessarily identified for two reasons:
\begin{enumerate}
	\item \emph{The optimal estimate is not identified if agents are not informed.}
	To see this, consider a loss function that is not a square-like.
	Then the optimal estimate $e^*(\alpha)$ will depend on the parameter $\alpha$.
	If agents do not acquire information, their willingness to pay does not identify the optimal estimate.
	
	\item \emph{It is difficult to incentive agents to acquire information.}
	Agent $i$ has two reasons to acquire information.
	First, to decide whether she is willing to ``buy'' the product at a given price.
	This is not sufficient because the information cost $\bar{c}$ far exceeds the agent's uncertainty over her value $v_i$.
	Second, to better predict the messages $m_{-i}$ sent by other agents.
	This is not sufficient because agent $i$'s value $v_i$ is independent of any private information that the other agents might have.
\end{enumerate}

One way to circumvent these difficulties is by having agents directly report the parameter $\alpha$---or, more generally, the instance $I$---and penalize them if any of their reports disagree.
This does not require agents to acquire information because the instance $I$ (and therefore, $\alpha$) is common knowledge.
Among other issues that make this mechanism impractical, it features severe equilibrium multiplicity.
In equilibria where agents collectively report $\alpha'\neq\alpha$, the planner's estimate $e^*(\alpha')\neq e^*(\alpha)$ may be suboptimal.
In that case, this mechanism fails to meet the ex-ante benchmark.

It is possible to imagine many other mechanisms.
The challenge at this point is to rule out all possible mechanisms, for any loss function that is not square-like.
We leave this to the proof outline in Appendix \ref{Ap2-T1}.

\paragraph{Sensitivity of Theorem \ref{T1}.}

We complete this section by discussing how sensitive this negative result is to natural restrictions on the space of instances.
It follows immediately from our proof that Theorem \ref{T1} holds even if we restrict attention to singleton state spaces, ternary type/value spaces, value functions $v(\theta_i)=\theta_i$, cost functions that do not depend on types, and binary signal spaces.

Depending on the application, it may be reasonable to impose other restrictions.
Here are three examples.
First, we might assume that a representative fraction of the population can learn their values at zero cost.
Second, we might assume that values are at least slightly correlated across agents (as in \cite{CM88}).
Third, we could impose a tight upper bound $\bar{c}$ on the cost of the revealing signal.
None of these restrictions are enough to meaningfully change Theorem \ref{T1}.\footnote{It is easy to extend the proof to handle the first and second restrictions. In particular, we can handle the second restriction by introducing a free state-revealing signal $s_i(\omega,\theta_i)=\omega$. For the third restriction, we would need to strengthen Definition \ref{D8} by requiring constants $|a-b|\ll\bar{c}$.}

It would also be natural to weaken Theorem \ref{T1} by considering mechanisms that only approximately meet our benchmarks.
This, paired with natural restrictions on the space of instances, would be a reasonable way to obtain more positive results.

	\section{Social Choice}\label{S5}

We now turn to the problem of social choice, where a planner must choose a single alternative on behalf of a population.
Specifically, we seek social choice mechanisms that work well even when processing information is costly (Theorem \ref{T3}).

We reinterpret the model in Section \ref{S2} as follows.
There are $n$ members of a committee (or citizen's assembly, mini-public, jury, etc.) that are randomly chosen to represent a large population.%
\footnote{This practice, known as sortition, has its origins in ancient Athenian democracy. It is the subject of renewed interest in political science (e.g., \cite{Fishkin11}), computer science (e.g., \cite{FGGHP21}), and economics (e.g., \cite{BB23}).}
The planner chooses a population-wide outcome $e\in\{0,1\}$ that indicates, say, whether or not to pass a ballot measure.
The allocations $x_i$ to each agent $i$ must match the population-wide outcome $e$, and the planner's objective is to maximize the population's welfare.

\begin{assume}\label{A3}
	Let estimates $\mathcal{E}=\{0,1\}$ be binary and the loss function be utilitarian; i.e.,
	$L\left(e,F_{\mathcal{V}\mid\omega}\right)=-e\cdot\ex[F_{\mathcal{V}\mid\omega}]{v_i}$.
	Restrict attention to mechanisms $(\textbf{x},\textbf{t},\textbf{e})$ where, for all message profiles $m$,
	\begin{equation}\label{E72}
	\textbf{e}(m)=\textbf{x}_1(m)=\textbf{x}_2(m)=\ldots=\textbf{x}_n(m)
	\end{equation}
	Going forward, we refer to \emph{social choice mechanisms} $(\textbf{x},\textbf{t})$, where $\textbf{x}(m)=\eqref{E72}$.
\end{assume}

When evaluating welfare, we ignore the committee members' cost of information acquisition.
This strikes us as reasonable when the size $n$ of the committee is much smaller than the size of the population.

	
	
	\subsection{Challenges}\label{S5-2}

When trying to solve the planner's problem, we encounter challenges that would not arise if agents knew their own values.
To explain these challenges, we use the \emph{majority-rule mechanism} as a foil.

\begin{defn}
	The \emph{majority-rule mechanism} $(\textbf{x},\textbf{t})$ asks each agent $i$ to report her preferred alternative $\hat{x}_i\in\{0,1\}$.
	Then it sets
	\[
	\textbf{t}\left(\hat{x}_1,\ldots,\hat{x}_n\right)=0\quad\text{and}\quad\textbf{x}\left(\hat{x}_1,\ldots,\hat{x}_n\right)=\textnormal{\textbf{1}}\left(\frac{1}{n}\sum_{i=1}^n\hat{x}_i\geq\frac{1}{2}\right)
	\]
\end{defn}

In addition to the standard problems with majority rule, there are three additional problems that arise when information acquisition is costly.
\begin{enumerate}
	\item \emph{Majority rule only rewards information acquisition in the event that an agent is \emph{pivotal}, which is is unlikely (e.g., \cite{MH03}).}
	More precisely, an agent $i$ is pivotal if changing her report $\hat{x}_i$ would change the allocation.
	The probability that an agent is pivotal tends to shrink extremely quickly as the sample size $n$ grows, so there is little expected gain from acquiring information.
	
	\item \emph{Majority rule does not account for the positive externalities of information acquisition, where one informed vote can benefit everyone (e.g., \cite{Gersbach95})}
	For example, suppose that agents have common values and that agent $i$ knows she will be pivotal.
	Acquiring information that identifies the optimal alternative would not only improve her well-being, but also the well-being of everyone else. 
	
	\item \emph{When agents are imperfectly informed, majority rule does not incentivize them to truthfully report the alternative they prefer given the information they acquired.
	}
	Since agent $i$'s report only matters if she is pivotal, it is optimal for her to condition her beliefs on being pivotal before making a report.
	This can distort the outcome of the vote (e.g., \cite{AB96}).
	
\end{enumerate}
We stress that these problems are not unique to majority rule.
They apply to many other mechanisms -- like the Vickrey-Clarke-Groves (VCG) mechanism, or quadratic voting -- that do not account for voters' costs of processing information.

The third challenge is the most daunting.
The first two deal with incentives to acquire information, which we already studied in Section \ref{S4}.
The third challenge deals with incentives to be truthful when voters are not fully informed, and we know that it is not always possible to incentivize voters to become fully informed (Theorem \ref{T1}).

Fortunately, it turns out that we can address all three challenges at once, by adapting the BDM-with-betting mechanism to the social choice setting.

	\subsection{Statement of Theorem \ref{T3}}\label{S5-3}

We can now state the main result of our application to social choice.

\begin{theorem}\label{T3}
	There exists a sequence of social choice mechanisms that:
	\begin{enumerate}
		\item Guarantees the ex-ante benchmark in all equilibria.
		\item Guarantees the ex-post benchmark in favorable equilibria.
	\end{enumerate}
\end{theorem}

As in Section \ref{S4}, even the ex-ante guarantee is non-trivial.
\citeauthor{NMS25} (\citeyear{NMS25}, p.2-3) offer an excellent example where majority rule fails to meet the ex-ante benchmark (which they call the no-information benchmark).
Majority rule performs poorly in this example because voters are not truthful, in contrast to Proposition \ref{P1}.

The mechanism we use to prove Theorem \ref{T3} builds on the Vickrey-Clarke-Groves (VCG) mechanism.
Unfortunately, the VCG mechanism involves vote buying, which is a nonstarter in many applications.
For these reasons, we now restrict attention to a special case of the model where we can motivate a more practical mechanism.

\begin{assume}\label{A4}
	Restrict attention to instances $I$ with a binary value space $\mathcal{V}=\{-1,1\}$ and a binary signal space $\mathcal{S}=\{s_i^R,s_i^U\}$, where $s_i^R$ is a revealing signal and $s_i^U$ is an uninformative signal that sets $s^U_i(\cdot)=0$.
\end{assume}

Assumption \ref{A4} allows us to build on the majority-rule mechanism.
Majority rule is widely-used in practice, but fails to account for the intensity of voters' preferences.
This means that majority rule can be inefficient even if voters know their values $v_i$.
Assumption \ref{A4} ensures that voters' preferences are equally intense, since their values are $v_i\in\{-1,1\}$.
By restricting attention to all-or-nothing information acquisition, it also avoids complications that arise when voters are partially-informed.%
\footnote{For example, when voters are partially-informed, a minority of voters who are confident that the measure is bad for them may be overruled by a majority of voters who believe it is slightly more likely that the measure is good for them.}

We maintain Assumption \ref{A4} for the rest of this section.
In the Supplemental Appendix, we show that Theorem \ref{T3} holds even without this assumption.

\subsection{Majority-Rule-with-Betting Mechanism}\label{S5-5}

Assumption \ref{A4} allows us to motivate the \emph{majority-rule-with-betting mechanism}.
Majority-rule-with-betting has voters cast their votes according to the majority rule and asks them, before they cast their votes, to predict the results.

To be more precise, majority-rule-with-betting has voters cast their votes according to the majority rule with high probability.
For technical reasons, there is also a small probability $\delta_n>0$ that a randomly-chosen agent dictates the outcome.%
\footnote{This ensures that voters are pivotal with positive probability and therefore avoids situations where voters must condition on zero probability events when casting their votes.
}

To define this mechanism, we need to update our notation from Section \ref{S4-1}.
Let the \emph{vote share} from agent $i$'s perspective be
$\tilde{n}_i=\tfrac{1}{n}\sum_{j\neq i}\hat{x}_j$.
Let agent $i$'s belief over the vote share be
$
b_i\in\mathcal{B}=\Delta\left([0,1]\right)
$.
A scoring rule $\textnormal{SR}$ now maps a reported belief $\hat{b}_i\in\mathcal{B}$ and a vote share $\tilde{n}_i\in[0,1]$ to a numerical score.

\begin{defn}
	The \emph{majority-rule-with-betting mechanism} $\left(\textbf{x},\textbf{t}\right)$ is parameterized by a proper scoring rule $\textnormal{SR}$, scaling parameter $\lambda_n$, and probability $\delta_n$.
	Agent $i$ sends a message
	$m_i=\left(\hat{x}_i,\hat{b}_i\right)\in\{0,1\}\times\mathcal{B}$
	that consists of a reported alternative $\hat{x}_i$ and reported belief $\hat{b}_i$.
	With probability $1-\delta_n$, the planner selects the majority's preferred alternative, i.e.,
	\[
	\textbf{x}(\hat{x},\hat{b})
	=\textbf{\emph{1}}\left(\frac{1}{n}\sum_{i=1}^n\hat{x}_i\geq\frac{1}{2}\right)
	\]
	With probability $\delta_n$, the planner randomly chooses a voter $i\sim\textsc{Uniform}(1,\ldots,n)$ and selects their reported alternative, i.e.,
	$
	\textbf{x}(\hat{x},\hat{b})=\hat{x}_i
	$.
	Finally, each agent $i$ is paid for her prediction according to the scoring rule, i.e.,
	\[
	\textbf{t}_i\left(\hat{v},\hat{b}\right)=\lambda_n\cdot\textnormal{SR}\left(\hat{b}_i,\tilde{n}_i\right)
	\]
\end{defn}

This mechanism has two desirable features.
	First, it does not allow anyone to buy votes.
	Transfers are used only to reward accurate predictions.
	Second, the mechanism counts all votes equally.
	In particular, voter $i$'s ability to predict the vote share does not affect how much weight is attached to her vote.

It is also worth noting that, in practice, this mechanism would not necessarily require the planner to spend more on elections.
This is because existing elections engage the entire electorate, whereas majority-rule-with-betting only engages a small fraction of the electorate (i.e., the committee).
The planner could use his savings from scaling down the election to fund the rewards for accurate predictions.

Finally, majority-rule-with-betting may not appear that different from the status quo in many democracies, where voters are free to participate in political betting markets.
	However, we do not ask voters to bet against each other.%
    \footnote{Asking voters to bet against each other would be cheaper. However, doing so would likely undermine incentives to acquire information. It is likely that a no-trade theorem would apply in this setting, which means that voters would choose not to place bets in equilibrium. If voters do not place bets, then they do not have strong incentives to acquire information.}
	
	\subsection{Intuition for Theorem \ref{T3}}\label{S5-4}

Recall the three challenges from Section \ref{S5-2}.
The first two challenges concern incentives to acquire information.
The reasons why betting incentivizes information acquisition here are essentially the same as the reasons why betting incentivizes information acquisition in Theorem \ref{T2} (see the discussion in Section \ref{S4-2}).

The third challenge, which concerns incentives to be truthful, is new.
Fortunately, it turns out that the kind of information that voters need to acquire in order to be truthful is closely-related to the kind of information that is relevant to predicting the vote share.
This means that  a convenient side effect of betting on the vote share is that it restores truth-telling incentives.

To see this, recall why majority rule may not incentivize truth-telling on its own.
Agent $i$ knows that her vote $\hat{x}_i$ is pivotal only if the vote share $\tilde{v}_i$ is exactly 50\%.
Therefore, it is in her best interest to vote in favor of the measure whenever her expected value $v_i$ conditional on being pivotal is positive, i.e.,
\begin{equation}\label{S5-E1}
\ex{v_i\mid s_i,s_i(\omega,\theta_i),\tilde{v}_i=50\%}\geq 0
\end{equation}
A truthful agent would not condition on pivotality and vote in favor whenever
\begin{equation}\label{S5-E2}
\ex{v_i\mid s_i,s_i(\omega,\theta_i)}\geq 0
\end{equation}

What determines whether agent $i$ will be truthful?
It is whether agent $i$ perceives the vote share $\tilde{v}_i$ to be correlated with her value $v_i$ even after conditioning on her information.
If not, then conditioning on pivotality does not affect agent $i$'s beliefs about her value.
As a result, conditions \eqref{S5-E1} and \eqref{S5-E2} will be the same.

The reason why high-stakes betting incentivizes agent $i$ to acquire information until she is truthful is because it incentivizes her to acquire any information that is correlated with the vote share $\tilde{v}_i$.
On the one hand, if $v_i$ is correlated with the vote share $\tilde{v}_i$, then agent $i$ benefits from acquiring her revealing signal.
Once she acquires the revealing signal, she prefers to vote truthfully.
On the other hand, if $v_i$ is not correlated with the vote share $\tilde{v}_i$, then agent $i$ also prefers to vote truthfully.

We leave further discussion to the proof outline in Appendix \ref{Ap2-T3}.
	
	
	
	\section{Conclusion}\label{S7}

Revealed preference plays a foundational role in economics, by linking observed choices with unobserved preferences.
This link is disrupted when people find it costly to process information about the goods they consume, the services they receive, and the policies that affect them.
We ask to what extent it is possible to repair that link.
We obtain both positive results (Theorems \ref{T2} and \ref{T3}) and negative results (Theorem \ref{T1}), depending on the application.
In the process, we derive relatively-simple mechanisms (e.g., BDM-with-betting and majority-rule-with-betting) that appear to be more robust to costly information processing than their classical counterparts.

There are four natural directions for future work.
The first direction would empirically evaluate BDM-with-betting and majority-rule-with-betting.
The second would identify either theoretical conditions under which these mechanisms have unique equilibria, or empirical conditions under which people coordinate on favorable equilibria.
The third direction would ask when these mechanisms perform well for small sample sizes and tight budgets.
The fourth would explore whether it is possible to relax the assumption that the agents share a common prior.

More broadly, there are many institutions where costly information processing can lead to market failures (e.g., elections), and many institutions that exist in part to correct such failures (e.g., recommender systems).
Addressing these market failures may ultimately require better methods for eliciting informed preferences.

	\appendix
	
	\section{Proof Outlines}\label{Ap2}
	
	\subsection{Notation}

In strategy profile $(\textbf{s}^n,\textbf{m}^n)$ where agent $i$ acquires signal $s_i$,
$
\pr[i,n]{\cdot}=\pr{\cdot\mid s_i,s_i(\omega,\theta_i)}
$
is the probability conditioned on agent $i$'s information.
Define expected value $\ex[i,n]{\cdot}$, variance $\var[i,n]{\cdot}$, and covariance $\cov[i,n]{\cdot}$ similarly.
	
	\subsection{Proof of Theorem \ref{T2}}\label{Ap2-T2}

We outline the proof of Theorem \ref{T2}.
Let $(\textbf{x}^n,\textbf{t}^n,\textbf{e}^n)$ be a BDM-with-betting mechanism where the proper scoring rule is the continuous ranked probability score, i.e., 
\[
\textnormal{SR}\left(b_i,\tilde{v}_i\right)=-\int_{\left[0,\bar{v}\right]}\left(\pr[b_i]{\tilde{v}_i\leq z}-\textbf{1}(\tilde{v}_i\leq z)\right)^2dz
\]
To avoid negative transfers, we add a constant to this scoring rule.

Next, we introduce key notation and sufficient conditions for eliciting the population's average value (which minimizes expected square loss).
Lemmas \ref{L2}-\ref{L4} verify those sufficient conditions.
Corollary \ref{L8} handles the ex-ante benchmark.

\subsubsection{Notation and Sufficient Conditions}

We begin by introducing key notation and finding sufficient conditions for weakly identifying the population's average value.
Let agent $i$'s \emph{error} be the difference between her true and expected values, i.e.,
$\epsilon_i=v_i-\ex{v_i\mid s_i,s_i(\omega,\theta_i)}.$
Since the BDM mechanism incentivizes truth-telling, her reported value is
$\hat{v}_i=\ex{v_i\mid s_i,s_i(\omega,\theta_i)}=v_i-\epsilon_i.$
The \emph{average error} is the difference between the \emph{average value} and the average reported value (across all agents).
That is,
\[
\underbrace{\frac{1}{n}\sum_{i=1}^{n}\epsilon_i}_\text{average error}
=\underbrace{\frac{1}{n}\sum_{i=1}^{n}v_i}_\text{average value}
-\underbrace{\frac{1}{n}\sum_{i=1}^{n}\hat{v}_i}_\text{average reported value}
\]
If the average error is small, the estimate will be nearly optimal.
In the limit as the sample size $n$ grows, these averages converge in probability to expectations $\ex{\cdot\mid\omega}$ conditional on the realized state $\omega$:
\begin{equation}\label{E23}
	\frac{1}{n}\sum_{i=1}^{n}\epsilon_i\to_p\ex{\epsilon_i\mid\omega}
	\quad\quad
	\frac{1}{n}\sum_{i=1}^{n}v_i\to_p\ex{v_i\mid\omega}
	\quad\quad
	\frac{1}{n}\sum_{i=1}^{n}\hat{v}_i\to_p\ex{\hat{v}_i\mid\omega}
\end{equation}
This follows from the fact that the variables $\epsilon_i,v_i,\hat{v}_i$ are conditionally i.i.d., which allows us to invoke the law of large numbers after conditioning on $\omega$.

For convenience, we often refer to the probability limits \eqref{E23} as the average error, average value, and average reported value, respectively.
We want to show that the average error vanishes, i.e.,
	$\ex{\epsilon_i\mid\omega}\to_p0$,
which ensures that the average reported value converges to the population's average value.
But it is not critical that the individual errors vanish.
It is enough to show that each agent $i$'s error is uncorrelated with the average error, i.e., $
\cov{\epsilon_i,\ex{\epsilon_j\mid\omega}}\to 0
$.
Intuitively, this lack of correlation ensures that individual errors to not translate into aggregate errors.
This is why BDM-with-betting works even though it does not directly incentivize participants to learn about their own values:
it really only needs to incentivize participants $i,j$ to acquire information up to the point where their errors $\epsilon_i,\epsilon_j$ are uncorrelated.

Finally, observe that agent $i$'s error is uncorrelated with the average error when
\begin{equation}\label{E25}
\cov{\epsilon_i,\ex{\hat{v}_i\mid\omega}}\to0\quad\text{and}\quad\cov{\epsilon_i,\ex{v_j\mid\omega}}\to0
\end{equation}
This follows from the fact that $
\ex{\epsilon_j\mid\omega}=\ex{v_j\mid\omega}-\ex{\hat{v}_i\mid\omega}
$.



\subsubsection{Verifying Sufficient Conditions}

We verify the sufficient conditions for eliciting the average, in three lemmas.
First, we show that errors $\epsilon_i$ cannot be too correlated with the average reported value $\ex{\tilde{v}_i\mid\omega}$.

\begin{lemma}\label{L2}
	In any equilibrium of mechanism $(\textbf{x}_n,\textbf{t}_n)$, the covariance between the error $\epsilon_{i,n}$ and conditional expectation $\ex{\tilde{v}_{i,n}\mid\omega}$ is bounded: 
	$\left|\cov{\epsilon_{i,n},\ex{\tilde{v}_{i,n}\mid\omega}}\right|=O\left(\nicefrac{1}{\sqrt{\lambda_n}}\right)
    $
\end{lemma}

Lemma \ref{L2} follows from the fact that participant $i$ wants to acquire any signal that helps her better predict the average reported value $\ex{\hat{v}_j\mid\omega}$.
Suppose that, for the sake of contradiction, her error $\epsilon_i$ is correlated with the average reported value.
Then, if there were a signal that revealed her error $\epsilon_i$, she would want to acquire it.
However, this signal exists.
By Assumptions \ref{A5} and \ref{A6}, $i$ can combine her revealing signal (which reveals $v_i$) with the signals that she has already acquired (which determine $\hat{v}_i$).
This reveals her error $\epsilon_i=v_i-\hat{v}_i$ and costs at most $\bar{c}$ more than the signals she already acquired.
This combination is a profitable deviation whenever the scaling parameter $\lambda_n$ is large relative to the cost $\bar{c}$.
That, in turn, contradicts the premise that $\epsilon_i$ is correlated with the average reported value in equilibrium.

As a corollary of Lemma \ref{L2}, we find that BDM-with-betting mechanism guarantees the ex-ante benchmark in all equilibria.

\begin{cor}\label{L8}
	 Mechanisms $(\textbf{x}^n,\textbf{t}^n)$ guarantee the ex-ante benchmark in all equilibria.
\end{cor}

Next, we must show that $\epsilon_i$ is not too correlated with the average value $\ex{v_j\mid\omega}$.
However, this is not true in every equilibrium.
We must show that there exists an equilibrium in which $\epsilon_i$ is not too correlated with the average value.

We take an indirect approach.
For every instance, we define an \emph{auxiliary instance} where, by construction, $\epsilon_i$ is not correlated with the average value.
Using Lemma \ref{L2}, we show that the average error vanishes in every equilibrium of the auxiliary instance (Lemma \ref{L3}).
Then we show, when $n$ is large, every equilibrium of the auxiliary instance is also an equilibrium of the original instance (Lemma \ref{L4}).

Essentially, the auxiliary instance forces agents to only acquire signals that are maximally predictive of the conditional expectation $\ex{v_j\mid\omega}$.

\begin{defn}\label{D100}
	For any given instance $I=\left(\Omega,\mathcal{V},\Theta,\mathcal{S},F,v,c\right)$, we define an \emph{auxiliary instance} $\tilde{I}=\left(\Omega,\mathcal{V},\Theta,\mathcal{S},F,v,\tilde{c}\right)$.
	There are three steps to this construction.
	\begin{enumerate}
		\item Let $\bar{s}_i\in\mathcal{S}$ be the signal that combines all other signals $s_i\in\mathcal{S}$ that have finite cost $c\left(s_i,\theta_i\right)<\infty$ for all types $\theta_i$ in the support of $F$.\footnote{If signal $s_i$ is itself a combined signal (Assumption \ref{A5}), let $\bar{s}_i$ include all base signals that $s_i$ combines.}
		
		\item We can evaluate a given signal $s_i$ by how well it predicts average value $\ex{v_j\mid\omega}$.
		Let \emph{predictiveness} be the maximum expected score when agent $i$ acquires $s_i$, i.e.,
		\begin{equation}\label{E12}
			\rho(s_i)=\ex{\max_{b_i}\ex{\textnormal{SR}^{\textnormal{CRPS}}\left(b_i,\ex{v_i\mid\omega}\right)\mid s_i(\omega,\theta_i)}}
		\end{equation}
		
		\item For any given type $\theta_i\in\Theta$, let every signal $s_i$ that is less predictive than the combined signal $\bar{s}_i$ have infinite cost according to the auxiliary cost function $\tilde{c}$.
		Let other signals cost the same as in the original cost function $c$.
		More precisely,
		\begin{equation}\label{T2-E1}
			\tilde{c}(\theta_i,s_i)=\begin{cases}
				\infty&\rho\left(\theta_i,s_i\right)<\rho\left(\theta_i,\bar{s}_i\right)\\
				c(\theta_i,s_i)&\textnormal{otherwise}
			\end{cases}
		\end{equation}
	\end{enumerate}
\end{defn}

Next, we show average error vanishes in all equilibria of the auxiliary instance.

\begin{lemma}\label{L3}
	Fix the auxiliary instance.
	For any sequence of equilibria, the average reported value converges in probability to the population's average value.
\end{lemma}

The proof of Lemma \ref{L3} follows the logic of condition \eqref{E25}.
Lemma \ref{L2} already shows that the error $\epsilon_i$ cannot be too correlated with the average reported value $\ex{\hat{v}_j\mid\omega}$.
All that remains, by condition \eqref{E25}, is to show that the error $\epsilon_i$ cannot be correlated with the average value $\ex{v_j\mid\omega}$.
This is because the only signals $s_i$ with finite cost in the auxiliary instance are those that are maximally predictive of the average value.
Suppose, for the sake of contradiction, that $\epsilon_i$ is correlated with the average value.
Then there exists another signal -- which combines $s_i$ with the revealing signal, and reveals $\epsilon_i$ -- that predicts the average value better than $s_i$.
This contradicts the premise that $s_i$ is maximally predictive.

At this point, we have shown that the average error vanishes for all auxiliary instances.
Lemma \ref{L4} extends this result to general instances.

\begin{lemma}\label{L4}
	There exists a constant $N$ such that, for any sample size $n\geq N$, every equilibrium of the auxiliary instance is also an equilibrium of the original instance.
\end{lemma}

To prove Lemma \ref{L4}, it is enough to show that, in the candidate equilibrium of the original instance, no participant $i$ will deviate to a signal that is unavailable in the auxiliary instance.
That is, we want to show that participant $i$ wants to acquire only signals that are maximally predictive of the average value $\ex{v_j\mid\omega}$.
We know, in the limit, that she will acquire only signals that are maximally predictive of the average reported value $\ex{\hat{v}_j\mid\omega}$.
Moreover, by Lemma \ref{L3}, in any equilibrium of the auxiliary instance, the average reported value converges to the average value. 
Therefore, she will only acquire signals that are maximally predictive of the average value.%

	\subsection{Proof of Theorem \ref{T1}}\label{Ap2-T1}

Fix a loss function $L$ that is not square-like.
Let $a,b\in\mathcal{V}$ be the constants referred to in Definition \ref{D8}.
Fix a sequence of mechanisms $(\textbf{x}^n,\textbf{t}^n,\textbf{e}^n)$.

We construct two instances $I$ and $I'$ that must lead to different estimates $\textbf{e}(m)$.
In both instances, let $\Omega$ be a singleton.
This ensures that the ex-ante and ex-post benchmarks coincide.
Let the type and value spaces be $\Theta=\mathcal{V}=\{a,b,(a+b)/2\}$, where $v(\theta_i)=\theta_i$.
Let signal space $\mathcal{S}$ consist of a revealing signal $s^R_i(\omega,\theta_i)=v_i$ that costs $\bar{c}>2|b-a|$ and an uninformative signal $s^U_i(\omega,\theta_i)=0$ that costs nothing.

All that remains is to specify the distribution of values.
    For instance $I$, the distribution $F$ sets $v_i\sim\textsc{Uniform}\{a,b\}$.
	%
    For instance $I'$, the distribution $F'$ sets $v_i=(a+b)/2$.
It follows from Definition \ref{D8} that the optimal estimate differs between instances $I,I'$.
To guarantee the ex-ante benchmark in all equilibria, the planner's estimate must differ across these two instances, for any equilibrium selection rule.

\begin{lemma}\label{L1}
	For any mechanism $(\textbf{x},\textbf{t},\textbf{e})$, there are strategy profiles $(\textbf{s},\textbf{m}),(\textbf{s}',\textbf{m}')$ where:
	\begin{enumerate}
		\item Given instance $I$, $(\textbf{s},\textbf{m})$ is an equilibrium.
		
		\item Given instance $I'$, $(\textbf{s}',\textbf{m}')$ is an equilibrium.

		\item Let $s_i\sim \textbf{s}_i$, $m_i\sim \textbf{m}_i\left(s_i,s_i(\omega,\theta_i)\right)$, $s'_i\sim \textbf{s}'_i$, and $m'_i\sim\textbf{m}'_i\left(s'_i,s'_i(\omega,\theta_i)\right)$.
		The distribution of $m$ in instance $I$ is identical to the distribution of $m'$ in instance $I'$.
	\end{enumerate}
\end{lemma}

Lemma \ref{L1} ensures that, with suitable equilibrium selection, the planner's data has the same distribution regardless of whether the instance is $I$ or $I'$.
It follows that the planner's estimate is identically-distributed in both cases.
Therefore, the planner fails to meet the ex-ante benchmark in at least one of these instances.

	\subsection{Proof of Theorem \ref{T3} (Special Case)}\label{Ap2-T3}

We outline the proof of Theorem \ref{T3} in the special case where Assumption \ref{A4} holds.
We construct a sequence of majority-rule-with-betting mechanisms $(\textbf{x}^n,\textbf{t}^n)$ with proper scoring rule $\textnormal{SR}$, scaling parameter $\lambda_n$ and probability parameter $\delta_n$.

First, we define the proper scoring rule.
Without loss of generality, let the sample size $n$ be odd.
Let $Q_{i,n}$ be the event that agent $i$ is pivotal; that is, either $i$ is dictator or
$(n-1)\cdot\tilde{n}_{i,n}=\sum_{j\neq i}\hat{x}_{i,n}=\frac{n-1}{2}$.
Let $q_{i,n}$ be the probability of this conditional on $i$'s information.
Let $\textnormal{SR}$ be the sum of a quadratic scoring rule and a continuous ranked probability score.
That is,
\[
\textnormal{SR}\left(\hat{b}_i,\tilde{v}_{i,n}\right)
=\textnormal{SR}_n^{\textnormal{Q}}\left(\hat{b}_{i,n},\tilde{v}_{i,n}\right)
+\textnormal{SR}^{\textnormal{CRPS}}\left(\hat{b}_{i,n},\tilde{v}_{i,n}\right)
\]
where $\textnormal{SR}^{\textnormal{CRPS}}$ is defined as in Appendix \ref{Ap2-T2} (replacing $\tilde{v}_{i,n}$ with $\tilde{n}_{i,n}$).
To define $\textnormal{SR}_n^{\textnormal{Q}}$, let
\[
\textnormal{SR}_n^{\textnormal{Q}}\left(\hat{b}_i,\tilde{v}_{i,n}\right)=2\cdot\textbf{1}(Q_{i,n})\cdot\hat{q}_{i,n}+2\cdot\textbf{1}(\lnot Q_{i,n})\cdot(1-\hat{q}_{i,n})-\hat{q}_{i,n}^2-(1-\hat{q}_{i,n})^2
\]
Each of these scoring rules is convenient in different parts of the proof.
Since each one is proper, the sum $\textnormal{SR}$ is also proper.

It is important that $\delta_n\to0$ and $\lambda_n\to\infty$.
When $\delta_n$ vanishes quickly, agents are unlikely to be pivotal, and therefore $\lambda_n$ must grow even more quickly in order to ensure near-truthfulness.






Next, we show that these majority-rule-with-betting mechanisms are \emph{nearly truthful}.
This means that each agent's reported alternative $\hat{x}_i$ probably approximately maximizes her expected value conditional on the information she acquires.

\begin{defn}
	The majority-rule-with-betting mechanism $\left(\textbf{x}^n,\textbf{t}^n\right)$ is \emph{$(\phi_1,\phi_2)$-truthful} if the following holds for every instance $I\in\mathcal{I}$ and equilibrium $(\textbf{s},\hat{\textbf{x}},\hat{\textbf{b}})$.
	Let agent $i$ report alternative $\hat{x}_i=\hat{\textbf{x}}_i\left(\textbf{s}_i,\textbf{s}_i(\omega,\theta_i)\right)$.
	Then	
	\[
	\pr{\max_{i=1,\ldots,n}\left|\max_{x\in\mathcal{X}}\ex{v_ix\mid \textbf{s}_i,\textbf{s}_i(\omega,\theta_i)}-\ex{v_i\hat{x}_i\mid \textbf{s}_i,\textbf{s}_i(\omega,\theta_i)}\right|>\phi_1}\leq\phi_2
	\]
\end{defn}

\begin{lemma}\label{L5}
	The mechanism $\left(\textbf{x}^n,\textbf{t}^n\right)$ is $(\phi_{1n},\phi_{2n})$-truthful, where $\phi_{1n}\to 0$ and $\phi_{2n}\to 0$.
\end{lemma}

We also prove a result analogous to Lemma \ref{L2}.
This says that agent $i$'s error $\epsilon_i$ is uncorrelated with the vote share $\tilde{n}_{i,n}$.

\begin{lemma}\label{L6}
	In any strategy profile of the mechanism $(\textbf{x}_n,\textbf{t}_n)$ where agent $i$ does not have a profitable deviation,
	\[
	\left|\cov{\epsilon_{i,n},\ex{\tilde{n}_{i,n}\mid\omega}}\right|=O\left(\lambda_n^{-1/2}\right)
	\]
\end{lemma}
\begin{proof}
	Same as proof of Lemma \ref{L2}, except that $Z_{i,n}=\tilde{n}_{i,n}$.
\end{proof}

So far, we have not relied on Assumption \ref{A4}.
Now, we use it to show that efficient equilibria exist.
Let $\mu=\ex{v_i}$.
Without loss, let $\mu\neq 0$.
There are two cases to consider.

\emph{Case 1: }
Suppose that
	$
	\ex{(v_i-\mu)\cdot\pr{v_j=1\mid\omega}}>0
	$.
	Then there exists some number $N$ such that, for any sample size $n\geq N$, there is an equilibrium where every agent $i$ acquires the revealing signal and reports $\hat{x}_i=\textbf{1}(v_i=1)$.
	To see this, note that the only deviation for agent $i$ that might be profitable involves acquiring the uninformative signal.
	Then $\epsilon_i=v_i-\mu$, and
	\[
	\ex{\epsilon_i\cdot\ex{\tilde{n}_{i,n}\mid\omega}}=\ex{\epsilon_i\cdot\pr{v_j=1\mid\omega}}=\ex{(v_i-\mu)\cdot\pr{v_j=1\mid\omega}}>0
	\]
	Since $\lambda_n\to\infty$, this eventually contradicts Lemma \ref{L6}.
	
	Finally, observe that the vote share converges to $\pr{v_i=1\mid\omega}$, which exceeds 50\% if and only if $\ex{v_i\mid\omega}\geq 0$.
	In the limit, since $\delta_n\to0$, alternative $x=1$ is chosen if and only if it maximizes welfare.
	
    \emph{Case 2: } Suppose that $\ex{(v_i-\mu)\cdot\pr{v_j=0\mid\omega}}=0$.
	Observe that
	\begin{align*}
		\ex{v_i\mid\omega}
		&=\pr{v_i=1\mid\omega}-\pr{v_i=-1\mid\omega}\\
		&=\pr{v_i=1\mid\omega}-1+\pr{v_i=1\mid\omega}\\
		&=2\pr{v_i=1\mid\omega}-1
	\end{align*}
	Combining this with our supposition implies that
	$\ex{(v_i-\mu)\cdot\frac{\ex{v_i\mid\omega}+1}{2}}=0.$
	After simplifying this expression, we find that $\var{\ex{v_i\mid\omega}}=0$.
	Therefore, $\ex{v_i\mid\omega}=\mu$.
	The outcome $x=1$ is efficient if $\mu>0$ and inefficient if $\mu<0$.
	
	Now, consider a strategy profile where every agent $i$ acquires the uninformative signal and optimally reports $\hat{x}_i=\textbf{1}(\mu=1)$.
	Since the vote share $\tilde{n}_{i,n}$ does not depend on the state $\omega$, acquiring the revealing signal is not a profitable deviation.
	Therefore, this is an equilibrium.
	In this equilibrium, the chosen alternative is $x=\textbf{1}(\mu>0)$.
	As we just argued, this is efficient.

In both cases, an efficient equilibrium exists.
All that remains is to verify that, in all equilibria, welfare is at least as high as if the planner had chosen the ex-ante optimal alternative.
In a given equilibrium, let 
$N$ be the set of agents that acquire the revealing signal.
All agents $i\in N$ report $\hat{x}_i=\textbf{1}(v_i=1)$.
If this group of agents prevails in the vote, then the ex-post optimal alternative is chosen with high probability.
By Lemma \ref{L5}, with high probability, every agent $i\notin N$ reports $\hat{x}_i=\textbf{1}(\mu>0)$.
If this group of agents prevails in the vote, then the ex-ante optimal alternative is chosen.
Either way, the chosen alternative is at least as good as the ex-ante optimal alternative.

	\subsection{Proof of Lemma \ref{L2}}\label{Ap3-L2}

We prove a slightly generalized result.
Let $Z_{i,n}\in[z_L,z_H]$ be a random variable, where $\epsilon_{i,n}$ and $Z_{i,n}$ are independent conditional on the state $\omega$.
For Lemma \ref{L2}, let $Z_{i,n}=\tilde{v}_{i,n}$.

We begin by introducing notation.
First, let $\bar{c}<\infty$ be an upper bound on the cost of combined signal $\bar{s}_i$ (Definition \ref{D100}).
Second, define an analog $\tilde{\rho}_n$ to predictiveness $\rho$ that asks how well a given signal $s_i$ predicts $\ex{Z_{i,n}\mid\omega}$.
Formally
\[
\tilde{\rho}_n\left(s_i\right)=\lambda_n\ex{\max_{b_i}\ex[i,n]{\textnormal{SR}^{\textnormal{CRPS}}\left(b_i,Z_{i,n}\right)}}
\]

\begin{claim}\label{Cx9}
	Fix a strategy profile where agent $i$ acquires a signal $s_i$ where $c(\theta_i,s_i)<\infty$ for all types $\theta_i$ in the support of the joint distribution $F$.
	If the covariance between her error $\epsilon_{i,n}$ and the random variable $Z_{i,n}$ is positive, she can increase her expected score by acquiring the combined signal $\bar{s}$.
	That is,
	\[
	\tilde{\rho}_n\left(\bar{s}_i\right)\geq\tilde{\rho}_n\left(s_i\right)+\frac{\lambda_n}{(v_H-v_L)^3}\left|\cov{\epsilon_{i,n},Z_{i,n}}\right|^2
	\]
\end{claim}

Having established conditions under which the agent could increase her expected score by acquiring the combined signal $\bar{s}$, we can limit the extent to which those conditions hold in an equilibrium.
More precisely, we argue that the covariance between the error $\epsilon_{i,n}$ and the random variable $Z_{i,n}$ cannot be too large in equilibrium.

\begin{claim}\label{Cx4}
	In any equilibrium, the covariance between the error $\epsilon_{i,n}$ and the conditional expectation $\ex{Z_{i,n}\mid\omega}$ is bounded.
	That is,
	\[
	\left|\cov{\epsilon_{i,n},Z_{i,n}}\right|=\left|\cov{\epsilon_{i,n},\ex{Z_{i,n}\mid\omega}}\right|\leq\sqrt{\frac{(v_H-v_L)^3\bar{c}}{\lambda_n}}
	\]
\end{claim}

    \subsection{Proof of Corollary \ref{L8}}

Fix an instance $I\in\mathcal{I}$ and an equilibrium $(\textbf{s}^n,\textbf{m}^n)$.
Let agent $i$ send message $m_i=\textbf{m}^n_i(\textbf{s}^n_i,\textbf{s}^n_i(\omega,\theta_i))$.
The planner's expected square loss is
\begin{align}
	\ex{\left(\frac{1}{n}\sum_{i=1}^n\hat{v}_i-v_j\right)^2}
	&=\ex{v_j^2}
	+\ex{\left(\frac{1}{n}\sum_{i=1}^n\hat{v}_i\right)^2}
	-2\ex{v_j\cdot\frac{1}{n}\sum_{i=1}^n\hat{v}_i}
	\label{T2-Cl1-E1}
\end{align}
Focus on the (simplified) third term, i.e.,
\begin{align}
	&-2\ex{v_j\cdot\hat{v}_i}\\
	&=-2\ex{\hat{v}_j\cdot\hat{v}_i}
	-2\ex{\epsilon_j\cdot\hat{v}_i}
	\tag{defn. of $\epsilon_j$}\\
	&=-2\ex{\ex{\hat{v}_j\mid\omega}^2}
	-2\ex{\ex{\epsilon_j\mid\omega}\cdot\ex{\hat{v}_i\mid\omega}}
	\tag{LIE and conditional i.i.d.}\\
	&=-2\ex{\ex{\hat{v}_j\mid\omega}^2}
	-2\ex{\epsilon_j\cdot\ex{\hat{v}_i\mid\omega}}
	\tag{LIE}\\
	&=-2\ex{\ex{\hat{v}_j\mid\omega}^2}
	+o(1)
	\tag{by Lemma \ref{L2}}
\end{align}
Returning to the earlier expression \eqref{T2-Cl1-E1}, we find that
\begin{align}
	\eqref{T2-Cl1-E1}
	&=\ex{v_j^2}
	+\ex{\left(\frac{1}{n}\sum_{i=1}^n\hat{v}_i\right)^2}
	-2\ex{\ex{\hat{v}_i\mid\omega}^2}
	+o(1)
\end{align}
Since $\hat{v}_i$ depends on the equilibrium $(\textbf{s}^n,\textbf{m}^n)$, its distribution varies with $n$.
The weak law of large numbers for triangular arrays ensures that $\frac{1}{n}\sum_{i=1}^n\hat{v}_i\to_p\ex{\hat{v}_i\mid\omega}$.
Combining this with the continuous mapping theorem, we find that
\begin{align}
	\eqref{T2-Cl1-E1}
	&=\ex{v_j^2}
	-\ex{\ex{\hat{v}_i\mid\omega}^2}
	+o(1)\notag\\
	&=\ex{v_j^2}
	-\ex{v_i}^2
	+\ex{\hat{v}_i}^2
	-\ex{\ex{\hat{v}_i\mid\omega}^2}
	+o(1)
	\tag{since $\ex{\epsilon_i}=0$}\\
	&=\ex{v_j^2}
	-\ex{v_i}^2
	+\ex{\ex{\hat{v}_i\mid\omega}}^2
	-\ex{\ex{\hat{v}_i\mid\omega}^2}
	+o(1)
	\tag{LIE}\\
	&=\var{v_j}
	-\var{\ex{\hat{v}_i\mid\omega}}
	+o(1)
	\tag{defn. of $\var{\cdot}$}\\
	&\leq\var{v_i}+o(1)
	\tag{since $\var{\cdot}\geq 0$}
\end{align}
The last term is precisely the ex-ante benchmark.
	
	\subsection{Proof of Lemma \ref{L3}}\label{Ap3-L3}

This proof consists of three claims.
It uses notation from the proof of Lemma \ref{L2}.
In order to save space, we prove a slightly stronger result that refers to notation from the proof of Theorem \ref{T3} (both the special and general cases).
Lemma \ref{L3} corresponds to the case where $\phi_{1,n}=\phi_{2,n}=0$, $n'=n-1$, $v_{0,n}=0$, and $\delta_n=0$.

Let agent $i$ acquire signal $s_i$ and report values $\hat{v}_i=\hat{\textbf{v}}_i\left(s_i,s_i(\omega,\theta_i)\right)$.
We assume that
\begin{equation}\label{E15}
\pr{\max_{i=1,\ldots,n}\left|\hat{v}_i-\ex{v_i\mid s_i,s_i(\omega,\theta_i)}\right|>\phi_{1,n}}\leq\phi_{2,n}
\end{equation}
We define
$
\tilde{v}_{i,n}=\tfrac{1}{n'}\left(v_{0,n}+\sum_{j\neq i}\hat{v}_{j,n}\right).
$
We begin by showing that the cross-covariance between the error $\epsilon_{i,n}$ and the conditional expectation $\ex{v_j\mid\omega}$ is zero.
This follows by construction of the auxiliary instance and an argument similar to Lemma \ref{L2}.

\begin{claim}\label{Cx5}
	Consider a strategy profile in the auxiliary instance where agent $i$ acquires a signal $s_i$ with finite expected cost
	$
	\ex{\tilde{c}(\theta_i,s_i)}<\infty
	$.
	Then the covariance between the error $\epsilon_{i,n}$ and the conditional expected values $\ex{v_j\mid\omega}$ is zero, i.e.,
	$\left\|\cov{\epsilon_{i,n},\ex{v_j\mid\omega}}\right\|=0.$
\end{claim}

With Claims \ref{Cx4} and \ref{Cx5} in hand, we can argue that the variance of the conditional expectation of the error
cannot be too large.

\begin{claim}\label{Cx6}
	In any equilibrium of the auxilliary instance, the variance of the conditional expected error 
	is bounded, i.e.,
	\begin{equation*}
	\var{\ex{\frac{1}{n}\sum_{j=1}^n\epsilon_{j,n}\mid\omega}}=o(1)
	\quad\text{and}\quad
	\var{\ex{\frac{1}{n-1}\sum_{j\neq i}\epsilon_{j,n}\mid\omega}}=o(1)
	\end{equation*}
\end{claim}

We use Claim \ref{Cx6} to argue that the average error concentrates around zero.

\begin{claim}\label{Cx10}
	For agent $i$, the average reported value $\tilde{v}_{i,n}$ concentrates around the conditional expectation $\ex{v_j\mid\omega}$.
	That is, $\lim_{n\to\infty}\pr{\left|\tilde{v}_{i,n}-\ex{v_j\mid\omega}\right|\geq t}=0$ $\forall t > 0$.
\end{claim}

Finally, let the betting stakes $\lambda_n\to\infty$ as the sample size $n\to\infty$.
Claim \ref{Cx10} implies
$
\tilde{v}_{i,n}\to_p\ex{v_j\mid\omega}
$.
It follows from this and the continuous mapping theorem that
\[
\frac{1}{n'+1}\left(v_{0,n}+\sum_{j=1}^n\hat{v}_{j,n}\right)
=\frac{\hat{v}_{i,n}}{n'+1}+\left(\frac{n'}{n'+1}\right)\tilde{v}_{i,n}\to_p\ex{v_j\mid\omega}
\]
	
	\subsection{Proof of Lemma \ref{L4}}

This proof has four parts. It uses notation from Appendices \ref{Ap3-L2} and \ref{Ap3-L3}.

First, consider any signal $s_i$ where, for any technology $\theta_i$, its cost in the auxiliary instance is different from its cost in the original instance.
Equivalently, consider any signal $s_i$ where
$
\rho\left(\bar{s}_i\right)-\rho\left(s_i\right)>0
$.
Since $\mathcal{S}$ is finite, we can define a gap $\delta>0$ where
\begin{equation}\label{E18}
	\delta:=\min_{s_i\text{ s.t. }\rho\left(s_i\right)<\rho\left(\bar{s}_i\right)}
	\rho\left(\bar{s}_i\right)-\rho\left(s_i\right)
\end{equation}

Second, we consider a sequence of moments and show that it converges.
For any random variable $Z$, define a moment that depends on additional parameter $z\in\reals$, i.e.,
\[
f(Z,z)=\ex{\left(\pr[i]{Z\leq z}-\textbf{1}(Z\leq z)\right)^2}
\]

\begin{claim}\label{Cx11}
	The function $f(\tilde{v}_{i,n},\cdot)\to f(\ex{v_j\mid\omega},\cdot)$ almost everywhere.
\end{claim}

Third, we characterize the limiting behavior of $\tilde{\rho}_n(s_i)$ as $n\to\infty$.
\begin{align}
	\lim_{n\to\infty}\lambda_n^{-1}\tilde{\rho}_n(s_i)
	&=-\sum_{x=1}^k\lim_{n\to\infty}\int_0^{\bar{v}}f(\tilde{v}_{i,n},z)dz\tag{definition of $\tilde{\rho}_n$}\\
	&=-\sum_{x=1}^k\int_0^{\bar{v}}f(\ex{v_j\mid\omega},z)dz\tag{bounded convergence thm. \& Claim \ref{Cx11}}\\
	&=\rho(s_i)\tag{definition of $\rho$}
\end{align}
It follows that there exists a threshold $N(s_i)$ such that, for any $n\geq N(s_i)$,
\begin{equation}\label{E21}
	\left|\rho(s_i)-\lambda_n^{-1}\tilde{\rho}_n(s_i)\right|<\frac{\delta}{4}
\end{equation}
Let $N=\max_{s_i\in\mathcal{S}}N(s_i)$.
This quantity exists since $\mathcal{S}$ is a finite set.

Fourth and finally, we show that agent $i$ does not have a profitable deviation when $n\geq N$.
Assume that all other agents continue to play according to an equilibrium of the auxiliary instance.
Let $s_i$ be the signal that she acquires in the equilibrium of the auxiliary instance.
If agent $i$ continues to acquire the same signal $s_i$, then her best response $\hat{v}_i$ does not change.
Suppose that agent $i$ acquires a different signal $s'_i$.
We only need to consider deviations to signals $s'_i$ that were not available in the auxiliary instance, which implies
\begin{equation}\label{E19}
	\rho\left(\bar{s}_i\right)-\rho\left(s'_i\right)\geq\delta
\end{equation}
where $\delta>0$ was defined in equation \eqref{E18}.
Since $s_i$ was chosen in equilibrium of the auxiliary instance, it must have finite cost.
By definition of $\tilde{c}$ in \eqref{T2-E1}, this implies
\begin{equation}\label{E20}
	\rho\left(\bar{s}_i\right)
	=\rho\left(s_i\right)
\end{equation}
Combining the last few inequalities, we find
\begin{align}
	\tilde{\rho}_n\left(s_i\right)-\tilde{\rho}_n\left(s'_i\right)
	&>\lambda_n\rho\left(s_i\right)-\rho\left(s'_i\right)-\frac{\delta\lambda_n}{2}\tag{equation \eqref{E21}}\\
	&=\lambda_n\rho\left(\bar{s}_i\right)-\rho\left(s'_i\right)-\frac{\delta\lambda_n}{2}\tag{equation \eqref{E20}}\\
	&\geq\frac{\delta\lambda_n}{2}\tag{inequality \eqref{E19}}
\end{align}
This is the expected loss in transfers from the scoring rule $\textnormal{SR}^{\textnormal{CRPS}}$, which is a lower bound for the expected loss from the scoring rule $\textnormal{SR}$.
The maximum gain from deviating is at most the cost $\bar{c}$ of the combined signal, plus the maximum payoff $v_H-v_L$ from the BDM mechanism (Lemma \ref{L4}) or the VCG mechanism (Lemma \ref{L12}).
Therefore, a sufficient condition for the deviation to not be profitable is
$\frac{\delta\lambda_n}{2}\geq\bar{c}+v_H-v_L.$
Since $\lambda_n\to\infty$, there is some threshold $N'$ such that this condition holds for all $n\geq N'$,
so no profitable deviations exist whenever $n\geq\max\{N,N'\}$.
	
	\subsection{Proof of Lemma \ref{L1}}

Let $(\textbf{s},\textbf{m})$ be an equilibrium for instance $I$.
Let $s_i=\textbf{s}_i$ and 
$
m_i=\textbf{m}_i\left(\textbf{s}_i,\textbf{s}_i(\omega,\theta_i)\right)
$.
We claim that the vector $(v_i,\theta_i,s_i,m_i)$ is independent across agents $i$.
This follows from three observations.
First, the values $v_i$ and types $\theta_i$ are independent across agents $i$, since the state space is a singleton.
Third, $s_i$ and $m_i$ are random functions of $(\omega,\theta_i)$. The arguments are independent across agents and the mixing in mixed strategies is independent across agents.

Consider the distribution of acquired signals $s_i$.
We claim that agent $i$ does not learn her value $v_i$.
We begin by making some observations.
The good signal is only useful to agent $i$ insofar as it affects her choice of which message $m_i$ to send.
The cost of the revealing signal is $\tilde{c}>0$ and the agent will only acquire it if she benefits from the information $v_i$ that it provides.
However, this information does not affect her beliefs about other agents' messages, since $v_i$ and $m_{-i}$ are independent.

Suppose agent $i$ acquires the revealing signal
at a cost of $\bar{c}$.
To find a profitable deviation, consider 
her expected utility 
$U_i\left(\textbf{s},\textbf{m},\textbf{x},\textbf{t}\right)
=\ex{
	u_i(v_i,\textbf{x}_i(m),\textbf{t}_i(m),s_i,\theta_i)
}.$
Recall that her value is $v_i\in\{a,b\}$.
Suppose agent $i$ sends message $m^a_i$ if she learns $v_i=a$ and $m^b_i$ if she learns $v_i=b$.
Then
\begin{align}
	&\ex{u_i(a,\textbf{x}_i(m_i^b,m_{-i},\textbf{t}_i(m_i^b,m_{-i}),s_i,\theta_i)}\notag\\
	&\geq\ex{u_i(b,\textbf{x}_i(m_i^b,m_{-i},\textbf{t}_i(m_i^b,m_{-i}),s_i,\theta_i)}-|b-a|\label{L1-E1}\\
	&\geq\ex{u_i(b,\textbf{x}_i(m_i^a,m_{-i},\textbf{t}_i(m_i^a,m_{-i}),s_i,\theta_i)}-|b-a|\label{L1-E2}\\
	&\geq\ex{u_i(a,\textbf{x}_i(m_i^a,m_{-i},\textbf{t}_i(m_i^a,m_{-i}),s_i,\theta_i)}-2|b-a|\label{L1-E3}
\end{align}
where inequalities \eqref{L1-E1} and \eqref{L1-E3} hold by inspection of the utility function, and inequality \eqref{L1-E2} holds because agent $i$ prefers to send message $m^b_i$ when $v_i=b$.
Now, consider the deviation where agent $i$ does not acquire the revealing signal and always sends message $m_i^b$.
She reduces her costs by $\bar{c}$ and, by inequalities \eqref{L1-E1}-\eqref{L1-E3}, reduces her gains by at most $2|b-a|$.
This is a profitable deviation since $\bar{c}>2|b-a|$.

Next, we construct a strategy profile $(\textbf{s}',\textbf{m}')$ and show that it is an equilibrium for instance $I'$.
Let agent $i$ acquire the uninformative signal and send the same (random) message $m_i$ that she would have sent in instance $I$.
This means that the distribution of message profiles $m'$ given joint distribution $F'$, where
$
m'_i=\textbf{m}'_i\left(\textbf{s}'_i,\textbf{s}'_i(\omega,\theta_i)\right)
$,
is the same as the distribution of message profiles $m$ given joint distribution $F$.

To see that $(\textbf{s}',\textbf{m}')$ is an equilibrium, first consider information acquisition.
When the joint distribution is $F'$, agent $i$ knows both her type and the state perfectly.
Therefore, there is no benefit from deviating to the revealing signal.

We must also show that agent $i$ does not want to deviate from message rule $\textbf{m}'_i$.
Suppose that agent $i$ sends message $\tilde{m}_i$.
Her expected utility in instance $I'$ is
\begin{align*}
	&\ex[F']{\frac{a+b}{2}\cdot\textbf{x}_i(\tilde{m}_i,m'_{-i})+\textbf{t}_i(\tilde{m}_i,m'_{-i})}\\
	&=\ex[F]{\frac{a+b}{2}\cdot\textbf{x}_i(\tilde{m}_i,m_{-i})+\textbf{t}_i(\tilde{m}_i,m_{-i})}\tag{by construction}\\
	&=\ex[F]{\ex[F]{v_i}\cdot\textbf{x}_i(\tilde{m}_i,m_{-i})+\textbf{t}_i(\tilde{m}_i,m_{-i})}\tag{by inspection of $F$}\\
	&=\ex[F]{v_i\cdot\textbf{x}_i(\tilde{m}_i,m_{-i})+\textbf{t}_i(\tilde{m}_i,m_{-i})}\tag{by law of iterated expectations (LIE)}
\end{align*}
This objective is identical to her expected utility in instance $I$.
Therefore, since $(\textbf{s},\textbf{m})$ is an equilibrium in instance $I$, the message $m_i$ must (almost surely) maximize her expected utility in instance $I'$.
Therefore, $(\textbf{s}',\textbf{m}')$ is an equilibrium.
	
	\subsection{Proof of Lemma \ref{L5}}

Let $\bar{c}(\theta_i)$ be the maximum cost, given type $\theta_i$ and across all base signals $s_i\in\mathcal{S}$, of the signal $s'_i$ that combines $s_i$ with a revealing signal.
Assumptions \ref{A2}, \ref{A5} and \ref{A6} ensure that this maximum exists and is finite.
Let $\bar{c}=\max_{\theta_i\in\supp{F_{\Theta}}}\bar{c}(\theta_i)$, where $F_{\Theta}$ is the unconditional marginal distribution of types $\theta_i$.

Agent $i$ reports the alternative $\hat{x}_i$ that maximizes her expected value conditional on being pivotal.
Claim \ref{T3-Cl1} says that agent $i$ conditioning on being pivotal does not meaningfully affect her expected value $v_i$.

\begin{claim}\label{T3-Cl1}
	Fix a constant $t>0$.
	Then
	\[
	\pr{\left|\ex[i,n]{v_i\mid Q_{i,n}}-\ex[i,n]{v_i}\right|\geq t}
	\leq\frac{1}{t}\cdot\sqrt{\frac{|\mathcal{V}|\cdot|v_H-v_L|\cdot\bar{c}}{(\delta_n/n)^2\cdot\lambda_n}}
	\]
\end{claim}

Suppose $x_i$ is the truthful report.
Then, with high probability,
\begin{align}
	\ex[i,n]{v_i x_i}
	&\geq\ex[i,n]{v_i \hat{x}_i}
	\tag{since $x_i$ is truthful}\\
	&\geq\ex[i,n]{v_i \hat{x}_i\mid Q_{i,n}}-t
	\tag{Claim \ref{T3-Cl1}}\\
	&\geq\ex[i,n]{v_i x_i\mid Q_{i,n}}-t
	\tag{since $\hat{x}_i$ is optimal}\\
	&\geq\ex[i,n]{v_i x_i}-2t
	\tag{Claim \ref{T3-Cl1}}
\end{align}
This sequence of inequalities gives us upper and lower bounds for $\ex[i,n]{v_i\hat{x}_i}$.
Combining these bounds with the union bound ensures that
\[
\pr{\max_{i=1,\ldots,n}\left|
	\max_{x\in\mathcal{X}}\ex[i,n]{v_ix}-\ex[i,n]{v_i\hat{x}_i}
	\right|\geq 2t}
\leq\frac{n}{t}\cdot\sqrt{\frac{|\mathcal{V}|\cdot|v_H-v_L|\cdot\bar{c}}{(\delta_n/n)^2\cdot\lambda_n}}
\]
	
	\printbibliography[segment=0]
	\newrefsegment
	
	\newpage
	
	
	\section{Supplemental Appendix}\label{Ap3}
	
	\subsection{Proof of Theorem \ref{T3} (General Case)}\label{Ap3-T3}



To relax Assumption \ref{A4}, we introduce the \emph{VCG-with-betting mechanism}.
This is analogous to BDM-with-betting and majority-rule-with-betting.
The difference is that, in the second stage, agents report their willingness to pay to the VCG mechanism.
As with majority-rule-with-betting, we also need to add some noise to the mechanism to avoid issues that arise when agents condition on low-probability events.

The VCG-with-betting mechanism has four tuning parameters.
First, there is proper scoring rule $\textnormal{SR}$, which evaluates the accuracy of agent $i$'s reported beliefs $\hat{b}_i$ over the average reported value $\tilde{v}_i$.
Second, there is the scaling parameter $\lambda$ that controls how large the betting stakes are.
Third, there is a probability $\delta_n$ that some agent $i$ is chosen as dictator.
Fourth, there is a randomized \emph{bias term} $v_{0,n}\in\reals$.
This biases the mechanism towards alternatives $x$ where $v_{0x}$ is large.
Let the average reported value $\tilde{v}_i$ include the bias term, i.e.,
$
\tilde{v}_i=\frac{1}{n}\left(v_{0,n}+\sum_{j\neq i}\hat{v}_j\right)
$.
The agents do not know the realization of the bias term.

\begin{defn}
	The \emph{VCG-with-betting mechanism} $\left(\textbf{x},\textbf{t}\right)$ features a proper scoring rule $\textnormal{SR}$, scaling parameter $\lambda$, probability $\delta$, and randomized bias term $v_{0,n}$.
	Each agent $i$ sends a message
	$m_i=\left(\hat{v}_i,\hat{b}_i\right)
	$
	that consists of a reported value $\hat{v}_i$ and a reported belief $\hat{b}_i$.
	\begin{enumerate}
		\item With probability $1-\delta$, there is no dictator $(D=0)$. The planner selects the alternative that maximizes the average reported value plus the bias term, i.e.,
		\[\textbf{x}\left(\hat{v},\hat{b}\right)\in\arg\max_{x\in\mathcal{X}}\left(v_{0x}+\sum_{i=1}^n\hat{v}_{ix}\right)
		\]
		Each agent $i$ is paid according to the VCG mechanism's transfer rule and the proper scoring rule, i.e.,
		\[
		\textbf{t}_i\left(\hat{v},\hat{b}\right)=v_{0\textbf{x}\left(\hat{v},\hat{b}\right)}+\sum_{j\neq i}\hat{v}_{j\textbf{x}\left(\hat{v},\hat{b}\right)}
		-\max_{x\in X}\left(v_{0x}+\sum_{j\neq i}\hat{v}_{jx}\right)
		+\lambda\cdot\textnormal{SR}\left(\hat{b}_i,\tilde{v}_i\right)
		\]
		\item With probability $\delta_n$, the planner selects a random agent $D\sim\textsc{Uniform}(1,\ldots,n)$ as dictator.
		This agent is paid according to a BDM mechanism with random price $p\sim\textsc{Uniform}\left([v_L,v_H]\right)$ and the proper scoring rule, i.e.,
		\[
		\textbf{t}_D\left(\hat{v},\hat{b}\right)=\lambda\cdot\textnormal{SR}\left(\hat{b}_D,\tilde{v}_D\right)-p\cdot\textbf{\emph{1}}(\hat{v}_D\geq p)
		\]
		The other agents $j\neq D$ are according to the proper scoring rule, i.e.,
		$
		\textbf{t}_j\left(\hat{v},\hat{b}\right)=\lambda\cdot\textnormal{SR}\left(\hat{b}_j,\tilde{v}_j\right)
		$.
		Finally, the allocation is
		$
		\textbf{x}_i\left(\hat{v},\hat{b}\right)=\textbf{\emph{1}}(\hat{v}_D\geq p)
		$.
	\end{enumerate}
\end{defn}

\subsubsection{Sequence of VCG-with-Betting Mechanisms}

We construct a sequence $\left(\textbf{x}^n,\textbf{t}^n\right)$ of VCG-with-betting mechanisms.

First, we introduce useful notation.
Let $\tilde{p}_{i,n}$ be the effective price of switching from alternative 0 to 1, i.e.,
\[
\tilde{p}_{i,n}=\textbf{1}(D=i)\cdot p+\textbf{1}(D=0)\cdot(-n\tilde{v}_{i,n})+\textbf{1}(D\notin\{0,i\})\cdot\infty
\]
Let $Q_{i,n}$ indicate the event that agent $i$ is potentially pivotal, i.e.,
$
v_L\leq\tilde{p}_{i,n}\leq v_H
$.
Let $q_{i,n}$ be the probability that agent $i$ is potentially pivotal, i.e.,
$
q_{i,n}=\pr[i,n]{Q_{i,n}}
$.
Let $\hat{q}_{i,n}$ be that probability evaluated according to reported belief $\hat{b}_{i,n}$.

Second, we specify the scoring rule $\textnormal{SR}$.
Let $\textnormal{SR}$ be the sum of the quadratic scoring rule and two continuous ranked probability scores.
That is,
\[
\textnormal{SR}\left(\hat{b}_i,\tilde{v}_{i,n}\right)
=\textnormal{SR}_n^{Q}\left(\hat{b}_{i,n},\tilde{v}_{i,n}\right)
+\textnormal{SR}^{CRPS}\left(\hat{b}_{i,n},\tilde{v}_{i,n}\right)
+\textnormal{SR}^{CRPS}\left(\hat{b}_{i,n},\textbf{1}(\tilde{v}_{i,n}\geq 0)\right)
\]
where $\textnormal{SR}^{CRPS}$ is defined as in Appendix \ref{Ap2-T2}, and $\textnormal{SR}^{Q}$ is defined as in Appendix \ref{Ap2-T3} (replacing $\tilde{n}_{i,n}$ with $\tilde{v}_{i,n}$).
Each of these scoring rules is convenient in different parts of the proof.
Since each one is proper, the sum $\textnormal{SR}$ is also proper.

Third, we specify the distribution of the bias term.
Let $v_{0,n}\sim\textsc{Laplace}\left(0,\beta_n\right)$.
This ensures that agent $i$'s expected value conditional on the vote margin $\tilde{v}_{i,n}$ is a smooth function of $\tilde{v}_{i,n}$.

It is important that $\delta_n\to0$, $\beta_n\to\infty$, $\beta_n/n\to0$, and $\lambda_n\to\infty$.
The proof of Lemma \ref{L9} describes trade-offs between these parameters in more detail.

\subsubsection{Convergence to Efficiency}

We begin by formalizing what it means for VCG-with-betting to be nearly truthful.
Essentially, each agent's reported values should be probably approximately equal to her expected values (conditional on any information she has acquired).

\begin{defn}\label{D14}
	The VCG-with-betting mechanism $\left(\textbf{x}^n,\textbf{t}^n\right)$ is \emph{$(\phi_1,\phi_2)$-truthful} if the following holds for every instance $I\in\mathcal{I}$ and equilibrium $(\textbf{s},\hat{\textbf{v}},\hat{\textbf{b}})$.
	Let agent $i$ acquire signal $s_i=\textbf{s}_i(\theta_i,c_i)$ and report values $\hat{v}_i=\hat{\textbf{v}}_i\left(s_i,s_i(\omega,\theta_i)\right)$.
	Then
	\[
	\pr{\max_{i=1,\ldots,n}\left|\hat{v}_i-\ex{v_i\mid s_i,s_i(\omega,\theta_i)}\right|>\phi_1}\leq\phi_2
	\]
\end{defn}

Next, we verify that our VCG-with-betting mechanism is nearly-truthful.
This is the key step in proving Theorem \ref{T3}.

\begin{lemma}\label{L9}
	The mechanism $\left(\textbf{x}^n,\textbf{t}^n\right)$ is $(\phi_{1n},\phi_{2n})$-truthful, where $\phi_{1n}\to 0$ and $\phi_{2n}\to 0$.
\end{lemma}

The next three lemmas are analogs Lemmas \ref{L2}-\ref{L4}.
The main differences are that, rather than require truthful reporting, they only rely on nearly-truthful reporting.

\begin{lemma}\label{L10}
	In any equilibrium of the mechanism $(\textbf{x}_n,\textbf{t}_n)$,
	\[
	\left|\cov{\epsilon_{i,n},\ex{\tilde{v}_{i,n}\mid\omega}}\right|=O\left(\frac{1}{\sqrt{\lambda_n}}\right)
	\quad\text{and}\quad
	\left|\cov{\epsilon_{i,n},\textnormal{\textbf{1}}(\tilde{v}_{i,n}\geq 0)}\right|=O\left(\frac{1}{\sqrt{\lambda_n}}\right)
	\]
\end{lemma}
\begin{proof}
	Same as the proof of Lemmas \ref{L2} and \ref{L6}, where $Z_{i,n}$ is $\tilde{v}_{i,n}$ and $\textbf{1}(\tilde{v}_{i,n}\geq 0)$.
\end{proof}

\begin{lemma}\label{L11}
	Fix the auxiliary instance.
	For any sequence of equilibria $(\textbf{s}^n,\textbf{m}^n)$,
	\begin{equation}\label{E35}
	\frac{1}{n+1}\left(v_{0,n}+\sum_{j=1}^n\hat{v}_{j,n}\right)\to_p\ex{v_j\mid\omega}
	\end{equation}
\end{lemma}
\begin{proof}
	Same as the proof of Lemma \ref{L3}, for $n'=n$ and suitable $\phi_{1,n}$, $\phi_{2,n}$, $v_{0,n}$, $\delta_n$.
\end{proof}

\begin{lemma}\label{L12}
	There exists a constant $N$ such that, for any sample size $n\geq N$, every equilibrium of the auxiliary instance is also an equilibrium of the original instance.
\end{lemma}
\begin{proof}
	Same as the proof of Lemma \ref{L4}.
\end{proof}

Lemmas \ref{L11} and \ref{L12} guarantee the existence of an equilibrium $\left(\textbf{s}^n,\hat{\textbf{v}}^n,\hat{\textbf{b}}^n\right)$ where condition \eqref{E35} holds.
It follows from the argmax continuous mapping theorem \parencite{KP90} that the alternative $\textbf{x}^n\left(\hat{\textbf{v}}^n,\hat{\textbf{b}}^n\right)$ converges into the set of welfare-maximizing alternatives.
Finally, Lemma \ref{L14} confirms that we satisfy the ex-ante benchmark in worst-case equilibria, using on Lemmas \ref{L9} and \ref{L10}.

\begin{lemma}\label{L14}
	Mechanisms $(\textbf{x}^n,\textbf{t}^n,\textbf{e}^n)$ guarantee the ex-ante benchmark in all equilibria.
\end{lemma}

    \subsection{Proof of Lemma \ref{L9}}

Define $\bar{c}$ as in Lemma \ref{L5}.
Claim \ref{T3-Cl5} says that agent $i$ conditioning on being potentially pivotal does not meaningfully affect her expected value $v_i$.

\begin{claim}\label{T3-Cl5}
	Fix a constant $t>0$.
	Then
	\[
	\pr{\left|\ex[i,n]{v_i\mid Q_{i,n}}-\ex[i,n]{v_i}\right|\geq t}
	\leq\frac{1}{t}\cdot\sqrt{\frac{|\mathcal{V}|\cdot|v_H-v_L|\cdot\bar{c}}{(\delta_n/n)^2\cdot\lambda_n}}
	\]
\end{claim}

Claim \ref{T3-Cl7} is a useful fact that we use to prove Claims \ref{T3-Cl6} and \ref{T3-Cl9}. It refers to $G_n$, which is the cumulative distribution function of the bias term $v_{0,n}$.

\begin{claim}\label{T3-Cl7}
	Fix constants $a,b$ where $0\leq b-a\leq \beta_n$.
	Then
	\[
	\frac{1}{2}\left(e^{-\frac{\left|b\right|}{\beta_n}}\right)\left(\frac{b-a}{\beta_n}-e\left(\frac{b-a}{\beta_n}\right)^2\right)
	\leq G_n\left(b\right)-G_n\left(a\right)
	\leq\frac{1}{2}\left(e^{-\frac{\left|b\right|}{\beta_n}}\right)\left(\frac{b-a}{\beta_n}+e\left(\frac{b-a}{\beta_n}\right)^2\right)
	\]
\end{claim}

Claims \ref{T3-Cl6} and \ref{T3-Cl9} say that agent $i$'s expected value is roughly the same regardless of whether (i) she conditions on being potentially pivotal or (ii) she conditions on the price $\tilde{p}_{i,n}$ being below some threshold $z$.
Claim \ref{T3-Cl6} conditions on agent $i$ not being dictator, and Claim \ref{T3-Cl9} drops this requirement.
They refer to the following sequence:
\[
\psi_n=\left(\frac{\beta_n+e(v_H-v_L)}{\beta_n-e(v_H-v_L)}\right)^2-\left(\frac{\beta_n-e(v_H-v_L)}{\beta_n+e(v_H-v_L)}\right)^2
\]

\begin{claim}\label{T3-Cl6}
	Fix constants $z,z'$ where $v_L\leq z<z'\leq v_H$.
	Then
	\[
	\left|\ex[i,n]{v_i\mid v_L\leq-n\tilde{v}_{i,n}\leq v_H}-\ex[i,n]{v_i\mid z\leq-n\tilde{v}_{i,n}\leq z'}\right|\leq\psi_n
	\]
\end{claim}

\begin{claim}\label{T3-Cl9}
	Fix constants $z,z'$ where $v_L\leq z<z'\leq v_H$.
	Then
	\[
	\left|\ex[i,n]{v_i\mid Q_{i,n}}
	-\ex[i,n]{v_i\mid z\leq\tilde{p}_{i,n}\leq z'}\right|
	\leq\psi_n+O\left(\frac{1}{\beta_n}+\frac{\beta_n\delta_n}{n}\right)
	\]
\end{claim}

Claim \ref{T3-Cl8} says that agent $i$'s reported value will not be too far from her expected value.
It relies heavily on Claims \ref{T3-Cl5} and \ref{T3-Cl7}.

\begin{claim}\label{T3-Cl8}
	Fix a constant $t>0$.
	Then
	\begin{align*}
		\pr{\left|\hat{v}_{i,n}-\ex[i,n]{v_i}\right|\geq 2t+2\psi_n+O\left(\frac{1}{\beta_n}+\frac{\beta_n\delta_n}{n}\right)}
		\leq\frac{1}{t}\sqrt{\frac{|\mathcal{V}|\cdot|v_H-v_L|\cdot\bar{c}}{(\delta_n/n)^2\cdot\lambda_n}}
	\end{align*}
\end{claim}

Let $t_n=n^{-1/2}$.
Apply the union bound to Claim \ref{T3-Cl8} to show that
\begin{align*}
	\pr{\max_{i=1}^n\left|\hat{v}_{i,n}-\ex[i,n]{v_i}\right|\geq t_n+\psi_n+O\left(\frac{1}{\beta_n}+\frac{\beta_n\delta_n}{n}\right)}
	\leq\frac{n}{t_n}\sqrt{\frac{|\mathcal{V}|\cdot|v_H-v_L|\cdot\bar{c}}{(\delta_n/n)^2\cdot\lambda_n}}
\end{align*}
We set the truthfulness parameters as follows.
First, let
\[
\phi_{1,n}
=t_n+\psi_n+O\left(\frac{1}{\beta_n}+\frac{\beta_n\delta_n}{n}\right)
=n^{-1/2}+O(n^{-1/2})+O\left(n^{-1/2}+n^{-1}\right)
=O\left(n^{-1/2}\right)
\]
Second, let
\[
\phi_{2,n}=\frac{n}{t_n}\sqrt{\frac{|\mathcal{V}|\cdot|v_H-v_L|\cdot\bar{c}}{(\delta_n/n)^2\cdot\lambda_n}}
=n^{3/2}\sqrt{\frac{|\mathcal{V}|\cdot|v_H-v_L|\cdot\bar{c}}{n^{-3}\cdot n^7}}
=O\left(n^{-1/2}\right)
\]
	
	\subsection{Proof of Lemma \ref{L14}}

Condition on the probability $1-\delta_n\to 1$ event that no agent is chosen as dictator.
The difference between the expected welfare given mechanism $(\textbf{x}^n,\textbf{t}^n)$ and expected welfare under the ex-ante optimal alternative is
\begin{align}
	&\ex{
		\ex{v_i\mid\omega}\cdot\textbf{1}\left(v_{0,n}+\sum_{j=1}^n\hat{v}_{j,n}\geq 0\right)-\ex{v_i\mid\omega}\cdot\textbf{1}\left(\ex{v_i}\geq 0\right)
	}\label{T3-E6}\\
	&\geq\ex{
		\left(\frac{1}{n}\sum_{j=1}^nv_j\right)\cdot\textbf{1}\left(v_{0,n}+\sum_{j=1}^n\hat{v}_{jx}\geq 0\right)-\ex{v_i\mid\omega}\cdot\textbf{1}\left(\ex{v_i}\geq 0\right)
	}-o(1)
	\tag{Hoeffding's inequality}
\end{align}
We want to show that $\eqref{T3-E6}\geq-o(1)$.
The key step is the following claim.

\begin{claim}\label{T3-Cl4}
	The following inequality holds in every equilibrium $(\textbf{s}^n,\hat{\textbf{v}}^n,\hat{\textbf{b}}^n)$:
	\begin{align*}
		&\ex{\left(\frac{1}{n}\sum_{j=1}^nv_j\right)\cdot\textnormal{\textbf{1}}\left(v_{0,n}+\sum_{j=1}^n\hat{v}_{j,n}\geq 0\right)}\\
		&\geq\ex{\frac{1}{n}\cdot\left(v_{0,n}+\sum_{j=1}^n\hat{v}_{j,n}\right)\cdot\textnormal{\textbf{1}}\left(v_{0,n}+\sum_{j=1}^n\hat{v}_{j,n}\geq 0\right)}
		-o(1)
	\end{align*}
\end{claim}

Claim \ref{T3-Cl4} relies critically on Lemmas \ref{L9} and \ref{L10}.
We leave the proof to Section \ref{Ap3-T3-Cl4}.
Given this claim, we observe that
\[
\eqref{T3-E6}\geq
\ex{\frac{1}{n}\cdot\left(v_{0,n}+\sum_{j=1}^n\hat{v}_{j,n}\right)\cdot\textbf{1}\left(v_{0,n}+\sum_{j=1}^n\hat{v}_{j,n}\geq 0\right)-\ex{v_i\mid\omega}\cdot\textbf{1}\left(\ex{v_i}\geq 0\right)}-o(1)
\]
At this point, there are two cases to consider.
\begin{enumerate}
	\item Suppose that $\ex{v_i}<0$. Then
	\begin{align}
		\eqref{T3-E6}&\geq\ex{\frac{1}{n}\cdot\left(v_{0,n}+\sum_{j=1}^n\hat{v}_{j,n}\right)\cdot\textbf{1}\left(v_{0,n}+\sum_{j=1}^n\hat{v}_{j,n}\geq 0\right)}-o(1)\notag\\
		&\geq-o(1)
		\tag{since $\ex{X\cdot \textbf{1}(X\geq 0)}\geq 0$}
	\end{align}
	\item Suppose that $\ex{v_i}\geq 0$. Then
	\begin{align}
		\eqref{T3-E6}
		&\geq\ex{\frac{1}{n}\cdot\left(v_{0,n}+\sum_{j=1}^n\hat{v}_{j,n}\right)\cdot\textbf{1}\left(v_{0,n}+\sum_{j=1}^n\hat{v}_{j,n}\geq 0\right)-\ex{v_i\mid\omega}}-o(1)\notag\\
		&\geq\ex{\frac{1}{n}\cdot\left(v_{0,n}+\sum_{j=1}^n\hat{v}_{j,n}\right)-\ex{v_i\mid\omega}}-o(1)
		\tag{$\ex{X\cdot \textbf{1}(X\geq 0)}\geq \ex{X}$}\\
		&=\ex{\frac{1}{n}\sum_{j=1}^n\hat{v}_{j,n}-\ex{v_i\mid\omega}}-o(1)
		\tag{since $\ex{v_{0,n}}=0$}\\
		&\geq\ex{\frac{1}{n}\sum_{j=1}^n\ex[j,n]{v_j}-\ex{v_i\mid\omega}}-O(\phi_{1,n}+(v_H-v_L)\cdot\phi_{2,n})-o(1)
		\tag{Lm. \ref{L9}}\\
		&=\ex{\frac{1}{n}\sum_{j=1}^n\ex[j,n]{v_j}-\ex{v_i\mid\omega}}-o(1)
		\tag{since $\phi_{1,n}\to0$ and $\phi_{2,n}\to 0$}\\
		&=-o(1)\tag{LIE}
	\end{align}
\end{enumerate}
In both cases, the difference $\eqref{T3-E6}$ is at least $-o(1)$.
This completes the proof.

\subsection{Proof of Claim \ref{T3-Cl4}}\label{Ap3-T3-Cl4}

Recall that the scoring rule evaluates, among other things, agent $i$'s ability to predict random variable
$
\textbf{1}(\tilde{v}_{i,n}\geq 0)=\textbf{1}\left(v_{0,n}+\sum_{j\neq i}\hat{v}_{j,n}\geq 0\right)
$.
Consider the following event:
\begin{align}
	&\textbf{1}\left(v_{0,n}+\sum_{j\neq i}\hat{v}_{j,n}\geq 0\right)\neq\textbf{1}\left(v_{0,n}+\sum_{j=1}^n\hat{v}_{j,n}\geq 0\right)
	\notag\\
	&\implies
	v_{0,n}+\sum_{j\neq i}\hat{v}_{j,n}\in\left[-v_H,-v_L\right]\notag\\
	&\implies v_{0,n}\in\left[-v_H-A_{i,n},-v_L-A_{i,n}\right]
	\tag{where $A_{i,n}=\sum_{j\neq i}\hat{v}_{j,n}$}
\end{align}
We want to bound the probability of this event.
Observe that
\begin{align}
	&\pr{v_{0,n}\in\left[-v_H-A_{i,n},-v_L-A_{i,n}\right]}\notag\\
	&=\ex{\pr{v_{0,n}\in\left[-v_H-A_{i,n},-v_L-A_{i,n}\right]\mid A_{i,n}}}
	\tag{LIE}\\
	&\leq\frac{v_H-v_L}{2\beta_n}
	\tag{since $v_{0,n}\sim\textsc{Laplace}(0,\beta_n)$}
\end{align}
It follows from the preceding arguments that
\begin{equation}\label{T3-E2}
	\textbf{1}\left(v_{0,n}+\sum_{j=1}^n\hat{v}_{j,n}\geq 0\right)
	\geq
	\textbf{1}\left(v_{0,n}+\sum_{j\neq i}\hat{v}_{j,n}\geq 0\right)
	-O_p\left(\frac{v_H-v_L}{2\beta_n}\right)
\end{equation}
We will use this observation later on to bound the following expression.
\begin{align}
	&\ex{\left(\frac{1}{n}\sum_{j=1}^nv_j\right)\textbf{1}\left(v_{0,n}+\sum_{j=1}^n\hat{v}_{j,n}\geq 0\right)}=\ex{\left(\frac{1}{n}\sum_{j=1}^n\ex[j,n]{v_j}\right)\textbf{1}\left(v_{0,n}+\sum_{j=1}^n\hat{v}_{j,n}\geq 0\right)}\label{T3-E3}\\
	&\white{=}+\ex{\left(\frac{1}{n}\sum_{j=1}^n\epsilon_{j,n}\right)\cdot\textbf{1}\left(v_{0,n}+\sum_{j=1}^n\hat{v}_{j,n}\geq 0\right)}
	\tag{defn. of $\epsilon_{j,n}$}
\end{align}
There are two terms on the right-hand side of equation \eqref{T3-E3}.
The first one is
\begin{align}
	&\ex{\left(\frac{1}{n}\sum_{j=1}^n\ex[j,n]{v_j}\right)\cdot\textbf{1}\left(v_{0,n}+\sum_{j=1}^n\hat{v}_{j,n}\geq 0\right)}\label{T3-E4}\\
	&\geq\ex{\left(\frac{1}{n}\sum_{j=1}^n\hat{v}_j\right)\cdot\textbf{1}\left(v_{0,n}+\sum_{j=1}^n\hat{v}_{j,n}\geq 0\right)}-O(\phi_{1,n}+(v_H-v_L)\phi_{2,n})
	\tag{by Lemma \ref{L9}}\\
	&\geq\ex{\left(\frac{1}{n}\sum_{j=1}^n\hat{v}_j\right)\cdot\textbf{1}\left(v_{0,n}+\sum_{j=1}^n\hat{v}_{j,n}\geq 0\right)}-o(1)
	\tag{since $\phi_{1,n}\to0$, $\phi_{2,n}\to0$}\\
	&\geq\ex{\left(\frac{v_{0,n}}{n}-\frac{v_{0,n}}{n}\cdot\textbf{1}(v_{0,n}\geq 0)+\frac{1}{n}\sum_{j=1}^n\hat{v}_j\right)\cdot\textbf{1}\left(v_{0,n}+\sum_{j=1}^n\hat{v}_{j,n}\geq 0\right)}-o(1)
	\tag{since $v_{0,n}\cdot\textbf{1}(v_{0,n}\geq 0)\geq v_{0,n}$}\\
	&\geq\ex{\left(\frac{v_{0,n}}{n}+\frac{1}{n}\sum_{j=1}^n\hat{v}_j\right)\cdot\textbf{1}\left(v_{0,n}\sum_{j=1}^n\hat{v}_{j,n}\geq 0\right)
		-\frac{v_{0,n}}{n}\cdot\textbf{1}(v_{0,n}\geq 0)}-o(1)
	\tag{since $v_{0,n}\cdot\textbf{1}(v_{0,n}\geq 0)\geq 0$}\\
	&=\ex{\frac{1}{n}\cdot\left(v_{0,n}+\sum_{j=1}^n\hat{v}_j\right)\cdot\textbf{1}\left(v_{0,n}+\sum_{j=1}^n\hat{v}_{j,n}\geq 0\right)}-\frac{\beta_n}{n}-o(1)
	\tag{dist. of $v_{0,n}$}\\
	&=\ex{\frac{1}{n}\cdot\left(v_{0,n}+\sum_{j=1}^n\hat{v}_j\right)\cdot\textbf{1}\left(v_{0,n}+\sum_{j=1}^n\hat{v}_{j,n}\geq 0\right)}-o(1)
	\tag{since $\beta_n/n\to0$}
\end{align}
The second term on the right-hand side of equation \eqref{T3-E3} is
\begin{align}
	&\ex{\left(\frac{1}{n}\sum_{j=1}^n\epsilon_{j,n}\right)\cdot\textbf{1}\left(v_{0,n}+\sum_{j=1}^n\hat{v}_{j,n}\geq 0\right)}\label{T3-E5}\\
	&\geq
	\ex{\left(\frac{1}{n}\sum_{j=1}^n\epsilon_{j,n}\right)\cdot\textbf{1}\left(v_{0,n}+\sum_{k\neq j}\hat{v}_{k,n}\geq 0\right)}
	-O\left(\frac{(v_H-v_L)^2}{2\beta_n}\right)
	\tag{inequality \eqref{T3-E2}}\\
	&\geq
	-O\left(\lambda_n^{-1/2}\right)
	-O\left(\frac{(v_H-v_L)^2}{2\beta_n}\right)
	\tag{Lemma \ref{L10}}\\
	&=o(1)\tag{since $\lambda_n\to\infty$, $\phi_{1,n}\to0$, $\beta_n\to\infty$}
\end{align}
Plugging inequalities \eqref{T3-E4} and \eqref{T3-E5} into equation \eqref{T3-E3} completes the proof.

\subsection{Proof of Claim \ref{Cx9}}

There are five steps to this proof.
First, we derive a more useful expression for agent $i$'s maximum expected score given signal $s_i$.
\begin{align}
	\tilde{\rho}_n\left(s_i\right)
	&=\ex{\max_{b_i}\ex[i,n]{\textnormal{SR}^{\textnormal{CRPS}}\left(b_i,Z_{i,n}\right)}}\tag{definition of $\tilde{\rho}_n$}\\
	&=\ex{\max_{b_i}\ex[i,n]{-\lambda_n\int_{[z_L,z_H]}\left(\pr[b_i]{Z_{i,n}\leq z}-\textbf{1}(Z_{i,n}\leq z)\right)^2dz}}\tag{defn. of SR}\\
	&=\ex{\ex[i,n]{-\lambda_n\int_{[z_L,z_H]}\left(\pr[i,n]{Z_{i,n}\leq z}-\textbf{1}(Z_{i,n}\leq z)\right)^2dz}}\tag{SR is proper}\\
	&=-\lambda_n\int_{[z_L,z_H]}\ex{\ex[i,n]{\left(\pr[i,n]{Z_{i,n}\leq z}-\textbf{1}(Z_{i,n}\leq z)\right)^2}}dz\tag{linearity}\\
	&=-\lambda_n\int_{[z_L,z_H]}\ex{\ex[i,n]{\left(\ex[i,n]{\textbf{1}(Z_{i,n}\leq z)}-\textbf{1}(Z_{i,n}\leq z)\right)^2}}dz\tag{$\pr{E}=\ex{\textbf{1}_E}$}\\
	&=-\lambda_n\int_{[z_L,z_H]}\ex{\var[i,n]{\textbf{1}(Z_{i,n}\leq z)}}dz\tag{definition of variance}
\end{align}

Second, we derive a lower bound for agent $i$'s maximum expected score if she were to deviate to the combined signal $\bar{s}_i$.
Before proceeding, we reiterate two points.
\begin{enumerate}
	\item The combined signal $\bar{s}_i$ includes the signal $s_i$, because $s_i$ has finite cost for all types $\theta_i$ in the support of the joint distribution $F$.
	\item The combined signal $\bar{s}_i$ includes a revealing signal $s'_i$, which has finite cost for all types $\theta_i$ by Assumption \ref{A6}.
\end{enumerate}
It follows from these two observations that if agent $i$ acquires $\bar{s}_i$, she will know the value $v_i$ and the signal realization $s_i(\omega,\theta_i)$.
As a result, agent $i$ can infer $\epsilon_{i,n}$ -- the error in her report when she only acquires signal $s_i$.
With that, we can proceed.	
\begin{align}
	\tilde{\rho}_n\left(\bar{s}_i\right)
	&=\ex{\max_{b_i}\ex{\textnormal{SR}^{\textnormal{CRPS}}\left(b_i,Z_{i,n}\right)\mid\bar{s}_i,\bar{s}_i(\omega,\theta_i)}}\tag{definition of $\tilde{\rho}_n$}\\
	&\geq\ex{\max_{b_i}\ex{\textnormal{SR}^{\textnormal{CRPS}}\left(b_i,Z_{i,n}\right)\mid s_i,s_i(\omega,\theta_i),\epsilon_{i,n}}}\tag{throwing out info.}\\
	&=\ex{\max_{b_i}\ex[i,n]{\textnormal{SR}^{\textnormal{CRPS}}\left(b_i,Z_{i,n}\right)\mid\epsilon_{i,n}}}\tag{simplifying notation}\\
	&=-\lambda_n\int_{[z_L,z_H]}\ex{\var[i,n]{\textbf{1}(Z_{i,n}\leq z)\mid\epsilon_{i,n}}}dz\tag{similar to earlier derivation}\\
	&=-\lambda_n\int_{[z_L,z_H]}\ex{\ex[i,n]{\var[i,n]{\textbf{1}(Z_{i,n}\leq z)\mid\epsilon_{i,n}}}}dz\tag{LIE}
\end{align}

Third, we obtain a lower bound on agent $i$'s expected gain in score that agent $i$ if she were to deviate from signal $s_i$ to the combined signal $\bar{s}_i$.
Combining the first two parts of this proof,
\begin{align}
	\tilde{\rho}_n\left(\bar{s}_i\right)-\tilde{\rho}_n\left(s_i\right)
	&\geq\lambda_n\int_{[z_L,z_H]}\ex{\var[i,n]{\textbf{1}(Z_{i,n}\leq z)}-\ex[i,n]{\var[i,n]{\textbf{1}(Z_{i,n}\leq z)\mid\epsilon_{i,n}}}}dz\label{E10}\\
	&=\lambda_n\int_{[z_L,z_H]}\ex{\var[i,n]{\ex[i,n]{\textbf{1}(Z_{i,n}\leq z)\mid\epsilon_{i,n}}}}dz\tag{law of total variance}
\end{align}

Fourth, for any constant $z$, we use the conditional covariance between $\epsilon_{i,n}$ and $\textbf{1}(Z_{i,n}\leq z)$ to bound the conditional variance of $\ex[i,n]{\textbf{1}(Z_{i,n}\leq z)\mid\epsilon_{i,n}}$.
\begin{align}
	&\cov[i,n]{\epsilon_{i,n},\textbf{1}(Z_{i,n}\leq z)}\notag\\
	&=\ex[i,n]{\epsilon_{i,n}\left(\textbf{1}(Z_{i,n}\leq z)-\ex[i,n]{\textbf{1}(Z_{i,n}\leq z)}\right)}\tag{since $\ex[i,n]{\epsilon_{i,n}}=0$}\\
	&=\ex[i,n]{\ex[i,n]{\epsilon_{i,n}\left(\textbf{1}(Z_{i,n}\leq z)-\ex[i,n]{\textbf{1}(Z_{i,n}\leq z)}\right)\mid\epsilon_{i,n}}}\tag{LIE}\\
	&=\ex[i,n]{\epsilon_{i,n}\ex[i,n]{\textbf{1}(Z_{i,n}\leq z)-\ex[i,n]{\textbf{1}(Z_{i,n}\leq z)}\mid\epsilon_{i,n}}}\tag{linearity}\\
	&\leq\sqrt{\ex[i,n]{\epsilon_{i,n}^2\ex[i,n]{\textbf{1}(Z_{i,n}\leq z)-\ex[i,n]{\textbf{1}(Z_{i,n}\leq z)}\mid\epsilon_{i,n}}^2}}\tag{Jensen's inequality}\\
	&\leq\sqrt{\ex[i,n]{(v_H-v_L)^2\ex[i,n]{\textbf{1}(Z_{i,n}\leq z)-\ex[i,n]{\textbf{1}(Z_{i,n}\leq z)}\mid\epsilon_{i,n}}^2}}\tag{since $|\epsilon_{i,n}|\leq v_H-v_L$ and $\ex[i,n]{\cdot}^2\geq 0$}\\
	&=(v_H-v_L)\sqrt{\ex[i,n]{\left(\ex[i,n]{\textbf{1}(Z_{i,n}\leq z)\mid\epsilon_{i,n}}-\ex[i,n]{\ex[i,n]{\textbf{1}(Z_{i,n}\leq z)\mid\epsilon_{i,n}}}\right)^2}}\tag{linearity}\\
	&=(v_H-v_L)\sqrt{\var[i,n]{\ex[i,n]{\textbf{1}(Z_{i,n}\leq z)\mid\epsilon_{i,n}}}}\tag{definition of variance}
\end{align}
Rearranging this inequality gives us
\begin{equation}\label{E11}
	\var[i,n]{\ex[i,n]{\textbf{1}(Z_{i,n}\leq z)\mid\epsilon_{i,n}}}\geq\frac{1}{(v_H-v_L)^2}\cov[i,n]{\epsilon_{i,n},\textbf{1}(Z_{i,n}\leq z)}^2
\end{equation}

Fifth and finally, we relax the lower bound in step 3, in terms of quantities that appear in the claim statement.
Combining inequalities \eqref{E10} and \eqref{E11}, we obtain
\begin{align}
	&\tilde{\rho}_n\left(\bar{s}_i\right)-\tilde{\rho}_n\left(s_i\right)\label{E9}\\
	&\geq\frac{\lambda_n}{(v_H-v_L)^2}\int_{[z_L,z_H]}\ex{\cov[i,n]{\epsilon_{i,n},\textbf{1}(Z_{i,n}\leq z)}^2}dz\tag{combine \eqref{E10} and \eqref{E11}}\\
	&=\frac{\lambda_n}{(v_H-v_L)^2}\ex{\int_{[z_L,z_H]}\cov[i,n]{\epsilon_{i,n},\textbf{1}(Z_{i,n}\leq z)}^2dz}\tag{linearity}
\end{align}
Next, focus on the term inside the expectation.
\begin{align*}
	&\int_{[z_L,z_H]}\cov[i,n]{\epsilon_{i,n},\textbf{1}(Z_{i,n}\leq z)}^2dz\notag\\
	&=\left\|\cov[i,n]{\epsilon_{i,n},\textbf{1}(Z_{i,n}\leq z)}\right\|_2^2\tag{definition of $L^2$ norm}\\
	&\geq\frac{1}{(v_H-v_L)}\left\|\cov[i,n]{\epsilon_{i,n},\textbf{1}(Z_{i,n}\leq z)}\right\|_1^2\tag{H\"{o}lder's inequality}\\
	&=\frac{1}{(v_H-v_L)}\left(\int_{[z_L,z_H]}\left|\cov[i,n]{\epsilon_{i,n},\textbf{1}(Z_{i,n}\leq z)}\right|dz\right)^2\tag{defn. of $L^1$ norm}\\
	&\geq\frac{1}{(v_H-v_L)}\left(\int_{[z_L,z_H]}\cov[i,n]{\epsilon_{i,n},\textbf{1}(Z_{i,n}\leq z)}dz\right)^2\tag{since $|X|\geq X$}
\end{align*}
Now, focus on the term inside the square.
\begin{align}
	&\int_{[z_L,z_H]}\cov[i,n]{\epsilon_{i,n},\textbf{1}(Z_{i,n}\leq z)}dz\notag\\
	&=\int_{[z_L,z_H]}\ex[i,n]{\epsilon_{i,n}\textbf{1}(Z_{i,n}\leq z)}dz\tag{since $\ex[i,n]{\epsilon_{i,n}}=0$}\\
	&=\int_{[z_L,z_H]}\ex[i,n]{\ex[i,n]{\epsilon_{i,n}\textbf{1}(Z_{i,n}\leq z)\mid\epsilon_{i,n}}}dz\tag{LIE}\\
	&=\ex[i,n]{\epsilon_{i,n}\int_{[z_L,z_H]}\ex[i,n]{\textbf{1}(Z_{i,n}\leq z)\mid\epsilon_{i,n}}dz}\tag{linearity}\\
	&=\ex[i,n]{\epsilon_{i,n}\int_{[z_L,z_H]}\pr[i,n]{Z_{i,n}\leq z\mid\epsilon_{i,n}}dz}\tag{since $\pr{E}=\ex{\textbf{1}_E}$}\\
	&=(v_H-v_L)\ex[i,n]{\epsilon_{i,n}}-\ex[i,n]{\epsilon_{i,n}\int_{[z_L,z_H]}\left(1-\pr[i,n]{Z_{i,n}\leq z\mid\epsilon_{i,n}}\right)dz}\tag{linearity}\\
	&=-\ex[i,n]{\epsilon_{i,n}\int_{[z_L,z_H]}\left(1-\pr[i,n]{Z_{i,n}\leq z\mid\epsilon_{i,n}}\right)dz}\tag{since $\ex[i,n]{\epsilon_{i,n}}=0$}\\
	&=-\ex[i,n]{\epsilon_{i,n}\ex[i,n]{Z_{i,n}\mid\epsilon_{i,n}}}\tag{definition of $\ex{\cdot}$}\\
	&=-\ex[i,n]{\epsilon_{i,n}Z_{i,n}}\tag{linearity and LIE}\\
	&=-\cov[i,n]{\epsilon_{i,n},Z_{i,n}}\tag{since $\ex[i,n]{\epsilon_{i,n}}=0$}
\end{align}
Combine this sequence of inequalities, starting with \eqref{E9}, to obtain
\begin{align}
	\tilde{\rho}_n\left(\bar{s}_i\right)-\tilde{\rho}_n\left(s_i\right)
	&\geq\frac{\lambda_n}{(v_H-v_L)^2}\ex{\frac{1}{(v_H-v_L)}\left(-\cov[i,n]{\epsilon_{i,n},Z_{i,n}}\right)^2}\notag\\
	&=\frac{\lambda_n}{(v_H-v_L)^3}\ex{\cov[i,n]{\epsilon_{i,n},Z_{i,n}}^2}\tag{linearity}\\
	&=\frac{\lambda_n}{(v_H-v_L)^3}\cov{\epsilon_{i,n},Z_{i,n}}^2\tag{LIE and since $\ex[i,n]{\epsilon_{i,n}}=0$}
\end{align}
This, rearranged, gives us the statement of the claim.

\subsection{Proof of Claim \ref{Cx4}}

First, we study the covariance between agent $i$'s error $\epsilon_{i,n}$ and $Z_{i,n}$ in two cases.
\begin{enumerate}
	\item Suppose that agent $i$ does not acquire the combined signal $\bar{s}_i$.
	Then she could deviate from the signal $s_i$ that she did acquire to $\bar{s}_i$.
	This will change her expected transfers associated with the scoring rule $\textnormal{SR}^{\textnormal{CRPS}}$ from $\tilde{\rho}_n(s_i)$ to $\tilde{\rho}_n\left(\bar{s}\right)$.
	This may also increase her costs, but the increase will be at most $\bar{c}$.
	Putting everything together, the fact that the deviation cannot be profitable in equilibrium implies
	$
		\tilde{\rho}_n\left(\bar{s}_i\right)-\bar{c}\leq\tilde{\rho}_n(s_i)
	$.
	Combining this inequality with Claim \ref{Cx9} yields\footnote{Claim \ref{Cx9} requires the signal $s_i$ acquired by agent $i$ to have finite cost for all types in the support of the joint distribution $F$, i.e., $c(\theta_i,s_i)<\infty$. But this holds in every equilibrium; otherwise it would be profitable to deviate to the combined signal $\bar{s}_i$, which has finite expected cost.}
	\begin{equation}
		\bar{c}\geq\frac{\lambda_n}{(v_H-v_L)^3}\left|\cov{\epsilon_{i,n},Z_{i,n}}\right|^2\label{E8}
	\end{equation}
	\item Suppose that agent $i$ acquires the combined signal $\bar{s}_i$.
	Then, almost surely, she knows $v_i$ and reports $\hat{v}_{i,n}=v_i$.
	Conditional on agent $i$'s information, the error $\epsilon_{i,n}$ is almost surely constant and so its covariance with any other random variable is zero.
	Therefore, inequality \eqref{E8} still holds, albeit vacuously.
\end{enumerate}

Second, we verify that the covariance between $\epsilon_i$ and $Z_{i,n}$ is equal to the covariance between $\epsilon_i$ and the conditional expectation $\ex{Z_{i,n}\mid\omega}$.
\begin{align}
	\cov{\epsilon_{ix},Z_{i,n}}
	&=\ex{\epsilon_{ix}Z_{i,n}}\tag{since $\ex{\epsilon_{ix}}=0$}\\
	&=\ex{\ex{\epsilon_{ix}Z_{i,n}\mid\omega}}\tag{LIE}\\
	&=\ex{\ex{\epsilon_{ix}\mid\omega}\ex{Z_{i,n}\mid\omega}}\tag{cond. independence}\\
	&=\ex{\ex{\epsilon_{ix}\ex{Z_{i,n}\mid\omega}\mid\omega}}\tag{linearity}\\
	&=\ex{\epsilon_{ix}\ex{Z_{i,n}\mid\omega}}\tag{LIE}\\
	&=\cov{\epsilon_{ix},\ex{Z_{i,n}\mid\omega}}\tag{since $\ex{\epsilon_{ix}}=0$}
\end{align}
Combining this identity with inequality \eqref{E8} completes the proof.

\subsection{Proof of Claim \ref{Cx5}}

Note that the proof of Claim \ref{Cx9} holds verbatim if we replace $\tilde{\rho}_n$ with $\lambda_n\rho$ and $\tilde{v}_{i,n}$ with $\ex{v_j\mid\omega}$.
After dividing by $\lambda_n$, the statement of that claim becomes
\[
\rho\left(\bar{s}_i\right)\geq\rho\left(s_i\right)+\frac{1}{(v_H-v_L)^3}\left\|\cov{\epsilon_{i},\ex{v_j\mid\omega}}\right\|^2
\]
However, since the signal $s_i$ has finite expected cost according to auxiliary cost function $\tilde{c}$, it follows from the definition of $\tilde{c}$ in \eqref{T2-E1} that
$
\rho\left(\bar{s}_i\right)\geq\rho\left(s_i\right)
$.
Combining these two inequalities yields
$
0=\frac{1}{(v_H-v_L)^3}\left\|\cov{\epsilon_{i},\ex{v_j\mid\omega}}\right\|^2=\left\|\cov{\epsilon_{i},\ex{v_j\mid\omega}}\right\|
$

\subsection{Proof of Claim \ref{Cx6}}

Observe that, with probability $1-\phi_{2,n}$,
\begin{align}
	&\var{\ex{\frac{1}{n}\sum_{j=1}^n\epsilon_{j,n}\mid\omega}}\notag\\
	&=\ex{\ex{\ex{\frac{1}{n}\sum_{j=1}^n\epsilon_{j,n}\mid\omega}\cdot\frac{1}{n}\sum_{j=1}^n\epsilon_{j,n}\mid\omega}}\tag{linearity and since $\ex{\epsilon_{i,n}}=0$}\\
	&=\ex{\ex{\frac{1}{n}\sum_{j=1}^n\epsilon_{j,n}\mid\omega}\cdot\frac{1}{n}\sum_{j=1}^n\epsilon_{j,n}}\tag{LIE}\\
	&=\frac{1}{n^2}\sum_{j=1}^n\ex{\left(\ex{\epsilon_{j,n}\mid\omega}+\sum_{i\neq j}\ex{\epsilon_{i,n}\mid\omega}\right)\cdot\epsilon_{i,n}}\tag{linearity}\\
	&=\frac{1}{n^2}\sum_{j=1}^n\ex{\left(\ex{\epsilon_{j,n}\mid\omega}+\sum_{i\neq j}\ex{v_i\mid\omega}-\sum_{i\neq j}\ex{\ex[i,n]{v_i}\mid\omega}\right)\cdot\epsilon_{j,n}}\tag{defn. of $\epsilon_{i,n}$}\\
	&\leq\frac{1}{n^2}\sum_{j=1}^n\ex{\left(\ex{\epsilon_{j,n}\mid\omega}+\sum_{i\neq j}\ex{v_i\mid\omega}-\sum_{i\neq j}\ex{\hat{v}_{i,n}\mid\omega}\right)\epsilon_{j,n}}
	+\frac{n(n-1)(v_H-v_L)}{n^2}\phi_{1,n}
	\tag{condition \eqref{E15}}\\
	&=\frac{1}{n^2}\sum_{j=1}^n\ex{\left(\ex{\epsilon_{j,n}\mid\omega}-\sum_{i\neq j}\ex{\hat{v}_{i,n}\mid\omega}\right)\cdot\epsilon_{j,n}}
	+\frac{n(n-1)(v_H-v_L)}{n^2}\cdot\phi_{1,n}
	\tag{Claim \ref{Cx5}}\\
	&\leq\frac{(v_H-v_L)^2}{n}
	-\frac{1}{n^2}\sum_{j=1}^n\ex{\epsilon_{j,n}\sum_{i\neq j}\ex{\hat{v}_{i,n}\mid\omega}}
	+\frac{n(n-1)(v_H-v_L)}{n^2}\cdot\phi_{1,n}
	\tag{since $\epsilon_{i,n}\in[v_L,v_H]$}\\
	&=\frac{(v_H-v_L)^2}{n}
	-\frac{n'}{n^2}\sum_{j=1}^n\ex{\epsilon_{j,n}\tilde{v}_{j,n}}
	+\frac{n(n-1)(v_H-v_L)}{n^2}\cdot\phi_{1,n}
	\tag{defn. of $\tilde{v}_{i,n}$ and since $v_{0,n}\perp\epsilon_{j,n}$}\\
	&=\frac{(v_H-v_L)^2}{n}
	+\frac{n'}{n}\sqrt{\frac{(v_H-v_L)^3\bar{c}}{\lambda_n}}
	+\frac{n(n-1)(v_H-v_L)}{n^2}\cdot\phi_{1,n}
	\tag{Claim \ref{Cx4}}\\
	&\leq\frac{(v_H-v_L)^2}{n}
	+\sqrt{\frac{(v_H-v_L)^3\bar{c}}{\lambda_n}}
	+(v_H-v_L)\phi_{1,n}\tag{since $n'\leq n$ and $n-1\leq n$}
\end{align}
Note that the probability $\phi_{2,n}$ event where this may not hold can affect the variance by at most $(v_H-v_L)^2\cdot\phi_{2,n}$.
Finally, the second part of the claim follows from the fact that adding a term of size $O(n^{-1})$ has a vanishing impact on the variance.

\subsection{Proof of Claim \ref{Cx10}}

The proof has three parts.
First, we show that a particular sum of random variables has expected value $\mu_n=0$.
This sum will become relevant later on.
\begin{align}
	\mu_n
	&:=\ex{\frac{1}{n-1}\sum_{j\neq i}\left(v_j-\ex{v_j\mid\omega}-\epsilon_{j,n}\right)}\notag\\
	&=\frac{1}{n-1}\sum_{j\neq i}\left(\ex{v_j}-\ex{\ex{v_j\mid\omega}}\right)\tag{linearity, and since $\ex{\epsilon_{j,n}}=0$}\\
	&=0\tag{LIE}\notag
\end{align}
Second, we bound the variance $\sigma^2_n$ of this same sum of random variables.
\begin{align}
	\sigma_n^2
	:=\:&\var{\frac{1}{n-1}\sum_{j\neq i}\left(v_j-\ex{v_j\mid\omega}-\epsilon_{j,n}\right)}\notag\\
	=\:&\ex{\left(\frac{1}{n-1}\sum_{j\neq i}\left(v_j-\ex{v_j\mid\omega}-\epsilon_{j,n}\right)\right)^2}\tag{since $\mu_n=0$}\\
	\leq\:&2\ex{\left(\frac{1}{n-1}\sum_{j\neq i}\left(v_j-\ex{v_j\mid\omega}\right)\right)^2}
	+2\ex{\left(\frac{1}{n-1}\sum_{j\neq i}\epsilon_{j,n}\right)^2}
	\tag{Cauchy-Schwarz}\\
	=\:&2\ex{\ex{\left(\frac{1}{n-1}\sum_{j\neq i}\left(v_j-\ex{v_j\mid\omega}\right)\right)^2\mid\omega}}
	+2\ex{\ex{\left(\frac{1}{n-1}\sum_{j\neq i}\epsilon_{j,n}\right)^2\mid\omega}}
	\tag{LIE}
\end{align}
We will consider both of these terms in turn.
The first term is
\begin{align}
	\ex{\left(\frac{1}{n-1}\sum_{j\neq i}\left(v_j-\ex{v_j\mid\omega}\right)\right)^2\mid\omega}
	&=\var{\frac{1}{n-1}\sum_{j\neq i}v_j\mid\omega}\tag{defn. of variance}\\
	&=\frac{1}{n-1}\var{v_j\mid\omega}\tag{cond. indep.}\\
	&\leq\frac{(v_H-v_L)^2}{n-1}\tag{since $v_j\in[v_L,v_H]$}
\end{align}
Now, consider the second term.
\begin{align}
	&\ex{\left(\frac{1}{n-1}\sum_{j\neq i}\epsilon_{j,n}\right)^2\mid\omega}\\
	&=\var{\frac{1}{n-1}\sum_{j\neq i}\epsilon_{j,n}\mid\omega}
	+\left(\ex{\frac{1}{n-1}\sum_{j\neq i}\epsilon_{j,n}\mid\omega}\right)^2
	\tag{defn. of $\var{\cdot}$}\\
	&=\frac{1}{(n-1)^2}\sum_{j\neq i}\var{\epsilon_{j,n}\mid\omega}
	+\left(\ex{\frac{1}{n-1}\sum_{j\neq i}\epsilon_{j,n}\mid\omega}\right)^2
	\tag{cond. indep.}\\
	&\leq\frac{(v_H-v_L)^2}{n-1}
	+\left(\ex{\frac{1}{n-1}\sum_{j\neq i}\epsilon_{j,n}\mid\omega}\right)^2
	\tag{since $\epsilon_{j,n}\in[v_L,v_H]$}
\end{align}
We plug these new expressions for the two terms back into our inequality.
\begin{align}
	\sigma_n^2
	&\leq\frac{4(v_H-v_L)^2}{n-1}+2\ex{\left(\ex{\frac{1}{n-1}\sum_{j\neq i}\epsilon_{j,n}\mid\omega}\right)^2}\label{E14}\\
	&=\frac{4(v_H-v_L)^2}{n-1}+2\var{\ex{\frac{1}{n-1}\sum_{j\neq i}\epsilon_{j,n}\mid\omega}}+2\left(\ex{\ex{\frac{1}{n-1}\sum_{j\neq i}\epsilon_{j,n}\mid\omega}}\right)^2\tag{defn. of variance}\\
	&=\frac{4(v_H-v_L)^2}{n-1}+2\var{\ex{\epsilon_{j,n}\mid\omega}}\tag{LIE and since $\ex{\epsilon_{j,n}}=0$}\\
	&\leq\frac{4(v_H-v_L)^2}{n-1}+o(1)\tag{Claim \ref{Cx6}}
\end{align}
Third and finally, we show that $\tilde{v}_{i,n}$ concentrates around $\ex{v_j\mid\omega}$.
Given $t>0$,
\begin{align}
	&\pr{\left|\tilde{v}_{i,n}-\ex{v_j\mid\omega}\right|\geq t}
	\notag\\
	&=\pr{\left|\frac{1}{n'}\left(v_{0,n}+\sum_{j\neq i}\hat{v}_{j,n}\right)-\ex{v_j\mid\omega}\right|\geq t}
	\tag{defn. of $\tilde{v}_{i,n}$}\\
	&\leq\pr{\left|\frac{v_{0,n}}{n}\right|+\left|\frac{n'-n+1}{n'(n-1)}\left(\sum_{j\neq i}\hat{v}_{j,n}\right)\right|+\left|\frac{1}{n-1}\left(\sum_{j\neq i}\hat{v}_{j,n}\right)-\ex{v_j\mid\omega}\right|\geq t}
	\tag{triangle inequality}\\
	&\leq
	\pr{\left|\frac{v_{0,n}}{n}\right|\geq \frac{t}{3}}
	+\pr{\left|\frac{n'-n+1}{n'(n-1)}\left(\sum_{j\neq i}\hat{v}_{j,n}\right)\right|\geq \frac{t}{3}}
	\label{E26}\\
	&\white{\leq}+\pr{\left|\frac{1}{n-1}\left(\sum_{j\neq i}\hat{v}_{j,n}\right)-\ex{v_j\mid\omega}\right|\geq \frac{t}{3}}
	\tag{union bound}
\end{align}
The first two terms of line \eqref{E26} are vanishing, since $\beta_n=o(n)$ and $n'=O(n)$.
The third term is vanishing as well, since
\begin{align}
	&\pr{\left|\frac{1}{n-1}\sum_{j\neq i}\left(\hat{v}_j-\ex{v_j\mid\omega}\right)\right|\geq \frac{t}{3}}\notag\\
	&\leq\delta_n+\phi_{1,n}+\phi_{2,n}+\pr{\left|\frac{1}{n-1}\sum_{j\neq i}\left(\ex[i,n]{v_j}-\ex{v_j\mid\omega}\right)\right|\geq \frac{t}{3}}\tag{defn. $\delta_n$, $\phi_{1,n}$, $\phi_{2,n}$}\\
	&=\delta_n+\phi_{1,n}+\phi_{2,n}+\pr{\left|\frac{1}{n-1}\sum_{j\neq i}\left(v_j-\epsilon_{j,n}-\ex{v_j\mid\omega}\right)\right|\geq \frac{t}{3}}\tag{definition of $\epsilon_{i,n}$}\\
	&=\delta_n+\phi_{1,n}+\phi_{2,n}+\pr{\left|\frac{1}{n-1}\sum_{j\neq i}\left(v_j-\epsilon_{j,n}-\ex{v_j\mid\omega}\right)-\mu_n\right|\geq \frac{t}{3}}\tag{since $\mu_n=0$}\\
	&\leq\delta_n+\phi_{1,n}+\phi_{2,n}+\frac{9\sigma_n^2}{t^2}\tag{Chebyshev's inequality}\\
	&\leq\delta_n+\phi_{1,n}+\phi_{2,n}+\frac{9}{t^2}\left(\frac{4(v_H-v_L)^2}{n-1}+o(1)\right)\tag{inequality \eqref{E14}}
\end{align}
Therefore, $\tilde{v}_{i,n}$ converges in probability to $\ex{v_j\mid\omega}$.

	\subsection{Proof of Claim \ref{Cx11}}
	
	Let $Z_n$ be a sequence of random variables where $Z_n\to_pZ$.
To find $\lim_{n\to\infty}f(Z_n,\cdot)$, we first consider the limits of two random variables that appear in the definition of $f$.
\begin{enumerate}
	\item The random variable $\pr[i]{Z_n\leq z}$ converges in probability to $\pr[i]{Z\leq z}$ for almost all $z$.
	To see this, fix any constant $t>0$ and define the probability
	\[
	p^+_n:=\pr{\pr[i]{Z\leq z}-\pr[i]{Z_n\leq z}\geq t}
	\]
	I begin by showing that $p^+_n\to0$.
	Observe that
	\begin{align}
		\pr{Z\leq z}-\pr{Z_n\leq z}
		&=\ex{\pr[i]{Z\leq z}-\pr[i]{Z_n\leq z}}\tag{LIE}\\
		&\geq p^+_n\cdot t+(1-p^+_n)\cdot O(t)\tag{LIE and defn. of $p^+_n$}
	\end{align}
	However, since $Z_n\to_pZ$, we know that $Z_n\to_dZ$.
	By the Portmanteau theorem,
	$\pr{Z_n\leq z}\to\pr{Z\leq z}$ for almost all $z\in\reals$.\footnote{The distribution of $Z$ can have at most finitely many atoms. As long as $z$ is not one of those atoms, $\{z'\in\reals\mid z'\leq z\}$ is a continuity set for that distribution, and we can apply the Portmanteau theorem.}
    Combining this with the previous inequality implies
	$
	p^+_n\cdot t+(1-p^+_n)\cdot O(t)\to0
	$.
	Since $t$ does not depend on $n$, this can only be true when $p^+_n\to0$.
	By symmetry, we can apply this argument again to argue that
	$
	p_n^-:=\pr{\pr[i]{Z_n\leq z}-\pr[i]{Z\leq z}\geq t}\to0
	$.
	By the definition of convergence in probability, the fact that $p_n^+\to0$ and $p_n^-\to0$ implies that $\pr[i]{Z_n\leq z}\to_p\pr[i]{Z\leq z}$.
	
	\item The random variable $\textbf{1}(Z_n\leq z)$ converges in probability to $\textbf{1}(Z\leq z)$ for almost all $z$.
	This follows from the continuous mapping theorem.
	The distribution of $Z$ can have at most finitely many atoms and the indicator function has only one discontinuity point, at $z$.
	Provided that $z$ is not one of those atoms, the probability that $Z$ matches a discontinuity point is zero.
	Therefore, we can invoke the continuous mapping theorem for all but finitely many values of $z$.
\end{enumerate}
It follows from these two results and the continuous mapping theorem that
\[
\left(\pr[i]{Z_n\leq z}-\textbf{1}(Z_n\leq z)\right)^2\to_p\left(\pr[i]{Z\leq z}-\textbf{1}(Z\leq z)\right)^2
\]
for almost all $z\in\reals$.
Since convergence in probability implies convergence in distribution, we can apply the Portmanteau theorem to show that
$
f(Z_n,\cdot)\to f(Z,\cdot)
$
almost everywhere.
Set $Z_n=\tilde{v}_{i,n}$ and $Z=\ex{v_j\mid\omega}$, where $Z_n\to_pZ$ by Claim \ref{Cx10}.

	\subsection{Proof of Claims \ref{T3-Cl1} and \ref{T3-Cl5}}

For convenience, let $\mathcal{V}$ refer to the (finite) set of values in the support of $F$.
We want to bound the following term:
\begin{align}
	&\left|\ex[i,n]{v_i\mid Q_{i,n}}-\ex[i,n]{v_i}\right|\label{T3-E7}\\
	&=\left|\sum_{y\in\mathcal{V}}y\cdot\pr[i,n]{v_i=y\mid Q_{i,n}}-\sum_{y\in\mathcal{V}}y\cdot\pr[i,n]{v_i=y}\right|
	\tag{LIE}\\
	&\leq\sum_{y\in\mathcal{V}}y\cdot\left|\pr[i,n]{v_i=y\mid Q_{i,n}}-\pr[i,n]{v_i=y}\right|
	\tag{triangle inequality}\\
	&=\sum_{y\in\mathcal{V}}y\cdot\left|\frac{\pr[i,n]{ Q_{i,n}\mid v_i=y}\cdot\pr[i,n]{v_i=y}}{q_{i,n}}-\pr[i,n]{v_i=y}\right|
	\tag{Bayes' rule}\\
	&=\frac{1}{q_{i,n}}\sum_{y\in\mathcal{V}}y\cdot\pr[i,n]{v_i=y}\cdot\left|\pr[i,n]{ Q_{i,n}\mid v_i=y}-q_{i,n}\right|\notag\\
	&\leq\frac{\sqrt{|\mathcal{V}|}}{q_{i,n}}\sqrt{\sum_{y\in\mathcal{V}}\left(y\cdot\pr[i,n]{v_i=y}\cdot\left|\pr[i,n]{ Q_{i,n}\mid v_i=y}-q_{i,n}\right|\right)^2}
	\tag{since $\Vert\alpha\Vert_1\leq\sqrt{d}\cdot\Vert\alpha\Vert_2$ for $\alpha\in\reals^d$}\\
	&\leq\frac{\sqrt{|\mathcal{V}|}}{q_{i,n}}\sqrt{\sum_{y\in\mathcal{V}}\pr[i,n]{v_i=y}\cdot\left(y\cdot\left|\pr[i,n]{ Q_{i,n}\mid v_i=y}-q_{i,n}\right|\right)^2}
	\tag{since $\pr{\cdot}^2\leq\pr{\cdot}$}\\
	&\leq\frac{\sqrt{|\mathcal{V}|\cdot|v_H-v_L|}}{q_{i,n}}\sqrt{\sum_{y\in\mathcal{V}}\pr[i,n]{v_i=y}\cdot\left(\pr[i,n]{ Q_{i,n}\mid v_i=y}-q_{i,n}\right)^2}
	\tag{$y\leq|v_H-v_L|$}\\
	&=\frac{\sqrt{|\mathcal{V}|\cdot|v_H-v_L|}}{q_{i,n}}\sqrt{\ex[i,n]{\left(\pr[i,n]{ Q_{i,n}\mid v_i}-q_{i,n}\right)^2}}
	\tag{LIE}
\end{align}
Focus on the term inside the square root.
\begin{align}
	&\ex[i,n]{\left(\pr[i,n]{ Q_{i,n}\mid v_i}-q_{i,n}\right)^2}\notag\\
	&\leq\ex[i,n]{\left(\pr[i,n]{ Q_{i,n}\mid v_i}-q_{i,n}\right)^2
		+\left(\pr[i,n]{\tilde{v}_{i,n}\notin Q_n\mid v_i}-(1-q_{i,n})\right)^2}
	\tag{since $(\cdot)^2\geq 0$}\\
	&=\ex[i,n]{
		\pr[i,n]{ Q_{i,n}\mid v_i}^2
		+\pr[i,n]{\tilde{v}_{i,n}\notin Q_n\mid v_i}^2}
	\notag\\
	&\white{=}+q_{i,n}^2+(1-q_{i,n})^2
	-2(1-q_{i,n})\cdot\pr[i,n]{\tilde{v}_{i,n}\notin Q_n}
	-2q_{i,n}\cdot\pr[i,n]{ Q_{i,n}}
	\tag{LIE}\\
	&=\ex[i,n]{
		\pr[i,n]{ Q_{i,n}\mid v_i}^2
		+\pr[i,n]{\tilde{v}_{i,n}\notin Q_n\mid v_i}^2}
	-q_{i,n}^2-(1-q_{i,n})^2\tag{defn. of $q_{i,n}$}\\
	&=\ex[i,n]{\max_{b_i}\ex[i,n]{\textnormal{SR}^{\textnormal{Q}}(b_i,\tilde{v}_i)\mid v_i}}
	-\max_{b_i}\ex[i,n]{\textnormal{SR}^{\textnormal{Q}}(b_i,\tilde{v}_i)}\label{T3-E8}
\end{align}
The last equality follows from Example 1 of \textcite{GR07}.
Next,
\begin{align}
	\eqref{T3-E8}
	\leq&\:
	\eqref{T3-E8}+\ex[i,n]{\max_{b_i}\ex[i,n]{\left(\textnormal{SR}(b_i,\tilde{v}_i)-\textnormal{SR}^Q(b_i,\tilde{v}_i)\right)\mid v_i}}\notag\\
	&-\max_{b_i}\ex[i,n]{\textnormal{SR}(b_i,\tilde{v}_i)-\textnormal{SR}^Q(b_i,\tilde{v}_i)}\tag{$\ex{\max(\cdot)}\geq\max\ex{\cdot}$}\\
	=&\:\ex[i,n]{\max_{b_i}\ex[i,n]{\textnormal{SR}(b_i,\tilde{v}_i)\mid v_i}}
	-\max_{b_i}\ex[i,n]{\textnormal{SR}(b_i,\tilde{v}_i)}\label{T3-E9}
\end{align}
The last equality follows from the fact that $\textnormal{SR}-\textnormal{SR}^Q$ happens to be a proper scoring rule.
Recall that the agent can learn $v_i$ by acquiring a revealing signal  in addition to her current signal $s_i$.
The cost of this is no greater than $\bar{c}$.
In equilibrium, therefore,
\begin{equation}
	\lambda_n\cdot\ex{\eqref{T3-E9}}\leq\bar{c}\label{T3-E11}
\end{equation}
Combining all of the preceding arguments, we obtain
\begin{align}
	\ex{\eqref{T3-E7}}^2
	&\leq\ex{\eqref{T3-E7}^2}\tag{Jensen's inequality}\\
	&\leq\ex{\frac{|\mathcal{V}|\cdot|v_H-v_L|}{q_{i,n}^2}\cdot\ex[i,n]{\left(\pr[i,n]{ Q_{i,n}\mid v_i}-q_{i,n}\right)^2}}\tag{inequality \eqref{T3-E7}}\\
	&\leq\frac{|\mathcal{V}|\cdot|v_H-v_L|}{(\delta_n/n)^2}\cdot\ex{\ex[i,n]{\left(\pr[i,n]{ Q_{i,n}\mid v_i}-q_{i,n}\right)^2}}
	\tag{since $q_{i,n}\geq\delta_n/n$}\\
	&\leq\frac{|\mathcal{V}|\cdot|v_H-v_L|}{(\delta_n/n)^2}\cdot\ex{\eqref{T3-E8}}
	\tag{inequality \eqref{T3-E8}}\\
	&\leq\frac{|\mathcal{V}|\cdot|v_H-v_L|\cdot\bar{c}}{(\delta_n/n)^2\cdot\lambda_n}
	\tag{inequality \eqref{T3-E11}}
\end{align}
Combining this with Markov's inequality gives us
\[
\pr{\eqref{T3-E7}\geq t}\leq\frac{\ex{\eqref{T3-E7}}}{t}
\leq\frac{1}{t}\cdot\sqrt{\frac{|\mathcal{V}|\cdot|v_H-v_L|\cdot\bar{c}}{(\delta_n/n)^2\cdot\lambda_n}}
\]
	
	\subsection{Proof of Claim \ref{T3-Cl7}}

We rely on the fact that, for any constant $|w|\leq 1$,
\begin{equation}\label{T3-E23}
e^w\in\left[1+w-\frac{e\cdot w^2}{2},1+w+\frac{e\cdot w^2}{2}\right]
\end{equation}
There are three cases to consider.
\begin{enumerate}
	\item Suppose $a\geq 0$. Then
	\begin{align}
		G(b)-G(a)
		&=\frac{1}{2}\left(e^{-\frac{\left|a\right|}{\beta_n}}-e^{-\frac{\left|b\right|}{\beta_n}}\right)\tag{since $v_{0,n}\sim\textsc{Laplace}(\beta_n)$}\\
		&=\frac{1}{2}\left(e^{-\frac{\left|b\right|}{\beta_n}}\right)\left(e^{-\frac{\left|a\right|}{\beta_n}+\frac{\left|b\right|}{\beta_n}}-1\right)\notag\\
		&=\frac{1}{2}\left(e^{-\frac{\left|b\right|}{\beta_n}}\right)\left(e^{\frac{b-a}{\beta_n}}-1\right)\tag{since $b\geq a\geq 0$}\\
		&\in\frac{1}{2}\left(e^{-\frac{\left|b\right|}{\beta_n}}\right)\left(\frac{b-a}{\beta_n}+\frac{e}{2}\left[-\left(\frac{b-a}{\beta_n}\right)^2,\left(\frac{b-a}{\beta_n}\right)^2\right]\right)\tag{bound \eqref{T3-E23}}
	\end{align}
	
	\item Suppose $b< 0$. This is analogous to the previous case. 
	
	\item Suppose $b\geq 0$, $a<0$.
	Note that $-\beta_n\leq a-b\leq b+a\leq b<b-a\leq\beta_n$.
	Then
	\begin{align}
		&G(b)-G(a)\notag\\
		&=1-\frac{1}{2}\left(e^{-\frac{\left|b\right|}{\beta_n}}+e^{-\frac{\left|a\right|}{\beta_n}}\right)\tag{since $v_{0,n}\sim\textsc{Laplace}(\beta_n)$}\\
		&=\frac{1}{2}\left(2e^0-e^{-\frac{\left|b\right|}{\beta_n}}-e^{-\frac{\left|a\right|}{\beta_n}}\right)\notag\\
		&=\frac{1}{2}\left(e^{-\frac{\left|b\right|}{\beta_n}}\right)
		\left(2e^{\frac{\left|b\right|}{\beta_n}}-1-e^{\frac{\left|b\right|}{\beta_n}-\frac{\left|a\right|}{\beta_n}}\right)\notag\\
		&=\frac{1}{2}\left(e^{-\frac{\left|b\right|}{\beta_n}}\right)
		\left(2e^{\frac{b}{\beta_n}}-1-e^{\frac{b+a}{\beta_n}}\right)\tag{since $a<0\leq b$}\\
		&\in
		\frac{1}{2}\left(e^{-\frac{\left|b\right|}{\beta_n}}\right)\left(\frac{b-a}{\beta_n}
		+\frac{e}{2}\left[
		-\left(\frac{b}{\beta_n}\right)^2-\left(\frac{b+a}{\beta_n}\right)^2,
		\left(\frac{b}{\beta_n}\right)^2+\left(\frac{b+a}{\beta_n}\right)^2
		\right]\right)\tag{bound \eqref{T3-E23}}\\
		&\subseteq
		\frac{1}{2}\left(e^{-\frac{\left|b\right|}{\beta_n}}\right)\left(\frac{b-a}{\beta_n}
		+e\left[
		-\left(\frac{b-a}{\beta_n}\right)^2,
		\left(\frac{b-a}{\beta_n}\right)^2
		\right]\right)\tag{since $a<0\leq b$}
	\end{align}
\end{enumerate}
The bound in case 3 is the same as the bound in the statement of the claim, while the bounds in cases 1 and 2 are slightly tighter.

	\subsection{Proof of Claim \ref{T3-Cl6}}

For event $E\subseteq\reals$, let
$
F^{z,z'}(E)=\pr[i,n]{\sum_{j\neq i}\hat{v}_{j,n}\in E\mid z\leq-v_{0,n}-\sum_{j\neq i}\hat{v}_{j,n}\leq z'}
$.
Then
\begin{align}
	&\ex[i,n]{v_i\mid z\leq-n\tilde{v}_{i,n}\leq z'}\label{T3-E12}\\
	&=\ex[i,n]{v_i\mid z\leq-v_{0,n}-\sum_{j\neq i}\hat{v}_{j,n}\leq z'}\tag{defn. of $\tilde{v}_i$}\\
	&=\int_{\reals}\ex[i,n]{v_i\mid z\leq-v_{0,n}-\sum_{j\neq i}\hat{v}_{j,n}\leq z',\sum_{j\neq i}\hat{v}_{j,n}=y}\cdot F^{z,z'}(dy)\tag{LIE}\\
	&=\int_{\reals}\ex[i,n]{v_i\mid\sum_{j\neq i}\hat{v}_{j,n}=y}\cdot F^{z,z'}(dy)\tag{since $v_i\perp v_{0,n}$}
\end{align}
Since the distribution $F^{z,z'}$ is not necessarily discrete or continuous, we need to refer to the definition of Lebesgue integration.

\begin{defn}
	A function $h:\reals\to\reals$ is \emph{simple} if there is a sequence $V^h_1,\ldots,V^h_{K_h}\subseteq\mathcal{V}$ of disjoint $F^{z,z'}$-measurable sets and $\alpha^h_1,\ldots,\alpha^h_K\in\reals_+$ such that
	$
	h(y)=\sum_{k=1}^{K_h}\alpha_k\textbf{1}(y\in V_k)
	$.
\end{defn}

Let $\mathcal{H}$ be the set of simple functions.
By definition of Lebesgue integration,
\begin{align}
	\eqref{T3-E12}
	&=\sup_{h\in\mathcal{H}}\sum_{k=1}^{K_h}\alpha^h_k\cdot F^{z,z'}\left(V^h_k\right)-\sup_{h'\in\mathcal{H}}\sum_{k=1}^{K_{h'}}\alpha^{h'}_k\cdot F^{z,z'}\left(V^{h'}_k\right)\label{T3-E18}\\
	&
	\text{s.t.}\quad h(y)\leq\left(\ex[i,n]{v_i\mid\sum_{j\neq i}\hat{v}_{j,n}=y}\right)^+
	\text{and}\quad h'(y)\leq-\left(\ex[i,n]{v_i\mid\sum_{j\neq i}\hat{v}_{j,n}=y}\right)^-\notag
\end{align}
Consider a sequence of simple functions $(h_l)_{l=1}^\infty$.
For convenience, let $K_l=K_{h_l}$ and $V_k^l=V_k^{h_l}$.
Let
$
\Delta^k_l:=\sup V^l_k-\inf V^l_k$ and $\Delta_l=\max_{k=1,\ldots,K_l}\Delta^k_l
$.
We eventually let $h_l$ converge to one of the suprema in equation \eqref{T3-E18}, so it is without loss of generality to restrict attention to sequences that satisfy $\Delta_l\to0$.

Applying Bayes' rule, we find that
\begin{align}
	&F^{z,z'}(V_k^l)
	=\frac{\pr[i,n]{z\leq-v_{0,n}-\sum_{j\neq i}\hat{v}_{j,n}\leq z'\mid\sum_{j\neq i}\hat{v}_{j,n}\in V_k^l}
		\pr[i,n]{\sum_{j\neq i}\hat{v}_{j,n}\in V_k^l}}
	{\pr[i,n]{z\leq-v_{0,n}-\sum_{j\neq i}\hat{v}_{j,n}\leq z'}}\notag\\
	&=\frac{\pr[i,n]{z\leq-v_{0,n}-\sum_{j\neq i}\hat{v}_{j,n}\leq z'\mid\sum_{j\neq i}\hat{v}_{j,n}\in V_k^l}
		\pr[i,n]{\sum_{j\neq i}\hat{v}_{j,n}\in V_k^l}}
	{\sum_{k'=1}^{K^l}\pr[i,n]{z\leq-v_{0,n}-\sum_{j\neq i}\hat{v}_{j,n}\leq z'\mid\sum_{j\neq i}\hat{v}_{j,n}\in V^l_{k'}}
		\pr[i,n]{\sum_{j\neq i}\hat{v}_{j,n}\in V^l_{k'}}}\label{T3-E15}
\end{align}
where the second line follows from the law of iterated expectations.
We focus on the conditional probabilites in \eqref{T3-E15}.
Recall that $G_n$ is the cumulative distribution function of $v_{0,n}$.
Let $l$ be large enough that $\Delta^l<z'-z$.
Then
\begin{align}
	&\pr[i,n]{z\leq-v_{0,n}-\sum_{j\neq i}\hat{v}_{j,n}\leq z'\mid\sum_{j\neq i}\hat{v}_{j,n}\in V_k^l}\label{T3-E10}\\
	&\leq\pr[i,n]{z+\inf V_k^l\leq-v_{0,n}\leq z'+\sup V_k^l\mid\sum_{j\neq i}\hat{v}_{j,n}\in V_k^l}\notag\\
	&=\pr{z+\inf V_k^l\leq-v_{0,n}\leq z'+\sup V_k^l}\tag{independence of $v_{0,n}$}\\
	&=\pr{-z'-\sup V_k^l\leq v_{0,n}\leq-z-\inf V_k^l}\notag\\
	&=\pr{v_{0,n}\leq-z-\inf V_k^l}-\pr{-z'-\sup V_k^l\leq v_{0,n}}\notag\\
	&=G_n\left(-z-\inf V_k^l\right)-G_n\left(-z'-\sup V_k^l\right)\tag{defn. of $G_n$}\\
	&\leq\frac{1}{2}\exp\left(-\frac{\left|-z-\inf V_k^l\right|}{\beta_n}\right)\left(\frac{z'-z+\Delta^k_l}{\beta_n}+e\left(\frac{z'-z+\Delta^k_l}{\beta_n}\right)^2\right)
	\tag{by Claim \ref{T3-Cl7}}\\
	&\leq\frac{1}{2}\exp\left(-\frac{\left|-z-\inf V_k^l\right|}{\beta_n}\right)\left(\frac{z'-z+\Delta_l}{\beta_n}+e\left(\frac{z'-z+\Delta_l}{\beta_n}\right)^2\right)
	\tag{defn. of $\Delta_l$}
\end{align}
We can also obtain a lower bound as follows:
\begin{align}
	\eqref{T3-E10}&\geq\pr[i,n]{z+\sup V^l_k\leq-v_{0,n}\leq z'+\inf V^l_k\mid\sum_{j\neq i}\hat{v}_{j,n}\in V_k}\notag\\
	&=\pr{z+\sup V^l_k\leq-v_{0,n}\leq z'+\inf V^l_k}\tag{independence of $v_{0,n}$}\\
	&\geq G_n\left(-z-\sup V^l_k\right)-G_n\left(-z'-\inf V^l_k\right)\tag{defn. of $G_n$}\\
	&\geq\frac{1}{2}\exp\left(-\frac{\left|-z-\sup V_k^l\right|}{\beta_n}\right)
	\left(\frac{z'-z-\Delta^k_l}{\beta_n}-e\left(\frac{z'-z-\Delta^k_l}{\beta_n}\right)^2\right)
	\tag{by Claim \ref{T3-Cl7}}\\
	&\geq\frac{1}{2}\exp\left(-\frac{\left|-z-\sup V_k^l\right|}{\beta_n}\right)
	\left(\frac{z'-z-\Delta_l}{\beta_n}-e\left(\frac{z'-z-\Delta_l}{\beta_n}\right)^2\right)
	\tag{defn. of $\Delta_l$ and since $\beta_n\geq e(v_H-v_L)$}
\end{align}
Together with equation \eqref{T3-E15}, these bounds imply that
\begin{align}
	&\frac{\exp\left(-\frac{\left|-z-\sup V_k^l\right|}{\beta_n}\right)
		\left(\frac{z'-z-\Delta_l}{\beta_n}-e\left(\frac{z'-z-\Delta_l}{\beta_n}\right)^2\right)
		\pr[i,n]{\sum_{j\neq i}\hat{v}_{j,n}\in V_k^l}}
	{\sum_{k'=1}^{K_l}\exp\left(-\frac{\left|-z-\inf V_{k'}^l\right|}{\beta_n}\right)\left(\frac{z'-z+\Delta_l}{\beta_n}+e\left(\frac{z'-z+\Delta_l}{\beta_n}\right)^2\right)
		\pr[i,n]{\sum_{j\neq i}\hat{v}_{j,n}\in V^l_{k'}}}\label{T3-E13}\\
	&\leq\eqref{T3-E15}\notag\\
	&\leq
	\frac{\exp\left(-\frac{\left|-z-\inf V_k^l\right|}{\beta_n}\right)\left(\frac{z'-z+\Delta_l}{\beta_n}+e\left(\frac{z'-z+\Delta_l}{\beta_n}\right)^2\right)
		\pr[i,n]{\sum_{j\neq i}\hat{v}_{j,n}\in V_k^l}}
	{\sum_{k'=1}^{K_l}\exp\left(-\frac{\left|-z-\max V_{k'}^l\right|}{\beta_n}\right)
		\left(\frac{z'-z-\Delta_l}{\beta_n}-e\left(\frac{z'-z-\Delta_l}{\beta_n}\right)^2\right)
		\pr[i,n]{\sum_{j\neq i}\hat{v}_{j,n}\in V^l_{k'}}}\notag
\end{align}
To make our notation more compact, define
\[
\eta^+_{n,l}(z'-z)=\frac{\frac{z'-z+\Delta_l}{\beta_n}+e\left(\frac{z'-z+\Delta_l}{\beta_n}\right)^2}{\frac{z'-z-\Delta_l}{\beta_n}-e\left(\frac{z'-z-\Delta_l}{\beta_n}\right)^2}
\quad\text{and}\quad\eta^-_{n,l}(z'-z)=\frac{\frac{z'-z-\Delta_l}{\beta_n}-e\left(\frac{z'-z-\Delta_l}{\beta_n}\right)^2}{\frac{z'-z+\Delta_l}{\beta_n}+e\left(\frac{z'-z+\Delta_l}{\beta_n}\right)^2}
\]
and their respective limits as $l\to\infty$,
\[
\eta^+_n(z'-z)=\frac{\beta_n+e(z'-z)}{\beta_n-e(z'-z)}
\quad\text{and}\quad
\eta^-_n(z'-z)=\frac{\beta_n-e(z'-z)}{\beta_n+e(z'-z)}
\]
Similarly, define
\[
\xi_{i,n,l}(k)=\frac{\exp\left(-\frac{\left|-z-\inf V_k^l\right|}{\beta_n}\right)
	\pr[i,n]{\sum_{j\neq i}\hat{v}_{j,n}\in V_k^l}}
{\sum_{k'=1}^{K_l}\exp\left(-\frac{\left|-z-\max V_{k'}^l\right|}{\beta_n}\right)
	\pr[i,n]{\sum_{j\neq i}\hat{v}_{j,n}\in V^l_{k'}}}
\]
With this notation, inequality \eqref{T3-E13} becomes
\begin{equation}
	\eta^-_{n,l}(z'-z)
	\cdot\xi_{i,n,l}(k)\leq F^{z,z'}(V_k^l)\leq\eta^+_{n,l}(z'-z)\cdot\xi_{i,n,l}(k)\label{T3-E16}
\end{equation}
It follows that
\begin{align}
	&\limsup_{l\to\infty}\left(\sum_{k=1,\ldots,K_l}\alpha_k^l\cdot F^{v_L,v_H}(V^l_k)-\sum_{k=1,\ldots,K_l}\alpha_k^l\cdot F^{z,z'}(V^l_k)\right)\label{T3-E17}\\
	&=\limsup_{l\to\infty}\sum_{k=1,\ldots,K_l}\alpha_k^l\cdot F^{z,z'}(V^l_k)\left(\frac{F^{v_L,v_H}(V^l_k)}{F^{z,z'}(V^l_k)}-1\right)\notag\\
	&\leq\limsup_{l\to\infty}\sum_{k=1,\ldots,K_l}\alpha_k^l\cdot F^{z,z'}(V^l_k)\left(\frac{\eta^+_{n,l}(v_H-v_L)}{\eta^-_{n,l}(z'-z)}-1\right)\tag{inequality \eqref{T3-E16}}\\
	&=\left(\lim_{l\to\infty}\frac{\eta^+_{n,l}(v_H-v_L)}{\eta^-_{n,l}(z'-z)}-1\right)\left(\limsup_{l\to\infty}\sum_{k=1,\ldots,K_l}\alpha_k^l\cdot F^{z,z'}(V^l_k)\right)\tag{property of $\limsup$}\\
	&=\left(\frac{\eta^+_n(v_H-v_L)}{\eta^-_n(z'-z)}-1\right)\left(\limsup_{l\to\infty}\sum_{k=1,\ldots,K_l}\alpha_k^l\cdot F^{z,z'}(V^l_k)\right)\tag{property of $\lim$}\\
	&\leq\left(\frac{\eta^+_n(v_H-v_L)}{\eta^-_n(v_H-v_L)}-1\right)\left(\limsup_{l\to\infty}\sum_{k=1,\ldots,K_l}\alpha_k^l\cdot F^{z,z'}(V^l_k)\right)\tag{since $\eta^-_n(\cdot)$ decreasing}
\end{align}
Following similar reasoning, we find that
\begin{equation}
	\eqref{T3-E17}\geq\left(\frac{\eta^-_n(v_H-v_L)}{\eta^+_n(v_H-v_L)}-1\right)\left(\limsup_{l\to\infty}\sum_{k=1,\ldots,K_l}\alpha_k^l\cdot F^{z,z'}(V^l_k)\right)\label{T3-E19}
\end{equation}
Finally, we can return to the Lebesgue integral defined in equation \eqref{T3-E18}.
Let $h_l$ be the sequence of simple functions that satisfies
\[
\lim_{l\to\infty}\sum_{k=1}^{K_{h_l}}\alpha^{h_l}_k\cdot F^{z,z'}\left(V^{h_l}_k\right)=\sup_{h\in\mathcal{H}}\sum_{k=1}^{K_h}\alpha^h_k\cdot F^{z,z'}
\]
Let $h'_l$ be the sequence of simple functions that satisfies
\[
\sup_{h'\in\mathcal{H}}\sum_{k=1}^{K_{h'}}\alpha^{h'}_k\cdot F^{v_L,v_H}\left(V^{h'}_k\right)=\lim_{l\to\infty}\sum_{k=1}^{K_{h'_l}}\alpha^{h'_l}_k\cdot F^{v_L,v_H}\left(V^{h'_l}_k\right)
\]
Next, observe that
\begin{align}
	&\ex[i,n]{v_i\mid v_L\leq -n\tilde{v}_{i,n}\leq v_H}
	-\ex[i,n]{v_i\mid z\leq -n\tilde{v}_{i,n}\leq z'}\label{T3-E21}\\
	&=\sup_{h\in\mathcal{H}}\sum_{k=1}^{K_h}\alpha^h_k\cdot F^{v_L,v_H}\left(V^h_k\right)-\sup_{h\in\mathcal{H}}\sum_{k=1}^{K_h}\alpha^h_k\cdot F^{z,z'}\left(V^h_k\right)\notag\\
	&\white{=}-\sup_{h'\in\mathcal{H}}\sum_{k=1}^{K_{h'}}\alpha^{h'}_k\cdot F^{v_L,v_H}\left(V^{h'}_k\right)+\sup_{h'\in\mathcal{H}}\sum_{k=1}^{K_{h'}}\alpha^{h'}_k\cdot F^{z,z'}\left(V^{h'}_k\right)\tag{equation \eqref{T3-E18}}\\
	&\geq\limsup_{l\to\infty}\sum_{k=1}^{K_{h_l}}\alpha^{h_l}_k\cdot F^{v_L,v_H}\left(V^{h_l}_k\right)-\lim_{l\to\infty}\sum_{k=1}^{K_{h_l}}\alpha^{h_l}_k\cdot F^{z,z'}\left(V^{h_l}_k\right)\notag\\
	&\white{=}-\lim_{l\to\infty}\sum_{k=1}^{K_{h'_l}}\alpha^{h'_l}_k\cdot F^{v_L,v_H}\left(V^{h'_l}_k\right)+\limsup_{l\to\infty}\sum_{k=1}^{K_{h'_l}}\alpha^{h'_l}_k\cdot F^{z,z'}\left(V^{h'_l}_k\right)
	\tag{defn. of $\sup$ and $h_l,h'_l$}\\
	&\geq\left(\frac{\eta^-_n(v_H-v_L)}{\eta^+_n(v_H-v_L)}-\frac{\eta^+_n(v_H-v_L)}{\eta^-_n(v_H-v_L)}\right)(v_H-v_L)
	\tag{inequalities \eqref{T3-E17} and \eqref{T3-E19}}
\end{align}
Following similar reasoning, we find that
\begin{align}
	\eqref{T3-E21}\leq\left(\frac{\eta^+_n(v_H-v_L)}{\eta^-_n(v_H-v_L)}-\frac{\eta^-_n(v_H-v_L)}{\eta^+_n(v_H-v_L)}\right)(v_H-v_L)\label{T3-E22}
\end{align}
Combine inequalities \eqref{T3-E21} and \eqref{T3-E22}, and then simplify, to obtain
\begin{align*}
&\left|\ex[i,n]{v_i\mid z\leq-n\tilde{v}_{i,n}\leq v_H}-\ex[i,n]{v_i\mid z\leq-n\tilde{v}_{i,n}\leq z'}\right|\\
&\leq\left(\frac{\beta_n+e(v_H-v_L)}{\beta_n-e(v_H-v_L)}\right)^2-\left(\frac{\beta_n-e(v_H-v_L)}{\beta_n+e(v_H-v_L)}\right)^2
\end{align*}

	\subsection{Proof of Claim \ref{T3-Cl9}}

Let $z,z'$ be constants where $v_L\leq z<z'\leq v_H$.
We begin by bounding the following:
\begin{align}
	&\pr[i,n]{D=0\mid z\leq\tilde{p}_{i,n}\leq z'}\label{T3-E32}\\
	&=\frac{\pr[i,n]{z\leq\tilde{p}_{i,n}\leq z'\mid D=0}\cdot\pr[i,n]{D=0}}{\pr[i,n]{z\leq\tilde{p}_{i,n}\leq z'}}
	\tag{Bayes' rule}\\
	&=\frac{\pr[i,n]{z\leq\tilde{p}_{i,n}\leq z'\mid D=0}\pr{D=0}}{\pr[i,n]{z\leq\tilde{p}_{i,n}\leq z'\mid D=i}\pr{D=i}+\pr[i,n]{z\leq\tilde{p}_{i,n}\leq z'\mid D=0}\pr{D=0}}
	\tag{LIE}\\
	&=\frac{\pr[i,n]{z\leq -n\tilde{v}_{i,n}\leq z'\mid D=0}\cdot(1-\delta_n)}{\pr[i,n]{z\leq p\leq z'\mid D=i}\cdot(\delta_n/n)+\pr[i,n]{z\leq -n\tilde{v}_{i,n}\leq z'\mid D=0}\cdot(1-\delta_n)}
	\tag{defn. of $D$ and $\tilde{p}_{i,n}$}\\
	&=\frac{\pr[i,n]{z\leq -n\tilde{v}_{i,n}\leq z'}\cdot(1-\delta_n)}{\pr{z\leq p\leq z'}\cdot(\delta_n/n)+\pr[i,n]{z\leq -n\tilde{v}_{i,n}\leq z'}\cdot(1-\delta_n)}
	\tag{since $D\perp (p,\tilde{v}_{i,n})$}\\
	&=\frac{\ex[i,n]{\pr[i,n]{z\leq -v_{0,n}-\sum_{i\neq j}\hat{v}_{j,n}\leq z',\sum_{j\neq i}\hat{v}_{j,n}}}\cdot(1-\delta_n)}{\pr{z\leq p\leq z'}\cdot(\delta_n/n)+\ex[i,n]{\pr[i,n]{z\leq -v_{0,n}-\sum_{i\neq j}\hat{v}_{j,n}\leq z',\sum_{j\neq i}\hat{v}_{j,n}}}\cdot(1-\delta_n)}\tag{LIE and defn. of $\tilde{v}_{i,n}$}\\
	&=\frac{\ex[i,n]{G_n\left(-z-\sum_{j\neq i}\hat{v}_{j,n}\right)-G_n\left(-z'-\sum_{j\neq i}\hat{v}_{j,n}\right)}\cdot(1-\delta_n)}{\pr{z\leq p\leq z'}\cdot(\delta_n/n)+\ex[i,n]{G_n\left(-z-\sum_{j\neq i}\hat{v}_{j,n}\right)-G_n\left(-z'-\sum_{j\neq i}\hat{v}_{j,n}\right)}\cdot(1-\delta_n)}\tag{defn. of $G_n$}
\end{align}
The upper bound is
\begin{align}
	\eqref{T3-E32}&\leq\frac{
		\frac{1}{2}\ex[i,n]{e^{-\frac{\left|-z-\sum_{j\neq i}\hat{v}_{j,n}\right|}{\beta_n}}}\left(\frac{z'-z}{\beta_n}+e\left(\frac{z'-z}{\beta_n}\right)^2\right)
		\cdot(1-\delta_n)}{\pr{z\leq p\leq z'}\cdot(\delta_n/n)+
		\frac{1}{2}\ex[i,n]{e^{-\frac{\left|-z-\sum_{j\neq i}\hat{v}_{j,n}\right|}{\beta_n}}}\left(\frac{z'-z}{\beta_n}-e\left(\frac{z'-z}{\beta_n}\right)^2\right)
		\cdot(1-\delta_n)}\tag{Claim \ref{T3-Cl7}}\\
	&\leq\eta^+_n(z'-z)\tag{since $\pr{\cdot}\geq 0$ and defn. of $\eta^+_n$}\\
	&\leq\eta^+_n(v_H-v_L)\tag{since $\eta^+_n$ increasing}
\end{align}
The lower bound is
\begin{align}
	\eqref{T3-E32}
	&\geq\frac{
		\frac{1}{2}\ex[i,n]{e^{-\frac{\left|-z-\sum_{j\neq i}\hat{v}_{j,n}\right|}{\beta_n}}}\left(\frac{z'-z}{\beta_n}-e\left(\frac{z'-z}{\beta_n}\right)^2\right)
		\cdot(1-\delta_n)}{\pr{z\leq p\leq z'}\cdot(\delta_n/n)+
		\frac{1}{2}\ex[i,n]{e^{-\frac{\left|-z-\sum_{j\neq i}\hat{v}_{j,n}\right|}{\beta_n}}}\left(\frac{z'-z}{\beta_n}+e\left(\frac{z'-z}{\beta_n}\right)^2\right)
		\cdot(1-\delta_n)}\tag{Claim \ref{T3-Cl7}}\\
	&=\frac{
		\frac{1}{2}\ex[i,n]{e^{-\frac{\left|-z-\sum_{j\neq i}\hat{v}_{j,n}\right|}{\beta_n}}}\left(\frac{z'-z}{\beta_n}-e\left(\frac{z'-z}{\beta_n}\right)^2\right)
		\cdot(1-\delta_n)}{
		\frac{z'-z}{v_H-v_L}\cdot\frac{\delta_n}{n}+\frac{1}{2}\ex[i,n]{e^{-\frac{\left|-z-\sum_{j\neq i}\hat{v}_{j,n}\right|}{\beta_n}}}\left(\frac{z'-z}{\beta_n}+e\left(\frac{z'-z}{\beta_n}\right)^2\right)
		\cdot(1-\delta_n)}\tag{dist. of $p$}\\
	&=\frac{
		\beta_n-e(z'-z)}{
		\frac{2}{1-\delta_n}\ex[i,n]{e^{-\frac{\left|-z-\sum_{j\neq i}\hat{v}_{j,n}\right|}{\beta_n}}}
		\cdot\frac{\beta_n^2}{v_H-v_L}\cdot\frac{\delta_n}{n}
		+\beta_n+e(z'-z)
		}\notag\\
	&\geq\eta^-_n(z'-z)
	\tag{since $\frac{a}{w+b}\geq\frac{a}{b}-w\cdot\frac{a}{b^2}$}\\
	&\white{=}-\frac{2}{1-\delta_n}\ex[i,n]{e^{-\frac{\left|-z-\sum_{j\neq i}\hat{v}_{j,n}\right|}{\beta_n}}}
	\cdot\frac{\beta_n^2}{v_H-v_L}\cdot\frac{\delta_n}{n}\cdot\frac{
		\beta_n-e(z'-z)}{
		\left(\beta_n+e(z'-z)\right)^2
	}\notag\\
	&=\eta^-_n(z'-z)-O\left(\frac{\beta_n\delta_n}{n}\right)\notag\\
	&\geq=\eta^-_n(v_H-v_L)-O\left(\frac{\beta_n\delta_n}{n}\right)\tag{since $\eta^+_n$ decreasing}
\end{align}
For convenience, let
$
\zeta_{i,n}(z,z')=\pr[i,n]{D=0\mid z\leq\tilde{p}_{i,n}\leq z'}
$.
Observe that
\begin{align}
	&\ex[i,n]{v_i\mid z\leq\tilde{p}_{i,n}\leq z'}\label{T3-E34}\\
	&=\zeta_{i,n}(z,z')\ex[i,n]{v_i\mid z\leq-n\tilde{v}_{i,n}\leq z', D=0}+(1-\zeta_{i,n}(z,z'))\ex[i,n]{v_i\mid z\leq p\leq z', D=i}
	\tag{LIE and defn. of $\tilde{p}_{i,n}$}\\
	&=\zeta_{i,n}(z,z')\ex[i,n]{v_i\mid z\leq-n\tilde{v}_{i,n}\leq z'}+(1-\zeta_{i,n}(z,z'))\ex[i,n]{v_i}
	\tag{$(p,D_i)\perp (v_i,\tilde{v}_{i,n})$}
\end{align}
We want to bound
\begin{align}
	&\left|\ex[i,n]{v_i\mid z\leq\tilde{p}_{i,n}\leq z'}-\ex[i,n]{v_i\mid Q_{i,n}}\right|\label{T3-E33}\\
	&=\left|\ex[i,n]{v_i\mid z\leq\tilde{p}_{i,n}\leq z'}-\ex[i,n]{v_i\mid v_L\leq\tilde{p}_{i,n}\leq v_H}\right|\tag{defn. of $Q_{i,n}$}\\
	&=\big|\zeta_{i,n}(z,z')\ex[i,n]{v_i\mid z\leq-n\tilde{v}_{i,n}\leq z'}+(1-\zeta_{i,n}(z,z'))\ex[i,n]{v_i}\notag\\
	&\white{=}
	-\zeta_{i,n}(v_L,v_H)\ex[i,n]{v_i\mid v_L\leq-n\tilde{v}_{i,n}\leq v_H}-(1-\zeta_{i,n}(v_L,v_H))\ex[i,n]{v_i}
	\big|\tag{eq. \eqref{T3-E34}}\\
	&\leq\left|\zeta_{i,n}(v_L,v_H)\right|\cdot\left|\ex[i,n]{v_i\mid z\leq-n\tilde{v}_{i,n}\leq z'}-\ex[i,n]{v_i\mid v_L\leq-n\tilde{v}_{i,n}\leq v_H}\right|
	\notag\\
	&\white{=}+\left|\zeta_{i,n}(z,z')-\zeta_{i,n}(v_L,v_H)\right|\cdot\left|\ex[i,n]{v_i\mid z\leq-n\tilde{v}_{i,n}\leq z'}\right|
	\notag\\
	&\white{=}
	+\left|\zeta_{i,n}(v_L,v_H)-\zeta_{i,n}(z,z')\right|\cdot\left|\ex[i,n]{v_i}\right|
	\tag{properties of $|\cdot|$}\\
	&\leq\left|\zeta_{i,n}(v_L,v_H)\right|\cdot\psi_n+2\left|\zeta_{i,n}(z,z')-\zeta_{i,n}(v_L,v_H)\right|\cdot(v_H-v_L)\tag{Claim \ref{T3-Cl6}}\\
	&\leq\psi_n+2\left(\eta^+_n(v_H-v_L)-\eta^-_n(v_H-v_L)+O\left(\frac{\beta_n\delta_n}{n}\right)\right)\cdot(v_H-v_L)\tag{bounds on \eqref{T3-E32}}\\
	&\leq\psi_n+O\left(\frac{1}{\beta_n}+\frac{\beta_n\delta_n}{n}\right)\tag{since $\frac{w+1}{w-1}-\frac{w-1}{w+1}=O\left(1/w\right)$}
\end{align}
	
	\subsection{Proof of Claim \ref{T3-Cl8}}

Agent $i$'s report $\hat{v}_{i,n}$ influences her allocation and the part of her transfers. attributable to the VCG mechanism (if $D=0$) or the BDM mechanism (if $D=i$).

I claim that her payoffs from her report can be represented as purchasing alternative $x=1$ at price $\tilde{p}_{i,n}$.
This is clearly true if agent $i$ is dictator ($D=i$) and participates in the BDM mechanism.
Suppose instead that there is no dictator ($D=0$) and agent $i$ participates in the VCG mechanism.
There are two cases.
\begin{enumerate}
	\item If $\tilde{v}_{i,n}<0$, then agent $i$ chooses between (i) alternative $x=1$ and paying $-n\tilde{v}_{i,n}$ and (ii) alternative $x=0$ and paying nothing.
	\item If $\tilde{v}_{i,n}\geq 0$, then she chooses between (i) alternative $x=1$ and paying nothing and (ii) alternative $x=0$ and paying $n\tilde{v}_{i,n}$.
	This is strategically equivalent to choosing between (i) $x=1$ and paying $n\tilde{v}_{i,n}$ and (ii) $x=0$ and paying nothing.
\end{enumerate}

Next, we show that the agent's optimal report $\hat{v}_{i,n}$ is close to her expected value $\ex[i,n]{v_i}$.
Fix a constant $t>0$.
For the rest of this proof, we condition agent $i$ having a signal realization such that
\[
\left|\ex[i,n]{v_i\mid v_L\leq-n\tilde{v}_{i,n}\leq v_H}-\ex[i,n]{v_i\mid z\leq-n\tilde{v}_{i,n}\leq z'}\right|\leq t+\psi_n+O\left(\frac{1}{\beta_n}+\frac{\beta_n\delta_n}{n}\right)
\]
It follows from the triangle inequality, Claim \ref{T3-Cl5}, and Claim \ref{T3-Cl9} that this holds with high probability.
Let $z<\ex[i,n]{v_i}-2t'$ where
$
t'=t\psi_n+O\left(1/\beta_n+\beta_n\delta_n/n\right)
$.
Agent $i$'s expected payoff from reporting $\ex[i,n]{v_i}-2t'$ is
\begin{align}
	&\ex[i,n]{\textbf{1}\left(v_L\leq\tilde{p}_{i,n}\leq\ex[i,n]{v_i}-2t'\right)\cdot\left(v_i-\tilde{p}_{i,n}\right)}\notag\\
	&=\ex[i,n]{\textbf{1}\left(z<\tilde{p}_{i,n}\leq\ex[i,n]{v_i}-2t'\right)\cdot\left(v_i-\tilde{p}_{i,n}\right)}+\ex[i,n]{\textbf{1}\left(v_L\leq\tilde{p}_{i,n}\leq z\right)\cdot\left(v_i-\tilde{p}_{i,n}\right)}\notag
\end{align}
This is greater than her expected payoff from reporting $z$ if the following is positive:
\begin{align}
	&\ex[i,n]{\textbf{1}\left(z<\tilde{p}_{i,n}\leq\ex[i,n]{v_i}-2t'\right)\cdot\left(v_i-\tilde{p}_{i,n}\right)}\notag\\
	&=\pr[i,n]{z<\tilde{p}_{i,n}\leq\ex[i,n]{v_i}-2t'}\cdot\ex[i,n]{v_i-\tilde{p}_{i,n}\mid z<\tilde{p}_{i,n}\leq\ex[i,n]{v_i}-2t'}\tag{LIE}\\
	&\geq\pr[i,n]{z<\tilde{p}_{i,n}\leq\ex[i,n]{v_i}-2t'}\cdot\left(\ex[i,n]{v_i\mid z<\tilde{p}_{i,n}\leq\ex[i,n]{v_i}-2t'}-\ex[i,n]{v_i}+2t'\right)\notag\\
	&\geq\pr[i,n]{z<\tilde{p}_{i,n}\leq\ex[i,n]{v_i}-2t'}\cdot\left(\ex[i,n]{v_i}-t'-\ex[i,n]{v_i}+2t'\right)\label{T3-E35}\\
	&\geq\frac{\delta_n}{n}\cdot\frac{\ex[i,n]{v_i}-2t'-z}{v_H-v_L}\cdot t'\tag{defn. of $\tilde{p}_{i,n},D$ and since $p\sim\textsc{Uniform}[v_L,v_H]$}\\
	&>0\tag{defn. of $z$}
\end{align}
Line \eqref{T3-E35} follows from the definition of $t'$, provided that
\[
\ex[i,n]{v_i\mid z\leq\tilde{p}_{i,n}\leq\ex[i,n]{v_i}-2t'}=\ex[i,n]{v_i\mid z<\tilde{p}_{i,n}\leq\ex[i,n]{v_i}-2t'}
\]
This follows from the fact that, conditional on any value $D\in\{0,i\}$ and $\sum_{j\neq i}\hat{v}_{j,n}$, $\tilde{p}_{i,n}$ is continuously-distributed and therefore puts zero probability on any given point $z$.

Altogether, we have shown that the agent $i$'s optimal report $\hat{v}_{i,n}\geq\ex[i,n]{v_i}-2t'$.
To complete the proof, we use a similar argument to show that $\hat{v}_{i,n}\leq\ex[i,n]{v_i}+2t'$.

	\printbibliography[segment=1]
	
\end{document}